\documentclass[
twocolumn,
superscriptaddress,
pra,
10pt
]{revtex4-2}

\usepackage{dcolumn}
\usepackage{booktabs}

\usepackage{amsfonts}
\usepackage{amssymb}
\usepackage{mathrsfs}
\usepackage{mathtools}
\usepackage{amsmath}
\usepackage{amsthm}
\usepackage{graphicx}
\usepackage{color} 
\usepackage{array}
\usepackage[makeroom]{cancel}
\usepackage{changes}
\usepackage{braket}
\usepackage{bbm}
\usepackage{bm}
\usepackage{multirow}
\usepackage{cellspace} 
\usepackage{appendix}
\usepackage{algorithm}
\usepackage{algpseudocode}
\usepackage{comment}
\usepackage[export]{adjustbox}
\usepackage{xcolor}
\usepackage{overpic}

\usepackage{tikz}
\usetikzlibrary{shapes,arrows,positioning,fit}
\usepackage{quantikz}
\usetikzlibrary{quantikz2}
\usepackage[colorlinks=true, allcolors=blue]{hyperref}
\usepackage{orcidlink}
\usepackage[capitalize]{cleveref}

\crefname{section}{Sec.}{Secs.}
\crefname{thm}{Theorem}{Theorems}
\crefname{dfn}{Definition}{Definitions}
\crefname{rmk}{Remark}{Remarks}
\crefname{lem}{Lemma}{Lemmas}
\crefname{cor}{Corollary}{Corollaries}
\crefname{prop}{Proposition}{Propositions}
\theoremstyle{plain}  
\newtheorem{prop}{Proposition}

\newtheorem{lem}{Lemma}

\theoremstyle{remark}
\definecolor{sprintRed}{HTML}{D7333B}
\definecolor{sprintGold}{HTML}{E2A300}
\definecolor{sprintGreen}{HTML}{00973E}
\definecolor{sprintBlue}{HTML}{4D53C8}
\definecolor{sprintGrey}{HTML}{555555}
\definecolor{sprintMidGrey}{HTML}{8D8D8D}   
\definecolor{sprintLightGrey}{HTML}{C5C5C5}

\begin{document}

\title{Theory and practice of Trotter product formulas for quantum chemistry}

\author{Pablo A. M. Casares~\orcidlink{0000-0001-5500-9115}}
\affiliation{Xanadu, Toronto, ON, M5G 2C8, Canada}

\author{William Maxwell}
\affiliation{Xanadu, Toronto, ON, M5G 2C8, Canada}

\author{Danial Motlagh~\orcidlink{0009-0003-7655-4341}}
\affiliation{Xanadu, Toronto, ON, M5G 2C8, Canada}

\author{Hitarth Choubisa}
\affiliation{Xanadu, Toronto, ON, M5G 2C8, Canada}

\author{Zy Niu~\orcidlink{0000-0002-4025-5994}}
\affiliation{Xanadu, Toronto, ON, M5G 2C8, Canada}
\author{Ignacio Loaiza~\orcidlink{0000-0001-9630-855X}}
\affiliation{Xanadu, Toronto, ON, M5G 2C8, Canada}

\author{Jonathan E. Mueller~\orcidlink{0000-0001-8811-8799}}
\affiliation{Volkswagen AG, Berliner Ring 2, 38440 Wolfsburg, Germany} 

\author{Arne-Christian Voigt}
\affiliation{Volkswagen AG, Berliner Ring 2, 38440 Wolfsburg, Germany} 

\author{Juan Miguel Arrazola~\orcidlink{0000-0002-0619-9650}}
\affiliation{Xanadu, Toronto, ON, M5G 2C8, Canada}

\author{Stepan Fomichev~\orcidlink{0000-0002-1622-9382}}
\affiliation{Xanadu, Toronto, ON, M5G 2C8, Canada}

\begin{abstract}
  %Hamiltonian simulation is the core subroutine for quantum algorithms in chemistry and materials science. 
  Trotter product formulas are a fundamental class of methods for Hamiltonian simulation, particularly attractive due to their low qubit requirements. However, they are often overlooked for use with fault-tolerant quantum algorithms, because of their perceived higher gate counts and the difficulty of estimating Trotter error.
  Here, we introduce Symmetry-Protected Randomized near-Integrable Trotter (SPRINT) formulas, a framework for building optimized product formulas for electronic structure Hamiltonians widely used in quantum chemistry. 
  SPRINT integrates a generalization of classical near-integrability, randomization, symmetry protection, use of QROM, and other techniques into a thoroughly optimized methodology for Hamiltonian simulation. When applied to concrete simulation tasks, we find SPRINT leads to substantial reduction in gate count compared to previous approaches. Alongside SPRINT, we introduce and analyze a Generalized Rank Decomposition (GRADE) of electronic Hamiltonians that generalizes previous factorization methods. We apply these techniques to the task of simulating the X-ray absorption spectrum of Li$_4$Mn$_2$O, a candidate battery cathode material, leveraging recent advances in tight Trotter error estimation to carefully identify the best version of SPRINT for this problem. Using a Trotter error estimation tool developed in the PennyLane software platform, we show that SPRINT reduces the Toffoli gate cost by a factor of $4.5$ relative to the previous state of the art for this problem, with a gate cost only $\times 2.5$ higher than qubitization, while requiring a dramatic $\times 5.5$ fewer logical qubits. These results establish well-designed Trotter product formulas as an attractive Hamiltonian simulation method for industrially relevant problems in chemistry and materials science.
\end{abstract}

 \maketitle

\section{Introduction}

Hamiltonian simulation is the core subroutine of the majority of quantum algorithms with applications to chemistry and materials science. The two main families of simulation methods -- Trotter product formulas and qubitization -- have complementary strengths. Trotter formulas require few auxiliary qubits but exhibit polynomial scaling with precision~\cite{childs2021theory}. Qubitization achieves polylogarithmic precision scaling, but at the cost of significant qubit overhead \cite{low2019hamiltonian,low2024trading}. Trotter methods offer a significant qubit advantage that makes them well-suited for early generations of fault-tolerant hardware. Yet they have remained comparatively underexplored in the context of concrete resource estimation for industrially relevant problems, with most studies still relying on the basic Trotter-Suzuki formulas \cite{childs2021theory} or qubitization. This is in large part because the analytical bounds available for Trotter error are overly loose, artificially inflating constant-factor resource estimates and making product formulas appear unattractive in practice \cite{childs2021theory,casas2026error}.

In this manuscript, we introduce Symmetry-Protected Randomized near-Integrable Trotter formulas (SPRINT), a framework for crafting state-of-the-art, problem-specific Trotter formulas. This framework combines several leading techniques for product formulas: a novel factorization scheme, a generalization of the near-integrability concept from the classical literature, randomization, processing, and symmetry protection. The techniques are combined in a mutually complementary way that addresses their respective weaknesses, while amplifying their strengths. The general structure of a product formula built with this framework is depicted in~\cref{fig:SPRINT_product_formula}. Leveraging recent advances in tight Trotter error estimation, we quantitatively identify the optimal combination of SPRINT techniques for a given problem. This delivers low cost at acceptable error, yielding resource counts that appear attractive in practice. Unlike in prior works, for certain problems we find product formulas are actually competitive with qubitized approaches in terms of Toffoli  counts, while maintaining Trotter's significant qubit advantage.

The starting point of SPRINT is the factorization of the electronic Hamiltonian into fast-forwardable fragments. Here we introduce the Generalized Rank Decomposition (GRADE,~\cref{ssec:GRADE}), which smoothly interpolates between compressed double factorization (CDF) and isometric tensor hypercontraction (THC), the leading rank factorization methods in Trotter simulation of electronic Hamiltonians \cite{cohn2021quantum,luo2025efficient}. GRADE splits the Hamiltonian into fragments whose norms span several orders of magnitude, creating a hierarchical structure that can be exploited by the product formula. While GRADE can reduce per-step gate counts, we find the use of auxiliary orbitals increases the simulation error, making it less cost efficient than standard compressed double factorization. Nevertheless, the norm hierarchy that GRADE reveals is the key enabler of the subsequent techniques.

\begin{figure*}
\centering
\scalebox{0.85}{
$\begin{quantikz}
  \lstick{$\ket{0}_{\text{aux}}$} & \gate[3, style={dashed, draw=sprintGrey, fill=sprintGrey!20}]{e^{+P}} & \gate[style={draw=sprintGreen, fill=sprintGreen!20}]{e^{+i (N_e - N_{\text{ph}})\phi}}
  \gategroup[3,steps=5,style={dotted, inner ysep=22pt}, label style={above=-6pt}]{Kernel $e^{iK}$}
  & \gate[3, style={draw=sprintBlue, fill=sprintBlue!20}]{\parbox[c][1.5cm]{2.2cm}{\centering $U_{1, B}(\tau)$ \\[4pt] 1st order \\[4pt] (Randomized)}}
  \gategroup[3,steps=3,style={dotted, inner sep=6pt}]{$\tilde{U}_{2,B+C}(\tau)$}
  & \gate[3, style={draw=sprintGold, fill=sprintGold!20}]{\parbox[c][1.5cm]{2.2cm}{\centering qDRIFT \\ RTE \\ BLISS \\ STAR}}
  & \gate[3, style={draw=sprintBlue, fill=sprintBlue!20}]{\parbox[c][1.5cm]{2.2cm}{\centering $U_{1, B}^\dagger(-\tau)$ \\[4pt] 1st order reversed\\[4pt] (Randomized)}} & \gate[style={draw=sprintGreen, fill=sprintGreen!20}]{e^{-i (N_e - N_{\text{ph}})\phi}} & \gate[3, style={dashed, draw=sprintGrey, fill=sprintGrey!20}]{e^{-P}} & \\
  \lstick[2]{$\ket{\psi}_{\text{ph}}$} & & \gate[2, style={draw=sprintRed, fill=sprintRed!20}]{\parbox{2cm}{\centering $U_{4, A}^{1/2}(\tau)$ \\[4pt] 4th order\\[4pt] (Randomized)}} & & & & \gate[2, style={draw=sprintRed, fill=sprintRed!20}]{\parbox{2cm}{\centering $U^{1/2}_{4, A}(\tau)$ \\[4pt] 4th order\\[4pt] (Randomized)}} & & \\
& & & & & & & &
\end{quantikz}$
}
  \caption{General structure of a product formula under the SPRINT framework. Here $U_{n,X}(\tau)$ represents a $n$-th order Trotter product formula step of length $\tau$ on the group of fragments $X$. The term ``randomized'' refers to the ordering of the fragments in the product formula. We assume the Hamiltonian is decomposed in three \textit{groups} of fragments $H = H_A + H_B + H_C$ in accordance with their norms (\cref{ssec:GRADE}). 
  %The two first fragments, one coming from the one-electron and another from the two-electron integrals, form the group $H_A$. The rest of the two body fragments are assigned to $H_B$, and we reserve group $H_C$ for the Pauli strings. 
  To take advantage of the different norms, we use a near-integrable product formula (\cref{ssec:Near_integrability}): in this example, we use a fourth order formula for $A$ (\textcolor{sprintRed}{red}) and a second order for $B$ (\textcolor{sprintBlue}{blue}). In the middle of the first order formula we propose simulating the smallest fragments (the Pauli string group $C$) 
  %or even the uncomputation of the Trotter error 
  with even cheaper methods like qDRIFT, RTE, or STAR, optionally enhanced with BLISS to decrease their cost by reducing the 1-norm
 (\textcolor{sprintGold}{yellow}). 
  %These methods may also be used to perform error mitigation. 
  We may also implement symmetry protection (\cref{ssec:Symmetry_protection}) to decrease the leakage error from fragments with $M_\ell > N$ (\textcolor{sprintGreen}{green}). 
  %$N_e$ and $N_{\text{ph}}$ represent the number of electrons in the $M_\ell$ or $N$ ``physical" orbitals, respectively.
  Finally, a processor $P$ (\textcolor{sprintGrey}{grey}) might be used to further decrease the simulation error at an additive constant overhead (\cref{ssec:Processing}). 
  }
  \label{fig:SPRINT_product_formula}
\end{figure*}
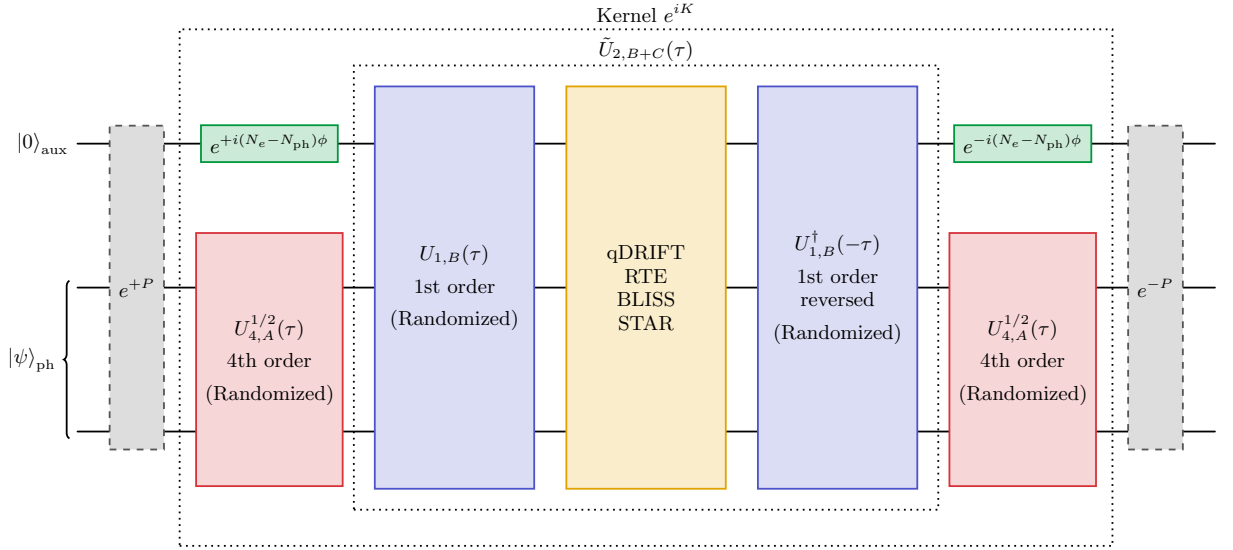

This hierarchical norm structure of different fragments can be directly exploited by a primitive we call a \textit{near-integrable} product formula (\cref{ssec:Near_integrability}). This is a concept we generalize from the classical geometric integrator literature for use in quantum computing that exploits energy-scale separation of different fragments. Rather than treating all fragments on equal footing, these formulas apply a high-order formula only to the dominant fragments. The remaining fragments have small norm individually but are numerous, making them collectively costlier; these are handled with a low-order method. These formulas achieve accuracy comparable to a uniformly high-order scheme at a fraction of the cost -- up to $\times 2$ Toffoli savings for the second-order near-integrable formula, or $\times 5$ for the fourth-order formula. A complementary technique called \textit{processing} (\cref{ssec:Processing}) further reduces the simulation error by conjugating the entire time-evolution circuit with a carefully optimized unitary. Processing cancels commutator error terms at constant additive cost \cite{blanes2000processing,blanes2004numerical,blanes2006composition}: because the processor is applied only once at the beginning and end of the full evolution, its overhead is negligible for all but the shortest simulation times.

When the factorization introduces auxiliary orbitals (e.g. in isometric THC), the wavefunction can leak into unphysical subspaces, producing an error that is uncontrolled by the Trotter step size $\tau$. Symmetry protection (\cref{ssec:Symmetry_protection}) addresses this issue by conjugating each Trotter step with phases on the auxiliary qubits, suppressing the leakage error from $O(1)$ to $O(\tau^2)$. 

\begin{figure*}[t]
  \centering
  \includegraphics[width=0.35\textwidth]{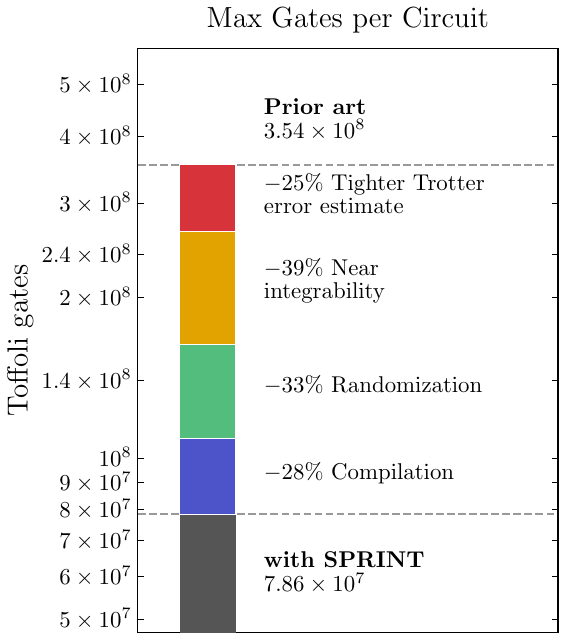}
  \includegraphics[width=0.35\textwidth]{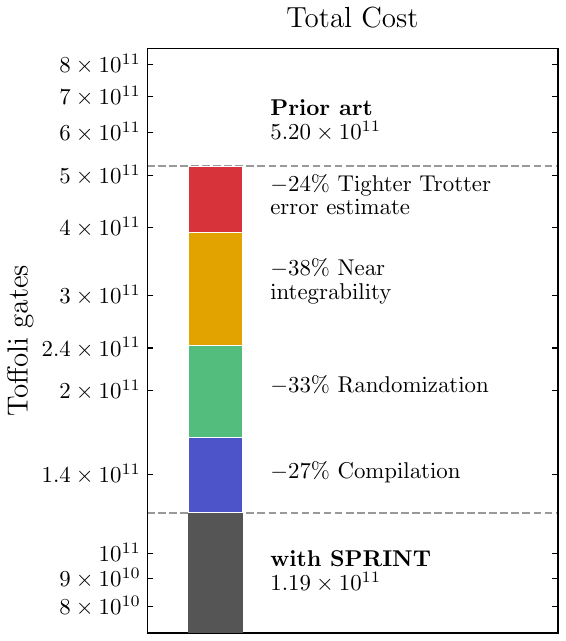}
  \caption{Cost reduction in the problem of X-ray absorption for Li$_4$Mn$_2$O in the active space of $N = 18$ orbitals, with respect to the previous state-of-the-art results in Ref. \cite{fomichev2025fast}. The savings factors are larger for maximum circuit depth because state preparation contributes a constant additive overhead that is unaffected by these optimizations. Overall, the Toffoli gate count is reduced by around $\times 4.5$ times. The specific sources of improvement are explained in detail in \cref{ssec:trotter_error_results}.}
  \label{fig:SPRINT_Towers}
\end{figure*}

The gate cost of the formula may be further reduced by randomization of the order of the fragments in-between Trotter steps (\cref{ssec:randomization}) \cite{childs2019faster}. In addition, random compilation such as qDRIFT \cite{campbell2019random}, randomized Taylor expansion (RTE) \cite{wan2022randomized}, or efficient partially fault tolerant compilation regimes like Space-Time efficient Analog Rotation (STAR) \cite{akahoshi2024partially,akahoshi2025compilation,chung2026partially,kanasugi2026enabling} can also be used to implement any terms remaining after the Hamiltonian factorization. Since the cost of those methods depends on the one-norm of the remainder being implemented, we may exploit symmetry shifts \cite{loaiza2023block} to decrease that 1-norm and thus the cost of the methods. Finally, QROM use (\cref{ssec:QROM}) can allow us to batch the rotation gates within each fragment using quantum read-only memory, also reducing the Toffoli footprint \cite{low2024trading}.

In~\cref{sec:Applications} we apply SPRINT and GRADE to the task of calculating the X-ray absorption spectrum of a CAS(22e, 18o) Li$_4$Mn$_2$O cluster, a system of interest in studies of Li-excess battery materials \cite{fomichev2024simulating,fomichev2025fast}. To estimate the Trotter error and decide on the optimal combination of SPRINT components, we use a new error estimation tool~\cite{Xanadu_Trotter_error_paper}. On this system the main cost saving drivers are near-integrability, randomization, and QROM-based compilation, yielding saving factors of approximately $\times 1.6$, $\times 1.5$, and $\times 1.4$, respectively, relative to prior state of the art \cite{fomichev2025fast}, see~\cref{fig:SPRINT_Towers}. The overall SPRINT pipeline delivers a $\times 4.5$ Toffoli reduction over Ref. \cite{fomichev2025fast}, and requires only $\times 2.5$ more Toffoli gates than qubitization, while using $\times 5.5$ fewer logical qubits. These results establish Trotter product formulas as a competitive Hamiltonian simulation method for industrially relevant problems in chemistry and materials science.

The rest of this manuscript is organized as follows. In~\cref{sec:Factorization} we present the GRADE factorization method that generalizes previous techniques from the literature. In~\cref{sec:Product_formulas} after presenting the key mathematical tools we use in \cref{ssec:Product_formulas}, we introduce each of the elements of the SPRINT framework shown in \cref{fig:SPRINT_product_formula}, namely near-integrability (\cref{ssec:Near_integrability}), processing (\cref{ssec:Processing}), symmetry protection (\cref{ssec:Symmetry_protection}), randomization (\cref{ssec:randomization}) and QROM use (\cref{ssec:QROM}). In~\cref{sec:building_pf} we provide a general methodology to decide which of these methods to use for a specific problem. Finally, in~\cref{sec:Applications} we evaluate the performance of the new methods on the case of the Li-excess cluster, including a detailed comparison against state-of-the-art qubitization methods.

\section{Factorization of the Hamiltonian\label{sec:Factorization}}

Implementing quantum simulation with Trotter product formulas requires making two key decisions: how to split the Hamiltonian into fast-forwardable fragments, and how to compose the resulting fragment unitaries into a product formula. A fragment is called fast-forwardable or exactly integrable if we can implement the unitary $e^{-iH_j t}$ for any time $t$ using a circuit whose cost is independent of $t$ or the simulation precision, up to polylogarithmic cost. The simplest fast-forwardable fragments are individual Pauli strings, natively implementable as qubit rotations. More structured fragments, however, can unlock significant efficiency gains.

\subsection{Generalized rank decomposition\label{ssec:GRADE}}

Consider the electronic Hamiltonian, given by
\begin{multline}\label{eq:el-ham}
    H = E + \sum_{p,q = 1}^N \sum_{\gamma \in \{\uparrow,\downarrow\}} (p|\kappa|q) a_{p\gamma}^\dagger a_{q\gamma}+ \\ \frac{1}{2}\sum_{p,q,r,s=1}^N\sum_{\gamma,\tau \in \{\uparrow,\downarrow\}}(pq|rs) a_{p\gamma}^\dagger a_{q\gamma} a_{r\tau}^\dagger a_{s\tau},
\end{multline}
where $a_{p\gamma}^{(\dagger)}$ is the annihilation (creation) operator for spatial orbital $p$ and spin $\gamma$, $E$ is the scalar energy offset, $N$ is the number of spatial orbitals, and $(p|\kappa|q)$ and $(pq|rs)$ are the one- and two-electron integrals.
\textbf{G}eneralized \textbf{Ra}nk \textbf{De}composition (GRADE) factorizes the two-electron
integrals in the following way
\begin{equation}\label{eq:GRADE}
 (pq|rs) = \sum_{\ell=1}^L \sum_{k,l = 1}^{M_\ell}
 U^{(\ell)}_{pk} U^{(\ell)}_{qk} Z^{(\ell)}_{kl}
 U^{(\ell)}_{rl} U^{(\ell)}_{sl} + R_{pqrs},
\end{equation}
where $R_{pqrs}$ is the residual two-electron integral tensor capturing what the rank-$L$ factorization leaves behind. We can further approximate
\begin{equation}\label{eq:GRADE_residual}
    \sum_{pqrs}R_{pqrs} a_{p\gamma}^\dagger a_{q\gamma} a_{r\tau}^\dagger a_{s\tau} \approx \sum_i \alpha_i P_i.
\end{equation}
The $\alpha_i P_i$ symbols represent additional weighted Pauli strings approximating the (unfactorized) remainder.
The matrices $U^{(\ell)}$ are used to implement single-particle basis rotations
 \begin{equation}
    a^{(\ell)\dagger}_{k\gamma} = \sum_{p = 1}^{M_\ell} U_{pk}^{(\ell)} a^\dagger_{p\gamma}, \quad a^{(\ell)}_{k\gamma} = \sum_{q=1}^{M_\ell} U_{qk}^{(\ell)} a_{q\gamma};
\end{equation}
 and diagonalize the fragments in the Hamiltonian.
 The $Z^{(\ell)}$ symmetric matrices are the coefficients for $\sigma_{z,k\gamma} \otimes \sigma_{z,l\tau}$ Pauli strings that appear when we construct the $\ell$-indexed fast-forwardable fragments. Performing these basis rotations using the ansatz, we rewrite the Hamiltonian in the following form, as shown previously in Ref. \cite[App. A]{fomichev2025fast}
\begin{multline}\label{eq:GRADE_Hamiltonian}
    H  \approx \left(E + \sum_k Z_k^{(0)} -\frac{1}{2}\sum_{\ell,kl}Z_{kl}^{(\ell)}+\frac{1}{4} \sum_{\ell,k} Z^{(\ell)}_{kk} \right)\bm{1}\\ 
    +\sum_{\ell \geq 1} \bm{U}^{(\ell)} \Bigg[\frac{1}{8}\sum_{(k,\gamma)\neq(l,\tau)}  \left( Z^{(\ell)}_{kl} \sigma_{z,k\gamma} \sigma_{z,l\tau}\right)  \\
    - \frac{1}{2}\sum_{k,\gamma}\left(\sum_l Z^{(\ell)}_{kl} \sigma_{z,k\gamma}\right)\Bigg](\bm{U}^{(\ell)})^T\\
    -\frac{1}{2}\bm{U}^{(0)} \left[\sum_k Z^{(0)}_{k} \sum_\gamma  \sigma_{z,k\gamma} \right](\bm{U}^{(0)})^T
    + \sum_i\alpha_i P_i.
\end{multline}
There is one key difference with respect to
Ref. \cite{fomichev2025fast}. In Ref. \cite{fomichev2025fast} the authors absorb the one-body contribution of each two-body fragment in the one-body fragment. In GRADE, that would lead to a one-body fragment acting on $\max_{\ell }M_{\ell}$ qubits. Instead, we keep each correction within its corresponding two-body fragment, where the correction terms are already diagonal. Thus, we explicitly include $1/2\sum_l Z^{(\ell)}_{kl} \sigma_{z,k\gamma}$. This leads to $M_\ell$ additional $\sigma_z$ rotations in each fragment, beyond the $\sigma_z\otimes \sigma_z$, but keeps the one-body fragment acting on $N$ qubits only.

The unitaries $\bm{U}^{(\ell)}$ appearing in the Hamiltonian are full Hilbert space rotations constructed from the single-particle basis changes $U^{(\ell)}$ using Thouless's theorem \cite{kivlichan2018quantum,thouless1960stability},
\begin{align}\label{eq:basis_rotation}
    \bm{U}^{(\ell)} = \exp\left(\sum_{p,q}[\log U^{(\ell)}]_{pq} (a^\dagger_p a_q -a^\dagger_q a_p) \right),
\end{align}
and implemented via Givens rotations \cite{arrazola2022universal},
\begin{equation}\label{eq:Givens}
    G(\theta) 
=  \exp\left( i \frac{\theta}{2} (X \otimes Y - Y \otimes X) \right).
\end{equation}
Time evolution under the rank factorized fragments can be performed with the circuit shown in~\cref{fig:CDF_circuit}, for more details, see Ref. \cite{fomichev2025fast}. There is one final subtle difference with respect to prior literature \cite{fomichev2025fast}: in CDF, it is common to merge two subsequent basis rotations \cite{cohn2021quantum,fomichev2025fast}; on the other hand, isometric THC adopts the basis of the one-body term and only implements the partial basis rotations of the two-body term \cite{luo2025efficient}. Since GRADE interpolates between both regimes, we must choose between two strategies: explicitly implementing the partial basis rotations $U^{(\ell)}$ and $(U^{(\ell)})^\dagger$ without merging nearby basis rotations, or implementing a single full basis rotation on the $M_\ell$ qubits that accounts for the product $U^{(\ell)}(U^{(\ell+1)})^\dagger$. Overall, we have generally found that it is cheaper to use the single full basis rotation.

Overall, the GRADE factorization framework unifies several factorization schemes for electronic Hamiltonians: 
\begin{enumerate}
\item direct qubit mappings, such as Jordan-Wigner or Bravyi-Kitaev, when $L=0$; 
\item CDF when $M_\ell=N$ and $\alpha_i=0$ \cite{cohn2021quantum, motta2021low, oumarou2024accelerating},
\item isometric THC for $L=1$ and $\alpha_i=0$ \cite{luo2025efficient}.
\end{enumerate}
Since standard and isometric THC are equivalent \cite{luo2025efficient}, and THC is DFTHC with parameters $(R, B, C) = (1, R_{THC}, R_{THC})$ \cite{low2025fast}, GRADE can also be seen as an extension of the DFTHC approach that incorporates Pauli strings and fragment-specific $M_\ell$ parameters. The $R, B$ DFTHC parameters are mapped to $L, M_\ell$ in GRADE, while $C$ is left unconstrained.

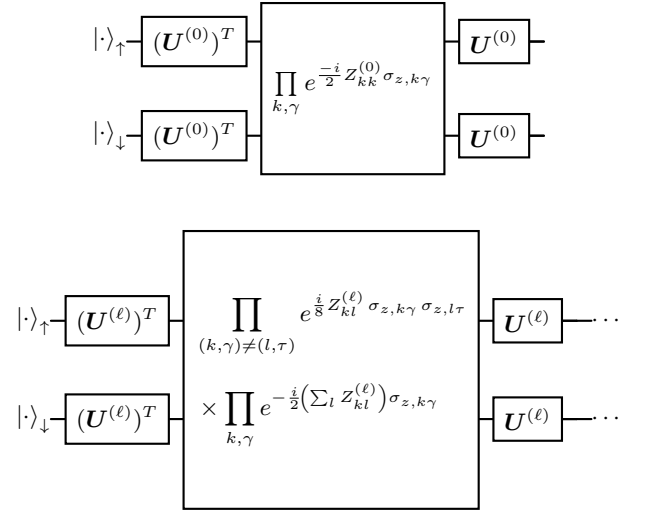
\begin{figure}[t]
    \centering
    \begin{quantikz}[row sep=0.7em, column sep=0.6em]
  \ket{\cdot}_{\uparrow}  & \gate{(\bm{U}^{(0)})^T} & \gate[wires=2]{\prod\limits_{k,\gamma} e^{\frac{-i}{2} Z^{(0)}_{kk}\sigma_{z,k\gamma}}} & \gate{\bm{U}^{(0)}} & \qw \\
  \ket{\cdot}_{\downarrow}  & \gate{(\bm{U}^{(0)})^T} & & \gate{\bm{U}^{(0)}} & \qw
    \end{quantikz}
    
    \vspace{0.5cm}
    
\begin{quantikz}[row sep={4.em,between origins}, column sep=0.6em]
\ket{\cdot}_{\uparrow} & \gate{(\bm{U}^{(\ell)})^T } & \gate[wires=2]{
 \begin{aligned}
 &\displaystyle\prod_{(k,\gamma)\neq(l,\tau)} e^{\frac{i}{8} Z^{(\ell)}_{kl}\,\sigma_{z,k\gamma}\,\sigma_{z,l\tau}}\\[6pt]
 &\displaystyle\times \prod_{k,\gamma} e^{-\frac{i}{2}\!\left(\sum_l Z^{(\ell)}_{kl}\right)\sigma_{z,k\gamma}}
 \end{aligned}
 } & \gate{\bm{U}^{(\ell)}} & \qw & \cdots \\
\ket{\cdot}_{\downarrow} & \gate{(\bm{U}^{(\ell)})^T } & & \gate{\bm{U}^{(\ell)}} & \qw & \cdots
\end{quantikz}
    
    \caption{Circuits implementing the one-body fragments (top) and two-body fragments (bottom) arising in the factorization of the electronic Hamiltonian \cite{cohn2021quantum}. The register is split into two wire groups corresponding to the spin up and spin down sectors. Here $Z^{(\ell)}$ are symmetric matrices, $\sigma_{z,k\gamma}$ is the Pauli $Z$ operator acting on the $(k,\gamma)$-th spin orbital, and $\bm{U}^{(\ell)}$ are the unitaries transforming the system register according to the effect of the corresponding single-particle basis rotation $U^{(\ell)}$ in~\cref{eq:GRADE}.}
    \label{fig:CDF_circuit}
\end{figure}

\subsection{Computing a GRADE factorization}
\label{ssec:grade-compute}

To compute a GRADE factorization for a given Hamiltonian, we first select the hyperparameters $L$, $M_\ell$, a Pauli-string threshold $\bar{t}$ (i.e. the threshold above which elements of the remainder $R_{pqrs}$ are still converted to Pauli strings, rather than  discarded) and $b$ bits of precision for the rotations. To start, in practice we typically set the size of the first two-body fragment $M_{\ell=1} = N$, and then slowly increase the $M_{\ell}$ value with $\ell$ as we optimize fragment-by-fragment. Ideally, we would choose $O(\log N)$ fragments of increasing size $M_\ell$ to achieve the same cost scaling as isometric THC, $O(N^2)$ per Trotter step. Typical values for the accuracy $b$ range between $10$ (approximately a rotation error of $\epsilon = 10^{-3}$) and $20$ (approximately $\epsilon = 10^{-6}$). The Pauli-string threshold $\bar{t}$ is often chosen in the $10^{-3}$ to $10^{-1}$ range, if used. We found that these parameter ranges for $b$ and $\bar{t}$ were sufficient for the approximately 1 eV spectral precision requirements of the application we discuss in \cref{sec:Applications}, X-ray absorption spectroscopy \cite{fomichev2025fast}.

Using a slowly increasing $M_\ell$ that starts at $M_{\ell = 1} = N$ and fitting each fragment to the residual left by the previous one has two key advantages. First, it makes larger-$\ell$ fragments act as smaller corrections, see~\cref{fig:Zs} and \cite[Fig. 7]{gunther2025phase}. This choice creates a hierarchical structure within the Hamiltonian that we explicitly exploit in our product formula design.
Second, having the leading fragments have low $M_\ell$ helps minimize the effect of the so-called \textit{leakage error}. The leakage error, first identified in the isometric THC method \cite{luo2025efficient}, arises when the wavefunction `leaks' to qubits representing auxiliary orbitals $M_\ell >N$ that were introduced in the course of factorizing the Hamiltonian. Because our $M_\ell$ schedule restricts large-norm early fragments to small $M_\ell$, leakage occurs mostly through the smaller-norm later fragments, minimizing its overall impact.

\begin{figure}[t]
  \centering
  \includegraphics[width=\linewidth]{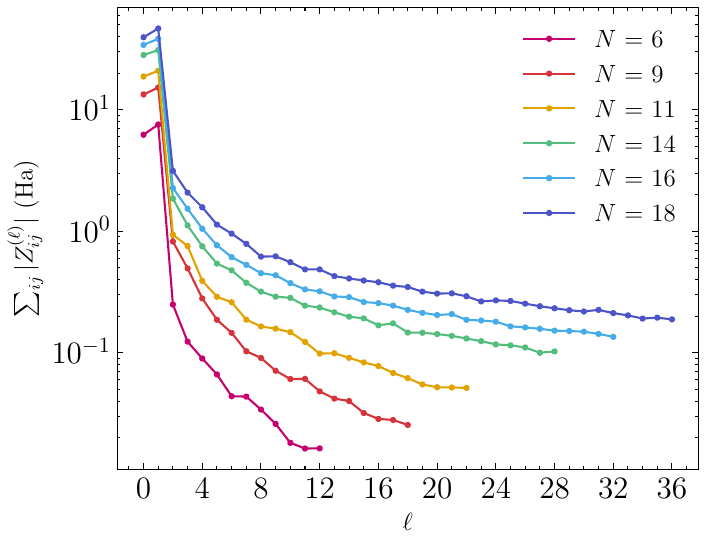}
  \caption{Frobenius norm of different fragments in the CDF decomposition for different active space sizes $N$ of Li$_4$Mn$_2$O. The term $\ell = 0$ denotes the one-body term.}
  \label{fig:Zs}
\end{figure}

Once the hyperparameters are chosen, the coefficients $U^{(\ell)}$ and $Z^{(\ell)}$ are numerically fitted, fragment by fragment, to the residual of the two-electron integrals tensor, under the accuracy restriction of $b$ bits for the rotations. Any residual terms with coefficients whose absolute value is larger than $\bar{t}$ are retained as explicit Pauli strings. More concretely, each fragment is numerically fitted using gradient descent to the residual left by previous fragments. Explicitly, we minimize
\begin{multline}
    \min_{U^{(\ell)}, Z^{(\ell)}}\left\|(pq|rs) - \sum_{\ell=1}^{L} \sum_{k,l = 1}^{M_\ell} U^{(\ell)}_{pk} U^{(\ell)}_{qk} Z^{(\ell)}_{kl} U^{(\ell)}_{rl} U^{(\ell)}_{sl}\right\|.
\end{multline}
If $M_{\ell} = N$, we parametrize
\begin{equation}
    U^{(\ell)} = \exp((X^{(\ell)}-X^{(\ell),T})/2),
\end{equation}
and impose symmetric matrices $(Z^{(\ell)}+Z^{(\ell),T})/2$ \cite{cohn2021quantum}. We then run gradient descent on parametric real matrices $X^{(\ell)}$ and $Z^{(\ell)}$, initialized uniformly at random or with the parameters of the explicit Double Factorization algorithm \cite{von2021quantum}. If instead $M_\ell > N$, we use the procedure described in Ref. \cite{luo2025efficient}. First, numerically find a THC decomposition of the residual \cite{lee2021even,caesura2025faster}; second, isometrize the rotations, dividing the partial rotations by an appropriate scalar; and third, numerically optimize the $V^{(\ell)}$ and $Z^{(\ell)}$ parameters via gradient descent.  The isometry $V^{(\ell)}$ corresponding to a partial basis rotation $U^{(\ell)}$ is defined as a mapping between vector spaces that preserves the inner product, $V^{(\ell)\dagger} V^{(\ell)} = 1$. The three-step method above results in approximate isometries, but only exact isometries are implementable in a quantum computer. To address this problem, we implement a singular value decomposition (SVD) of the isometry, project the singular values to the closest scalar of absolute value 1, and reconstruct the closest exact isometry. After the isometry has been set, we compute a QR decomposition of $V^{(\ell)}$ and recover the full basis rotation matrix $U^{(\ell)}$ \cite{kivlichan2018quantum}.

Finally, once a factorization is obtained, it needs to be tested for how well it is capturing the key spectral properties of the Hamiltonian. In previous works focused on ground-state energy problems, the standard approach was to take the final, optimized $U^{(\ell)}$ and $Z^{(\ell)}$ matrices, reconstitute the fermionic Hamiltonian, and perform a ground-state energy calculation using a classically tractable method such as coupled cluster. If the result of this calculation is close to that obtained by applying coupled cluster to the original Hamiltonian, the factorization is considered high-quality. In our case, since we are interested in Hamiltonian excited states and spectra more broadly, we will evaluate several (typically 10) of the lowest-lying energies of the factorized Hamiltonian using density matrix renormalization group (DMRG), and compare them with those of the original Hamiltonian. If the average of the deviations exceeds our pre-specified cutoff -- typically 1 eV for the XAS application \cite{fomichev2025fast} -- we restart the factorization with a different set of hyperparameters. We iterate this outer loop until we obtain a factorization that meets our energy deviation cutoff.

More specifically, once we have obtained all of the $Z^{(\ell)}$ and $U^{(\ell)}$ matrices for each fragment, we can map back the Pauli operator to a fermionic operator
\begin{multline}\label{eq:el-ham_expanded}
    H = E + \sum_{p,q = 1}^{\max M_\ell} \sum_{\gamma \in \{\uparrow,\downarrow\}} (p|\kappa|q) a_{p\gamma}^\dagger a_{q\gamma}+ \\
    \frac{1}{2}\sum_{p,q,r,s=1}^{\max M_\ell}\sum_{\gamma,\tau \in \{\uparrow,\downarrow\}}(pq|rs) a_{p\gamma}^\dagger a_{q\gamma} a_{r\tau}^\dagger a_{s\tau},
\end{multline}
whose low-lying eigenstates we can then compute using classical methods like DMRG. Doing the same for the original, non-factorized Hamiltonian and comparing the results allows us to conclude whether or not the factorization is faithfully preserving the spectral properties of the original Hamiltonian. 

\subsection{Leakage error}
\label{ssec:grade-leakage}

As a final note, we highlight that the GRADE-factorized Hamiltonian evidently has new terms that were not present in the original Hamiltonian. These are terms that exchange electrons between the physical and ``auxiliary orbitals'' -- orbitals corresponding to qubits that arise in the factorization in the case of $M_\ell > N$. We will call such terms `leakage' terms in the Hamiltonian: this introduces a new source of error to the simulation that must be carefully controlled (see \cref{ssec:Symmetry_protection}). In this sense, the formulation of isometric THC and GRADE differs from the description typically used in the THC qubitization literature \cite{lee2021even,caesura2025faster,luo2025efficient}. In the standard formulation of THC, the two body fragment of the Hamiltonian is decomposed as
\begin{equation}
    (pq|rs) = \sum_{\mu,\nu=1}^{M} \chi_p^{(\mu)} \chi_q^{(\mu)} \zeta_{\mu\nu} \chi_r^{(\nu)} \chi_s^{(\nu)}.
\end{equation}
The two body integrals tensor is described as a linear combination of $M^2$ unitaries, running over indices $\mu$ and $\nu$. Implementing each unitary requires $O(N)$ gates, which would naively suggest an overall $O(M^2N) = O(N^3)$ gate cost. However, qubitization relies on QROM methods to efficiently prepare such a linear combination at a $O(MN)$ per-qubitization-step cost. In contrast, Trotter formulas do not generate a linear combination of unitaries; rather, they rely on the product of exponentials of individual fragments. Consequently, a naive Trotter-based THC strategy would require $O(M^2N) = O(N^3)$ gates. To overcome this disadvantage, isometric THC creates a single fragment with a global basis rotation \cite{luo2025efficient}. This reduces the per-step gate complexity to $O(M^2)$, at the cost of requiring $2(M-N)$ auxiliary qubits.

\section{Designing product formulas\label{sec:Product_formulas}}

In this section we introduce each element of the SPRINT framework of \cref{fig:SPRINT_product_formula}. We will first describe the main tools and notation we will use in our analysis, in~\cref{ssec:Product_formulas}. In~\cref{ssec:Near_integrability} and~\cref{ssec:Processing} we will explain how to use such tools to design what we are calling \textit{near-integrable} and \textit{processed} product formulas adapted to the structure of the electronic Hamiltonian such as that generated by GRADE.  In~\cref{ssec:Symmetry_protection} we will analyze \textit{symmetry protection} as a way to address the leakage error. In~\cref{ssec:randomization} we describe how \textit{randomization} of the fragment ordering within a product formula may be used to further reduce the cost. We end with~\cref{ssec:QROM}, where we introduce a way to leverage QROM to reduce the Toffoli gate count of implementing Trotter steps.

\subsection{Product formulas\label{ssec:Product_formulas}}

The goal of this paper is to explain how to implement Hamiltonian simulation, or in other words, implement the operator $e^{-i \tau H}$.
Product formulas are one way to do this -- by what is essentially a divide-and-conquer approach. Since we cannot, in general, easily exponentiate $H$ directly, we split $H$ into easily exponentiable (also known as fast-forwardable) fragments $H_j$, $H = \sum_j H_j$. In our case, the $H_j$ are either the $\ell$-indexed factorized fragments from \cref{sec:Factorization} or the individual Pauli strings in~\cref{eq:GRADE_Hamiltonian}. 
Then, we compose a product formula $U(\tau)$ as the product of the evolution under different fragments: in general, it takes the form
\begin{equation}\label{eq:product_formula}
    e^{-i \tau H}\approx U(\tau) = \prod_k \exp(-i \tau a_k H_k).
\end{equation}
Each $H_j$ may appear multiple times in the product formula step. The number and order of terms and the coefficients $a_k$ control the accuracy of the approximation to $e^{-i \tau H}$. The main task in product formula design is to find the right fragment order $H_k$ and coefficients $a_k$.

\textbf{Baker-Campbell-Hausdorff expansion:} Our main tool to \textit{design} product formulas will be the Baker-Campbell-Hausdorff (BCH) expansion, which, for any two operators $X$ and $Y$, computes $Z$ such that $e^X e^Y = e^Z$ \cite{baker1905alternants,campbell1896law,hausdorff1906symbolische}:
\begin{multline}\label{eq:BCH}
    Z = X + Y + \frac{1}{2}[X, Y]\\
    + \frac{1}{12}\left([X, [X,Y]] + [Y, [Y, X]]\right) + \ldots
\end{multline} 
A key implication of the BCH expansion is that \textit{product formulas implement exact Hamiltonian simulation of an approximate Hamiltonian} (this property is sometimes called \emph{symplecticity} \cite{blanes2008splitting,blanes2024splitting}). We will show this explicitly below. Given this, our priority is to understand what effective Hamiltonians are generated by different product formulas. That effective Hamiltonian will determine everything from the Trotter error to the leakage error to the implementation cost.

We can compute the effective Hamiltonian generated by a product formula by recursively applying the BCH formula to the exponentials and keeping only the lowest order terms. For example, the first order formula, also called Lie-Trotter, is just a direct product of the exponentials of individual fragments \cite{trotter1959product}
\begin{equation}\label{eq:U1}
    U_{1}(\tau) = \prod_j \exp(-i H_j \tau)
\end{equation}
Using the BCH expansion, we can write the effective Hamiltonian implemented by a first order formula
\begin{multline}\label{eq:U1_eff_Ham}
    U_{1}(\tau) = \exp(-i \tau H + (-i\tau)^2 Y^{(1)}_{2} + (-i\tau)^3 Y^{(1)}_{3}\ldots),
\end{multline}
where
\begin{align}\label{eq:Y_2^1}
    Y^{(1)}_{2} &= \frac{\sum_k\left[H_k, \sum_{j>k} H_j\right]}{2},\\
    \label{eq:Y_3^1}
    Y^{(1)}_{3} &= \frac{\sum_k\left[H_k,\left[H_k, \sum_{j>k} H_j\right]\right]}{12}\\
    &+ \frac{\sum_k\left[\sum_{j>k}H_j,\left[\sum_{j>k}H_j, H_k\right]\right]}{12}.\nonumber
\end{align}

\textbf{Symmetric BCH expansion:}
While using the BCH expansion in~\cref{eq:BCH} is in principle sufficient to analyze any product formula like \cref{eq:product_formula}, we will find it very useful to leverage the \textit{symmetric} version of the BCH, which computes $Z$ such that $e^{X/2} e^Y e^{X/2} = e^Z$ \cite{casas2009efficient}: 
\begin{equation}\label{eq:symmBCH}
    Z = X + Y  - \frac{1}{24}[X, [X,Y]] -  \frac{1}{12}[Y, [X, Y]] + \ldots
\end{equation}
The symmetric version is helpful in analyzing symmetric product formulas ``from the middle outwards'': this is essential for all product formulas beyond the first-order one, and especially the near-integrable formulas that we introduce in this manuscript. As an example, consider the second-order Trotter formula, also known as the Strang or St\"{o}rmer--Verlet splitting~\cite{strang1968construction,stormer1907trajectoires,verlet1967computer}. This formula takes half-steps of duration~$\tau/2$ and is symmetric: it applies all fragment exponentials first in ascending and then in descending order:
\begin{equation}\label{eq:U2}
    U_{2}(\tau) = \underbrace{\prod_{j=1}^L \exp(-i H_j \tau/2)}_{U_{1}(\tau/2)} \times \underbrace{\prod_{j=L}^1 \exp(-i H_j \tau/2)}_{U^\dagger_{1}(-\tau/2)}
\end{equation}
Applying the symmetric BCH expansion, we find
\begin{equation}
    U_{2}(\tau) = \exp(-i \tau H + (-i\tau)^3 Y^{(2)}_{3} + (-i\tau)^5 Y^{(2)}_{5}\ldots),
\end{equation}
where
\begin{multline}\label{eq:Y_3^2}
 Y_3^{(2)}:=-\sum_{j} \Bigg[\frac{[H_j,[H_j, \sum_{k>j}H_k]]}{24}\\ + \frac{[\sum_{k>j}H_k,[H_j, \sum_{k>j}H_k]]}{12} \Bigg].
\end{multline}
A key fact to notice about this product formula is that even orders of $\tau$ in the effective Hamiltonian of~\cref{eq:U2} cancel out because the product formula is symmetric, $U_{2}(\tau) = U^\dagger_{2}(-\tau)$. By imposing this symmetry, which has a moderate cost, we automatically cancel not just the leading order error, but all even order error terms.

\textbf{Higher order product formulas:}
Finally, before we move to the near-integrable derivation, for completeness we also give a brief summary of higher order formulas. A product formula is said to achieve \textit{order $k$ if the leading order contribution in the effective Hamiltonian is $O(\tau^k)$}. For example, the first order formula effective Hamiltonian has leading error at $O(\tau)$, while the second order formula effective Hamiltonian has leading error $O(\tau^2)$.
We will denote such a $k$-th order product formula by $U_{k,X}(\tau)$ where $X$ designates a group of Hamiltonian fragments the product formula is being applied to. The Suzuki hierarchy provides a systematic way of constructing product formulas of arbitrary even order \cite{suzuki1990fractal,suzuki1991general}
\begin{equation}\label{eq:U_2k}
    U_{2k}(\tau) = U_{2k-2}^2(u_k\tau) U_{2k-2}((1-4u_k)\tau) U^2_{2k-2}(u_k\tau).
\end{equation}
The scalar $u_k$ is selected as
\begin{equation}~\label{eq:Suzuki_constant}
u_k = 1/(4-4^{1/(2k-1)}).
\end{equation}
This choice cancels the $O(\tau^{2k-2})$ error, as shown in~\cref{prop:Suzuki} in~\cref{app:proofs}.
In general, the standard Suzuki hierarchy product formula of order $k$ has the effective Hamiltonian
\begin{multline}
    U_{2k}(\tau) = \\
    \exp\left(-i \tau \underbrace{\left(H + \sum_l (-i\tau)^{2(k+l)}Y^{(2k)}_{2(k+l)+1}+\ldots\right)}_{\text{effective Hamiltonian}}\right),
\end{multline}
where $Y^{(2k)}_{2(k+l)+1}$ represents a linear combination of $[2(k+l)+1]$-nested commutators.\\

\subsection{Near-integrability\label{ssec:Near_integrability}}

With the key analytical tools at hand, we design so-called \textit{near-integrable} product formulas -- product formulas that exploit how different fragments contribute unequally to the Trotter error. 

\textbf{Intuition:} As discussed in~\cref{ssec:GRADE} and shown in~\cref{fig:Zs}, once factorized, the one-body and the first two-body fragment of the Hamiltonian will typically have norms significantly larger than the norms of the long tail of smaller fragments. To make this distinction explicit, we will write
\begin{equation}
    H_A = H_0 + H_1\qquad H_B = \alpha \sum_{\ell>1} H_\ell,
\end{equation}
with $\alpha \ll 1$ the parameter that captures the energy scale difference, and $\ell = 0$ the one-body fragment.

If the norms of the fragments in the $H_A$ group are larger than the norms of the fragments in $H_B$, which is captured by $\alpha\ll 1$, then nested commutators that have more factors of $H_0$ or $H_1$ will typically contribute more to the total error.
For example, it is expected that $[H_0, [H_0, H_1]]$ will usually contribute much more to the total Trotter error than $\alpha^3[H_2, [H_2, H_3]]$. Thus, it will typically make sense to use higher order methods for $H_A$ than for $H_B$.

\textbf{Near-integrable formulas:} 
Here we introduce two near-integrable formulas we derived for the characteristic Hamiltonian norm structure shown in \cref{fig:Zs}: we will explicitly apply these formulas to a concrete problem in~\cref{sec:Applications}. The first near-integrable product formula we design is the `second-order' near-integrable formula
\begin{equation}\label{eq:V_21}
    V_{2,1}(\tau) = U_{1,A}(\tau/2)\times U_{1,B}(\tau) \times  U_{1,A}^\dagger(-\tau/2),
\end{equation}
where $U_{n, A}$ is the $n$-th order Trotter product formula applied to fragment group $A$. The subindex 2,1 in $V_{2,1}(\tau)$ makes reference to the order of the product formula steps used: first order for $H_B$, and second order for $H_A$ (a second order formula is just two first-order formulas back to back, $U_2(\tau) = U_1(\tau/2) U_1^\dagger(-\tau/2)$). This generates an effective Hamiltonian satisfying
\begin{multline}\label{eq:V_21_Hamiltonian}
  V_{2,1}(\tau) = \exp\!\bigg(-i\tau H + (-i\tau)^2\alpha^2 Y_{2,B}^{(1)} + (-i\tau)^3 Y_{3,A}^{(2)} \\
  + \frac{(-i\tau)^3\alpha}{4}[Y_{2,A}^{(1)},\,H_B] - \frac{(-i\tau)^3\alpha}{24}[H_A,[H_A,H_B]] \\- \frac{(-i\tau)^3\alpha^2}{12}[H_B,[H_A,H_B]] + \ldots\bigg),
\end{multline}
as can be seen directly from applying the symmetric BCH expansion to the $V_{2,1}$ ansatz. Here, $Y_{n,X}^{(p)}$ represent the linear combination of $n$-nested commutators of the $p$-th order formula of the terms in $H_X$, see~\cref{eq:Y_2^1,eq:Y_3^1,eq:Y_3^2} for $Y_2^{(1)}, Y_3^{(1)}$ and $Y_3^{(2)}$ respectively. 
 
 Note the advantage of using this near-integrable formula. Comparing to the Trotter second order formula $U_2(\tau)$, where the leading error term scales as $O(\tau^3)$. In regimes where $\tau^2 \alpha^2 < \tau^3$, the product formula design of $V_{2,1}$ achieves error comparable to the second-order formula but at a cost resembling that of a first-order formula. This is because most of the costly-to-implement fragments are in $U_{1,B}$, while the two main Trotter-error-contributing fragments are in the more tightly-error controlling $U_{2,A}$. 

The second product formula we design is a `fourth-order' near-integrable product formula, given by
\begin{equation}\label{eq:V_42}
  V_{4,2}(\tau) = U_{4,A}(\tau/2)\times U_{2,B}(\tau)\times U_{4,A}(\tau/2),
\end{equation}
which generates an effective Hamiltonian satisfying
\begin{multline}\label{eq:V_42_Hamiltonian}
  V_{4,2}(\tau) = \exp\bigg(-i \tau H 
  + (-i\tau)^3 \alpha^3 Y_{3,B}^{(2)} + 2(-i\tau/2)^5 Y_{5,A}^{(4)} \\
  - (-i\tau)^3 \alpha^2 \frac{[H_B, [H_A, H_B]]}{12}
  - (-i\tau)^3 \alpha \frac{[H_A, [H_A, H_B]]}{24} + \ldots\bigg)\\
  = \exp\!\left(-i\tau H + O(\tau^5) + O(\alpha\tau^3)\right),
\end{multline}
as shown in~\cref{prop:near_integrable_error}. We provide an expression for $Y^{(4)}_5$ in~\cref{Y^{4}_5}, while the $Y_3^{(2)}$ term is as in \cref{eq:Y_3^2}. This near-integrable formula should be compared to the Suzuki fourth order formula $U_4(\tau)$, where the leading error term scales as $O(\tau^5)$. Similar to the case of $V_{2,1}$, what the product formula design of $V_{4,2}$ achieves is error comparable to the fourth order formula, but with the cost resembling that of a second order formula. 

These product formulas are not the only possible design choices. For example, for the second order near-integrable formula, one may instead consider
\begin{equation}
    \tilde{V}_{2,1}(\tau) = U_{2,A}(\tau/2) \times U_{1,B}(\tau) \times U_{2,A}(\tau/2).
\end{equation}
instead of~\cref{eq:V_21}. This further halves the $O(\tau^3)$ error associated to the $H_A$ group of fragments in~\cref{eq:V_21} and  eliminates the error term $(-i\tau)^3\alpha[Y_{2,A}^{(1)},\,H_B]/4$, at the expense of duplicating the cost associated with $U_{2,A}(\tau)$.
Similarly, in~\cref{eq:V_42} one may consider a variation where we replace the fourth order step with $n$ second order steps:
\begin{align}
  \tilde{V}_{2n,2}(\tau) &=  U_{2,A}^n\left(\frac{\tau}{2n}\right) \times U_{2,B}(\tau) \times U_{2,A}^n\left(\frac{\tau}{2n}\right).
 \end{align}
 This decreases the order associated to $H_A$ from $O(\tau^5)$ to $O(\tau^3)$, but also multiplies the $H_A$-associated second order error by a factor of $(1/2n)^2$ compared to the fragments in $H_B$. 
 To decide between these variations, we need to incorporate further information about the specific system we are simulating. We suggest using Trotter error estimation tools to evaluate the contribution of each commutator to the Trotter error, and prioritize the resource allocation accordingly. Crucially, such choices cannot be based on crude upper bounds to the norm, since we already exploited that information. Instead, we suggest estimating the dominant matrix elements in the perturbation expansion of the effective Hamiltonian’s eigenvalues via numerical methods. We show an example of this on the case of the Li-excess cluster calculation in~\cref{sec:Applications}.

\textbf{Crafting near-integrable formulas:} The general procedure to build near-integrable formulas for different desired error orders and fragment norm structures is the following. Given a factorization of the Hamiltonian that naturally splits into a large-norm group $H_A$ and a more numerous small-norm group $H_B$, 
\begin{enumerate}
    \item Evaluate the norm of the different fragments in the Hamiltonian, which will serve as a proxy for their contribution to the error; 
    \item Use the BCH expansion to evaluate the symbolic error terms in the effective Hamiltonian of different product formula ans\"{a}tze;
    \item Adjust parameters $a_k$ in \cref{eq:product_formula} or the ans\"{a}tze to minimize the resulting error, based on the available information about the relative sizes of the norms. 
\end{enumerate}
Overall, we prioritize canceling the coefficients of nested commutators with either a lower $\tau$ scaling, \textit{or larger norm}, as these typically contribute the most to the error.

\textbf{How to estimate norms:} A natural question at this stage is how exactly the norm of a Hamiltonian fragment should be measured. Throughout this manuscript, we will characterize the fragment's contribution to the error in terms of the Frobenius norm of its coefficient matrix $Z^{(\ell)}$: we take $\|Z^{(\ell)}\|_F$ as a cheap proxy for the fragment's contribution to overall Trotter error. 
While $\|Z^{(\ell)}\|_F$ can in principle be used in a formal Trotter error bound via standard commutator inequalities, such a bound is too loose to be practically informative. We therefore use it as an empirical proxy. Nevertheless, as we will see in~\cref{sec:Applications}, empirically we observe it to be a cheap and reliable indicator of Trotter error contribution.

\subsection{Processing\label{ssec:Processing}}

Near-integrability is not the only way to exploit knowledge of the structure of commutator error terms in the effective Hamiltonian. As motivation, consider again \cref{eq:V_42_Hamiltonian}. Ideally, we would like to eliminate the $O(\alpha\tau^3)$ term in the error of $V_{4,2}(\tau)$ without increasing the cost. It turns out it is possible to do so with \textit{processing}, a basis change by an operator $P$ constructed from the evolution under $H_A$ and $H_B$ \cite{blanes2000processing,blanes2004numerical,blanes2024splitting,morales2025selection, blanes2024families}. A key property of the processor is that it incurs only a constant additive implementation cost. If $K$ is the effective Hamiltonian (the kernel) generated by $V_{4,2}(\tau)$, written as $V_{4,2}(\tau) = e^K$, the processor $P$ will only need to be implemented once at the beginning and once at the end of the full time evolution, because the intermediate rotations cancel out:
\begin{equation}
  U = e^{-P}e^K e^P \Rightarrow U^n = e^{-P}(e^K)^n e^P.
\end{equation}
Here $K$ represents the effective Hamiltonian. For example, if we used~\cref{eq:V_42},
\begin{multline}
  K = H 
  - \tau^2 \alpha^3 Y_{3,B}^{(2)} + \left(\frac{\tau}{2}\right)^4 Y_{5,A}^{(4)} \\
  + \tau^2 \alpha^2 \frac{[H_B, [H_A, H_B]]}{12}
  + \tau^2 \alpha \frac{[H_A, [H_A, H_B]]}{24} + \ldots.
\end{multline}
Note that $K$, the effective Hamiltonian, will carry an implicit dependence on the time step $\tau$. Our goal is to figure out which $P$ we can use to make the error in the effective Hamiltonian smaller.

Perhaps the simplest example of a processor is the case of the two-fragment Hamiltonian $H = H_0 + H_1$, where it is possible to use it to convert a first order formula into a second order formula:
\begin{multline}
    e^{+i H_0 \tau/2} U_1(\tau) e^{-i H_0 \tau/2}
    = \\
    e^{+i H_0 \tau/2}(e^{-i H_0 \tau}e^{-i H_1 \tau})e^{-i H_0 \tau/2} =  U_2(\tau).
\end{multline}

In~\cref{prop:processing} we show how to design a processor for the near-integrable formula $V_{4,2}(\tau)$ presented in the previous section, see~\cref{eq:V_42}. Our goal is to remove the $O(\alpha \tau^3)$ nested commutator in the effective Hamiltonian. Processing can also be applied to cancel the same nested commutator in $V_{2,1}(\tau)$ and $\tilde{V}_{2,1}(\tau)$, though in this case there are other terms of the same importance, so the impact could be smaller.

To understand the effect of the processor we use the
Hadamard lemma: for any operators $P$ and $K$, we can expand
\begin{equation}\label{eq:hadamard_lemma}
 e^{P} K e^{-P} = \sum_{j=0}^{\infty} \frac{1}{j!} \underbrace{[P,\ldots[P,}_{j} K]\cdots]
\end{equation}
We state and prove it formally using the Taylor expansion in~\cref{lemma:Hadamard} in~\cref{app:hadamard}.
Applying the Hadamard lemma gives
\begin{equation}\label{eq:processed_Hamiltonian}
 H_{\mathrm{eff}} = K - [P, K] + \frac{1}{2!}[P,[P,K]] + \ldots
\end{equation}
In our case, $K$ is the effective Hamiltonian in~\cref{eq:V_42_Hamiltonian}. In \cref{prop:processing}, we show that if we choose the processor
\begin{equation}\label{eq:processor}
 P = \alpha \tau^2  \frac{[H_A, H_B]}{24},
\end{equation}
it will result in a nested commutator $-[P, K]$ that cancels the $O(\alpha\tau^3)$ error term in~\cref{eq:V_42_Hamiltonian}. Using this processor generates an effective Hamiltonian satisfying
\begin{equation}
 H_{\mathrm{eff}} = \tau(H_A + \alpha H_B) + O(\tau^5) + O(\alpha^2 \tau^3).
\end{equation}
Compared to~\cref{eq:V_42_Hamiltonian}, the processor $P$ not only eliminates the $O(\alpha\tau^3)$ error term, but also halves the coefficient of the $O(\alpha^2\tau^3)$ term. Overall, the processed formula delivers pseudo-fourth-order accuracy at the computational cost of a second-order formula, provided $\alpha \sim \tau$ and $\|H_{A}\| = \|H_B\| = 1$. 

The final question to answer is how to implement the processing basis change in practice. Just like time evolution, evolution under the processor $e^P$ can be approximated by product formulas, some of which are described in Ref. \cite{casas2025approximating}. These product formulas have different order conditions: the $[H_A, H_B]$ prefactor in their effective Hamiltonian should be equal to 1, and the rest of the prefactors -- including the linear order -- should be equal to 0. For example, to implement the processor of~\cref{eq:processor} for $V_{4,2}(\tau)$ we can use
\begin{multline}
e^P = e^{\frac{\alpha\tau^2}{24} [H_A, H_B]} \\
\approx
    e^{+i \frac{\alpha \tau}{\sqrt{24}} H_B} e^{+i \frac{\tau}{\sqrt{24}} H_A}  e^{-i \frac{\alpha \tau}{\sqrt{24}} H_B} e^{-i \frac{\tau}{\sqrt{24}} H_A}.
\end{multline}
Higher order formulas are available in Ref. \cite{casas2025approximating}.
Since it is applied only once at the beginning and end of the simulation, this overhead is negligible for all but the shortest evolution times.

\subsection{Symmetry protection\label{ssec:Symmetry_protection}}
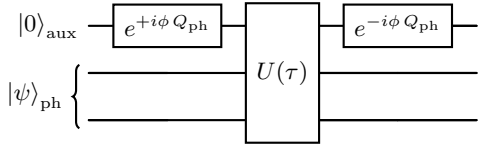
\begin{figure}[t]
    \centering
    \begin{tikzpicture}
        \node[anchor=north] at (9, -1.0) {
        \begin{quantikz}[row sep=1.0em, column sep=1.0em]
          \lstick{$\ket{0}_{\mathrm{aux}}$}
            & \gate[]{e^{+i\phi\, Q_{\text{ph}}}}
            & \gate[3, ]{U(\tau)}
            & \gate[]{e^{-i\phi\, Q_{\text{ph}}}}
            & \qw \\
          \lstick[2]{$\ket{\psi}_{\mathrm{ph}}$}
            & \qw &  & \qw & \qw \\
            & \qw &  & \qw & \qw
        \end{quantikz}
        };
    \end{tikzpicture}
    \caption{Implementation of the symmetry protection. We conjugate the auxiliary qubits in the Trotter step with phases $e^{\pm i \phi Q_{\text{ph}}}$, where $Q_{\text{ph}} = 1-\ket{0}\bra{0}_{\text{aux}}\otimes 1_{\text{ph}}= \sum_{n>0} \ket{n}\bra{n}_{\text{aux}}$ is a projector onto the auxiliary space.}
    \label{fig:symmetry_protection}
\end{figure}

There are three main error sources in the Hamiltonian simulation method discussed so far. First, we have the error from approximation of the Hamiltonian in~\cref{eq:GRADE}, which can be controlled by the factorization optimization parameters. Second, we have the Trotter error, which we are addressing with near-integrability and processing in~\cref{ssec:Near_integrability,ssec:Processing}. Now, we discuss how to mitigate the third error source: the \textit{leakage} of the wavefunction to auxiliary modes whenever $M_\ell > N$.

A priori, a standard product formula gives no control over the leakage error. As noted in~\cref{ssec:grade-compute}, choosing lower $M_\ell$ for the leading fragments confines leakage to the terms that have smaller $Z^{(\ell)}$-norms. Consequently, leakage is primarily driven by smaller-norm fragments, which helps decrease this error somewhat. However, even with this choice, the leakage error is still formally of order $O(1)$. The reason is that the factorized Hamiltonian contains terms coupling the physical and auxiliary subspaces, as discussed in \cref{ssec:grade-leakage}. Time evolution under the factorized Hamiltonian therefore inevitably drives evolution under these leakage terms as well, even though only the physical subspace is meaningful. This leakage error is therefore on the same order as the Hamiltonian itself -- potentially much larger in magnitude than the Trotter error -- and it cannot be controlled by reducing the Trotter time step. To allow us to control and ultimately reduce the leakage error, we propose using symmetry protection.

Symmetry protection addresses the leakage error by conjugating each Trotter step unitary $U$ with phases on the auxiliary orbital qubits, see~\cref{fig:symmetry_protection}. The intuition is that whenever parts of the wavefunction leak into auxiliary orbitals, we will force such components to acquire different phases in different Trotter steps, so over multiple steps they will interfere and cancel out. Specifically, we propose the following symmetry-protected product formula construction for a given product formula $U(\tau)$
\begin{align}
    \label{eq:U_s(tau)}
    U_s(\tau) &= \left(e^{-i \pi Q_{\text{ph}}}\, U(\tau)\, e^{+i \pi Q_{\text{ph}}}\right)\, U(\tau),\\
\label{eq:S(tau)}
    S(\tau) &= U_s^\dagger(-\tau)\, U_s(\tau).
  \end{align}
  Here $Q_\text{ph}$ is a projector into the auxiliary space that we define formally below. We claim that an $S(\tau)$ constructed in this way reduces the order of the leakage error from $O(1)$ to $O(\tau^2)$. The mechanism proceeds in two stages: first, summing over roots of unity cancels the zeroth-order leakage, reducing it to $O(\tau)$; second, symmetrizing the product formula (composing forward and backward cycles) cancels all remaining odd-order terms, yielding $O(\tau^2)$.

This is better seen in the Hamiltonian. We can decompose the effective Hamiltonian into components defined by number of electrons exchanged between the physical and auxiliary orbitals,
\begin{equation}
    H_{\mathrm{eff}} = \sum_{\Delta m} Y_{\Delta m}.
\end{equation}
Here $Y_{\Delta m}$ indicates the part of the Hamiltonian that either preserves ($\Delta m = 0$) or changes ($\Delta m = 1$)  the number of electrons in the auxiliary orbitals. We also define the operator
\begin{equation}
    Q_{\text{ph}} = 1-\ket{0}\bra{0}_{\text{aux}}\otimes 1_{\text{ph}},
\end{equation}
as a projector into the auxiliary orbitals. 
This projector obeys the commutation relation
\begin{equation}
    [Q_{\text{ph}}, Y_{\Delta m}] = (\Delta m)
Y_{\Delta m},
\end{equation}
which can be seen by applying both sides to a state with a constant number of electrons in the auxiliary space. With the projector $Q_\text{ph}$ defined, we conjugate the Trotter step by exponentials of the projector (\cref{eq:U_s(tau)}). Such exponentials of $Q_{\text{ph}}$ can be implemented with a multi-OR gate in the auxiliary orbitals and a single qubit rotation. The multi-OR gate is implemented via multi-controlled NOT.

To see the impact of symmetry protection, we can leverage the fact that the Trotter step is unitary, which means that such phases can be brought `up' directly to the effective Hamiltonian:
\begin{equation}
 e^{-i \phi Q_{\text{ph}}} U(\tau) e^{+i \phi Q_{\text{ph}}}= \exp(e^{-i \phi Q_{\text{ph}}} H_{\text{eff}} e^{+i \phi Q_{\text{ph}}}).
\end{equation}
Using that $[Q_{\text{ph}}, Y_{\Delta m}] = (\Delta m)\, Y_{\Delta m}$ and applying the Hadamard lemma (\cref{lemma:Hadamard}) gives
\begin{multline}
 e^{-i \phi Q_{\text{ph}}} Y_{\Delta m} e^{+i \phi Q_{\text{ph}}} = \sum_j \frac{(-i\phi)^j}{j!} \underbrace{[Q_{\text{ph}}, \ldots [Q_{\text{ph}}, Y_{\Delta m}]]}_{j\text{ nested commutators}}.\\
 = \sum_j \frac{(-i\phi)^j}{j!} (\Delta m)^j Y_{\Delta m}= e^{-i\phi \Delta m} Y_{\Delta m}.
\end{multline}
This expression makes it clear that any term in the Hamiltonian that leaks the wavefunction to auxiliary orbitals (i.e. one with $\Delta m = 1$) will pick up a phase $\phi$.
The strategy is then to select different phases $\phi$ in different Trotter steps, so that we can cancel out, to linear order, the $Y_{\Delta m}$ contributions to time evolution under the Hamiltonian. 
Such angles might be chosen at random, or with a given structure to cancel the most error \cite{tran2021faster}. In~\cref{prop:symmetry_protection} we show that a simple choice of $\phi$ that works in practice is choosing it to be the roots of unity on different Trotter steps. With this choice, and the time-symmetry property of product formulas, we can build \textit{second-order symmetry protection} into any product formula $U(\tau)$ using~\cref{eq:U_s(tau),eq:S(tau)}.
As shown in \cref{prop:symmetry_protection}, this product formula structure ensures the leakage error in the effective Hamiltonian behaves as $O(\tau^2)$:
this places it on par with the intrinsic Trotter error of the second-order formula -- a significant improvement over the previous, uncontrolled $O(1)$ scaling:
\begin{multline}
  H_{\mathrm{eff}}^{(2)} = H
  + \tau^2 (Y_{\mathrm{leak}}^{(2)} + Y_{\mathrm{Trotter}}^{(2)})
  + O(\tau^3),\\ 
  Y_{\text{leak}}^{(2)} = -\left(\frac{1}{24}[H', [H', H_{\text{eff}} ]] + \frac{1}{12} [H_{\text{eff}}, [H', H_{\text{eff}} ]]\right),
\end{multline}
where
\begin{equation}
  H' = \sum_{\Delta m} e^{-i\pi\Delta m} Y_{\Delta m}.
\end{equation}
For an extensive analysis of different variations of symmetry protection and its impact on problems where the Hamiltonian spectrum is of interest, see~\cref{app:symmetry_protection}.

As an alternative to Symmetry Protection, one might consider Quantum Singular Value Transformation (QSVT) \cite{gilyen2019quantum,martyn2021grand} to suppress leakage. However, the impact of QSVT on the spectrum is limited by a fundamental structural reason: it modifies the singular values (magnitudes) of the evolution operator, but it cannot alter its eigenphases. Thus, it will not suppress the energy shift beyond what symmetry protection may or may not have done already.

\subsection{Randomization\label{ssec:randomization}}
Another key tool in SPRINT is randomization, which involves modifying the Hamiltonian fragment ordering \cite{childs2019faster} at each Trotter step. For any Hamiltonian fragment ordering $H = \sum_l H_l$ and Trotter formula, some nested commutators in the effective Hamiltonian will contain no repeated fragments (e.g., $[\ldots, [H_i, H_j]]$). For every such ordering, there exists a permutation $H_i \leftrightarrow H_j$ in the ordering of the fragments that makes the Trotter step~\cref{eq:product_formula} generate an effective Hamiltonian with the reversed nested commutator: $[\ldots, [H_j, H_i]] = -[\ldots, [H_i, H_j]]$. If we compute the effective Hamiltonian of the full evolution operator, the linear order of the BCH expansion indicates we should sum the effective Hamiltonian $H_{\text{eff},j}$ of different Trotter steps $U_j(\tau)$: 
\begin{equation}
    U(J\tau) = \prod_{j=1}^J U_j(\tau)\approx \exp\left(-i\tau\sum_j H_{\text{eff},j}\right),
\end{equation}
This leading-order addition of effective Hamiltonians will in expectation cancel those nested commutators: $[\ldots, [H_j, H_i]]  + [\ldots, [H_i, H_j]] =0$.
This cancellation only applies to commutators where the two innermost terms appear just once in the nested commutator expression, since for example $[H_i, [H_i, H_j]] \neq -[H_j, [H_j, H_i]]$.
This method is most effective at low Trotter orders and incurs no additional computational overhead \cite{childs2019faster}. When combined with a near-integrable formula, we restrict the randomization process to the fragment ordering within each specific group. Randomization does not increase the simulation cost, and incurs only a small amount of spectral line broadening, see~\cref{app:randomized_spectra}, so it is usually worth exploiting.

\subsection{Using QROMs in Trotter formulas\label{ssec:QROM}}

\begin{figure}[t]
\centering

\vspace{4pt}

\begin{quantikz}[row sep=0.8em, column sep=0.5em]
\lstick{$s_k$}
  & \gate[2]{R_{XX}\!\!\left(\!-\!\frac{\pi}{4}\right)}
  & \ctrl{3}
  & \qw
  & \ctrl{3}
  & \gate[2]{R_{XX}\!\!\left(\!+\!\frac{\pi}{4}\right)}
  & \qw \\
\lstick{$s_{k{-}1}$}
  &
  & \qw
  & \ctrl{1}
  & \qw
  &
  & \qw \\
\lstick{$|\theta\rangle_{b-1}$}
  & \qw
  & \qw
  & \gate[2]{\mathrm{Add}}
  & \qw
  & \qw
  & \qw \\
\lstick{$\ket{\phi_b}$}
  & \qw
  & \targ{}
  &
  & \targ{}
  & \qw
  & \qw
\end{quantikz}

\vspace{4pt}
{\small
$\underbrace{\rule{7.5cm}{0pt}}_{\displaystyle G_k(\theta): \;\text{one controlled } \mathrm{Add}_{b-1} = 2(b{-}1) \text{ Toff.}}$
}

\caption{Efficient programmable Givens rotation circuits from Ref. \cite{caesura2025faster}. A Givens rotation $G_k(\theta) = e^{i\pi\theta X_k Y_{k-1}} e^{-i\pi\theta Y_k X_{k-1}}$ has eigenvalues $\{1, e^{i2\pi\theta}, e^{-i2\pi\theta}\}$ and is implemented using a phase gradient state $\ket{\phi_b}$.
The two rotations are fused into a single circuit. A joint $R_{XX}(\pm\pi/4)$ basis change on $s_k, s_{k-1}$ and conditional bit-flips controlled by~$s_k$ replace the two separate diagonalizations. The key change is that both rotations are now implemented by a single \emph{controlled} Gidney adder \cite{gidney2018halving}, with $s_{k-1}$ as the control qubit. Because the fused circuit operates on the full angle~$\theta$ rather than~$\theta/2$, the angle register is one bit shorter: $b{-}1$ bits. A controlled $(b{-}1)$-bit adder costs $2(b{-}1)$~Toffolis, saving $2$~Toffolis per Givens rotation. }
\label{fig:givens_fused_adder}
\end{figure}
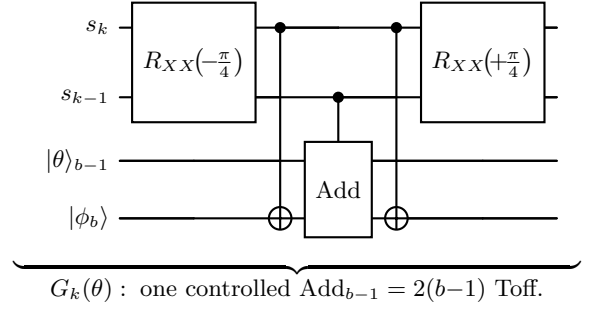

A key advantage of qubitization methods is that they effectively leverage Quantum Read-Only Memories (QROMs) \cite{low2024trading}, which increase the qubit cost but decrease the gate cost. For example, QROM decreases the Tensor Hypercontraction cost scaling from a naive $O(M^2)$ to $O(M)$ \cite{lee2021even,caesura2025faster}.
A natural question is then  whether QROM might similarly benefit Trotter methods. In this subsection we will show that the use of QROM provides a moderate constant factor improvement (between $\times 1.25$ and $\times 1.9$) in the overall Toffoli cost of product formulas.

\textbf{Compiling $\sigma_z\otimes \sigma_z$ rotations with QROM:}
In rank-factorized Trotter formulas such as CDF, isometric THC, or GRADE, fragments consist of two main components: (i) basis rotations, which are computed via QR decomposition and applied independently to each spin sector using Givens rotations, see~\cref{eq:basis_rotation,eq:Givens}; and (ii) $\sigma_z\otimes\sigma_z$ rotations, which account for the largest computational cost. This is depicted in~\cref{fig:CDF_circuit}.
Specifically, the basis rotations in the $\ell$-th fragment require \cite{luo2025efficient}
\begin{align}\label{eq:cost_partial_basis_rotation}
  2\binom{M_\ell}{2} - 2\binom{M_\ell-N}{2} = 2M_\ell N - N^2 - N 
\end{align}
Givens rotations, each of which can be implemented with 2 $(b-1)$-bit rotations using the Gidney adder \cite{caesura2025faster}, depicted in~\cref{fig:givens_fused_adder}. In contrast, the $\sigma_z\otimes\sigma_z$ rotation block requires
\begin{equation}
  \frac{2M_\ell(2M_\ell-1)}{2} = 2M_\ell^2 - M_\ell
\end{equation}
single qubit rotations -- any Pauli string rotation might be mapped to Clifford gates plus a single qubit rotation. 

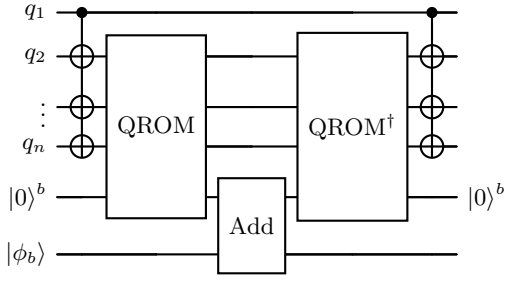
\begin{figure}[t]
\centering
\adjustbox{max width=\columnwidth}{
\begin{quantikz}[row sep=0.6em, column sep=0.5em]
\lstick{$q_1$}                     & \ctrl{1}       & \qw                                       & \qw                    & \qw                                          & \ctrl{1}       & \qw \\
\lstick{$q_2$}                     & \targ{}\vqw{2} & \gate[4]{\mathrm{QROM}}        & \qw                    & \gate[4]{\mathrm{QROM}^{\dagger}} & \targ{}\vqw{2} & \qw \\
\lstick{$\vdots$}                  &   \targ{}             &                                            &                        &                                               &      \targ{}          & \\
\lstick{$q_n$}                     & \targ{}        &                                            & \qw                    &                                              & \targ{}        & \qw \\
\lstick{$\ket{0}^{b}$}            & \qw            &                                            & \gate[2]{\mathrm{Add}} &                                              & \qw            & \rstick{$\ket{0}^{b}$}\qw \\
\lstick{$\ket{\phi_b}$}           & \qw            & \qw                                        &                        & \qw                                          & \qw            & \qw
\end{quantikz}
}
\caption{Quantum circuit for the QROM-based compilation of the $\sigma_z \otimes \sigma_z$ rotation block $\prod_{i<j} \exp(i\theta_{ij} \sigma_{z,i} \sigma_{z,j})$ acting on a group of $n$ system qubits. A fan-out of CNOTs from $q_1, q_2,\ldots ,q_n$ canonicalizes each computational basis state so that $q_1=0$, exploiting the parity symmetry $\phi(\bar{z}) = \phi(z)$ to halve the QROM address space from $2^n$ to $2^{n-1}$. The QROM loads the $b$-bit cumulative phase angle into an ancilla register, which is then added to a phase gradient state $\ket{\phi_b}$ via a Gidney adder circuit. After uncomputing the QROM and reversing the CNOTs, the ancilla returns to $\ket{0}^b$. The total Toffoli cost is C$_\text{QROM}$(n,b), yielding a $\times 2 - 4$ improvement over the baseline cost of $bn(n-1)/2$ per group.}
\label{fig:QROM_circuit}
\end{figure}

\begin{table}[t]
  \centering
  \begin{tabular}{c|c|c|c}
  $b$ & $\epsilon_r$ & $n$ & $C_{\text{QROM}}(n,b) / C_{\text{baseline}}(n,b)$\\
  \hline\hline
   10  &  $10^{-3}$ & 6 & 0.52\\
   15  & $3\times 10^{-5}$ & 7 & $0.295238$\\
   20 & $10^{-6}$ & 7 & $0.257143$
  \end{tabular}
  \caption{Potential savings from using QROM to implement the rotations corresponding to the $\sigma_z\otimes \sigma_z $ Pauli strings. The values are $k = 2$ and $k' = 8$ in all three cases.}
  \label{tab:QROM_savings}
\end{table}

All $\sigma_z\otimes\sigma_z$ rotations commute, so we may implement them in any order. If magic states are abundant, both the $\sigma_z\otimes\sigma_z$ blocks and the basis rotations can be parallelized, reducing the circuit depth to $O(N)$ vs the total Toffoli gate cost of $O(N^2)$. In this manuscript we count Toffoli gates, so we instead propose a strategy to leverage QROMs to lower the Toffoli count of the $\sigma_z\otimes\sigma_z$ blocks. 

The key observation is that the product of all $\sigma_z\otimes\sigma_z$ rotations within a group of $n$ qubits is diagonal in the computational basis. On a basis state $\ket{z} = \ket{z_1 \cdots z_n}$, each Pauli operator satisfies $\sigma_{z,i}\ket{z_i} = (-1)^{z_i}\ket{z_i}$, so the full rotation block acts as
\begin{align}\label{eq:cumulative_phase}
 \prod_{i<j} e^{i\theta_{ij}\,\sigma_{z,i}\sigma_{z,j}} \ket{z} &= e^{i\phi(z)}\ket{z}, \\
 \text{where } \phi(z) \;=\; &\sum_{i<j} \theta_{ij}\,(-1)^{z_i + z_j}.
\end{align}
That is, each computational basis state $\ket{z}$ accumulates a single cumulative phase $\phi(z)$ that is fully determined by the bitstring $z$.
Consequently, on an arbitrary superposition $\ket{\psi} = \sum_z c_z \ket{z}$, the rotation block maps $\ket{\psi} \mapsto \sum_z c_z\, e^{i\phi(z)}\ket{z}$. Our compilation strategy then is as follows: rather than implementing $n(n{-}1)/2$ individual rotations, we precompute the $2^n$ values $\{\phi(z)\}$ and store them in a QROM that acts as a lookup table for the overall phase to implement. The circuit, depicted in~\cref{fig:QROM_circuit}, proceeds as follows: (i) the QROM, addressed by the system qubits $\ket{z}$, loads the $b$-bit representation of $\phi(z)$ into an ancilla register; (ii)~the ancilla is added to a phase gradient state $\ket{\phi_b}$ via a Gidney adder, which imprints the phase $e^{i\phi(z)}$ on $\ket{z}$ \cite{gidney2018halving}; and (iii) the QROM is uncomputed, returning the ancilla to $\ket{0}^b$. The Gidney adder uses a phase gradient state
\begin{equation}
    \ket{\phi_b} = \frac{1}{\sqrt{2^b}}\sum_{k=0}^{2^b-1} e^{-2\pi i k / 2^b}\ket{k}
\end{equation}
to implement a $\sigma_z$ rotation.
This can be seen by noting that adding an integer to the phase gradient state register results in a phase via phase kickback implementing the desired rotation:
\begin{multline}
    \ket{n}\ket{\phi_b} \mapsto
    \ket{n}\sum_{k=0}^{2^b-1} e^{-2\pi i k / 2^b}\ket{k-n} =\\ \ket{n}\sum_{k=0}^{2^b-1} e^{-2\pi i (k+n) / 2^b}\ket{k}
    = e^{-2\pi i n / 2^b}\ket{n}\ket{\phi_b}.
\end{multline}

\textbf{Exploiting symmetries:}
More generally, let us assume we want to implement one set of $n(n-1)/2$ rotations acting on $n$ qubits to $b$ bits of precision. We will now compare how attractive is this method. The Toffoli cost with the Gidney adder would be \cite{gidney2018halving}
\begin{equation}
  C_{\text{baseline}}(n,b) = b \frac{n (n-1)}{2}.
\end{equation}
Using a QROM we need to consider the $2^n$ possible states of the $n$ qubits. The cost would be \cite{low2024trading}
\begin{equation}
  C'_{\text{QROM}}(n,b) = \min_{k, k'}\left(\frac{2^n}{k} + \frac{2^n}{k'} + 2b(k-1) + k' + b\right),
\end{equation}
where $k$ and $k'$ are powers of 2 chosen to minimize the cost. 

As illustrated in the example above, we can improve this estimate: since the $\sigma_z\otimes\sigma_z$ rotations only depend on the relative parity of the input bitstrings, the phases to be implemented are duplicated. We can leverage this information to reduce the size of the QROM. For instance, $z = 001$ would implement the same phases as $\bar{z}=110$. Thus, we can reduce the input cost from $2^n$ possible bitstrings to $2^{n-1}$ by forcing each bitstring into a canonical form, where the first bit is $0$. This can be done with CNOT gates from the first bit to the rest, see~\cref{fig:QROM_circuit}. As a consequence, the cost of this method would be
\begin{equation}
  C_{\text{QROM}}(n,b) = \min_{k, k'}\left(\frac{2^{n-1}}{k} + \frac{2^{n-1}}{k'} + 2b(k-1) + k' + b\right).
\end{equation}
We can study how much more efficient this QROM method is compared to the baseline. In~\cref{tab:QROM_savings} we estimate that the QROM-based implementation of the rotations may amount to $\times 2$ to $\times 4$ fewer Toffoli gates in the $\sigma_z\otimes \sigma_z$ block of rotations. 

\textbf{Overall savings from QROM use:}
The remaining question is how many of these two-qubit rotations we can package in sets without repetitions. An upper bound is given by the Sch\"{o}nheim bound \cite{schonheim1966maximal}. The number of sets of size $n$ that one can create is
\begin{equation}
  m(N,n,2) \leq \left\lfloor\frac{N}{n}  \left\lfloor\frac{N-1}{n-1} \right\rfloor \right\rfloor.
\end{equation}
While the tightness of this bound depends on the specific values of $N$ and $n$, it is also known that \cite{rodl1985packing}
\begin{equation}
  \lim_{N\rightarrow\infty} \frac{ m(N,n,2) \binom{n}{2}}{\binom{N}{2}} = 1,
\end{equation}
so when $N$ grows only a small fraction of two-qubit rotations are not packed into those sets.

Put together, the compilation tricks mentioned in this article improve the Toffoli cost of a Trotter step by a factor between $\times 1.25$ and $\times 1.9$, depending on the values of $b$, $M_\ell$ and $N$. This estimate is obtained by combining the number of rotations, with cost reductions in the $\sigma_z\otimes\sigma_z$ rotation block in~\cref{tab:QROM_savings}.

\section{How to construct the SPRINT formula\label{sec:building_pf}}

Having described in detail all the elements of the SPRINT framework shown in \cref{fig:SPRINT_product_formula}, we now present in concrete terms a guideline to construct the SPRINT product formula when applied to a given Hamiltonian and simulation task. Specifically, we provide a step-by-step guide for evaluating whether to include each technique described in \cref{sec:Product_formulas}, which variation of it to select, and how to ensure the techniques chosen play to each other's strengths. For this guide, we assume that the prospective user has as input a second-quantized chemical Hamiltonian $H$ and a specific simulation task with error budget $\epsilon$.

\textbf{Step 1 -- Factorize the Hamiltonian:} As we saw in~\cref{sec:Factorization}, the first step of any product formula design is to split the Hamiltonian into fast-forwardable fragments. To start, we recommend performing the CDF, isometric THC, and GRADE factorizations, as described in \cref{ssec:grade-compute}. It is necessary to loop over the hyperparameters until the user's error criterion is satisfied: for ground-state energy, this would be minimizing the difference in the ground state energies of the original $H$ and factorized $H'$, according to a classical method like DMRG. For excited-state problems, it would be the same but for a range of eigenstates; it may be different for generic dynamics simulation problems. 

\textbf{Step 2 -- Obtain a baseline Trotter error:} Once a satisfactory factorization is obtained, the next step is to evaluate the accompanying Trotter error and leakage error (if applicable, see \cref{ssec:grade-leakage}), using methods such as those of Ref. \cite{Xanadu_Trotter_error_paper}, for a simple first-order Trotter step. For GRADE and isometric THC, one should implement symmetry protection as described in \cref{ssec:Symmetry_protection} and then use Trotter error estimation tools to evaluate how much leakage error can be suppressed. Combining both the Trotter + leakage error and the per-step costs of the factorization (the latter may be evaluated as in Ref. \cite{fomichev2025fast} and \cref{ssec:QROM}), it is possible to calculate the overall cost of performing time evolution for some default time step $\delta$. From these, it should be clear which of the factorizations offers the best balance of being able to take larger step sizes together with the lowest per-step cost.

\textbf{Step 3 -- Evaluate norms, group fragments:} Once the Hamiltonian is factorized, we evaluate the norms of the fragments, a proxy for their contributions to the error, using a method such as described at the end of \cref{ssec:grade-compute}. As \cref{fig:Zs} suggests, the one-body fragment and the first two-body fragment typically have a norm much larger than the rest of the fragments, so we suggest they form group $H_A$, while the rest can be group $H_B$. Finally, the Pauli strings in the remainder and any error mitigation terms form $H_C$, which acts as the middle fragment of any symmetric formula for $H_B$. Depending on the specific norm structure, for example how quickly it decays, it may be advantageous to re-allocate the fragments between the groups $A$ and $B$. This may be decided using the process in the next step of devising the near-integrable formula.

\textbf{Step 4 -- Choose near-integrable formula:} With the groups $H_A$ and $H_B$ identified, one should explore which near-integrable formula to use. This first amounts to choosing which Trotter product formulas -- first, second, fourth, sixth order and so on -- to use on groups $A$ and $B$, and even $C$. Then, one needs to select how to compose individual steps as if we were designing a product formula out of $H_A$ and $H_B$. 
Many options are available, as summarized in~\cite{blanes2024splitting}.
The near-integrable formulas we propose are presented in \cref{ssec:Near_integrability}, with the two common choices being (i) the second-order near-integrable formula $V_{2,1}$ as defined in \cref{eq:V_21}, where we use second-order Trotter on group $A$ and first-order Trotter on group $B$, and (ii) the fourth-order near-integrable formula $V_{4,2}$, where we use fourth order on $A$ and second order on $B$. To decide which variant to use, one should first perform a BCH expansion to obtain the expression for the effective Hamiltonian, and then leverage a Trotter error estimation tool such as that of Ref. \cite{Xanadu_Trotter_error_paper}. Such a tool will be very useful to allow straightforward comparisons between different product formula variations. The estimated Trotter error should be compared with the simulation task's error requirement $\epsilon$: the goal is to find the cheapest variant of near-integrability that satisfies the desired error budget. In spectroscopy-like applications, as a rule of thumb, second order formulas are best for accuracies of $1$ eV, while fourth order formulas are typically better suited for chemical accuracy (around $1$ mHa).

\textbf{Step 5 -- Evaluate processing:} Assuming that a near-integrable formula of fourth order or higher was chosen, it can be useful to consider processing to cancel some of the leading error terms. These may be constructed by using the BCH-derived expression for the effective Hamiltonian of the near-integrable product formula chosen in the previous step and the methods described in \cref{ssec:Processing}. It will be important to once again leverage the Trotter error estimation tool to evaluate the magnitude of the contribution of the term being canceled to the overall Trotter error. One should only deploy the processor if the error turns out to be substantial, such that the savings from being able to take larger Trotter steps outweigh the cost of implementing the processor. 

\textbf{Step 6 -- Randomize the fragment order:} In general, for any product formula, we advise to use randomization -- to randomize the order of the fragments within the groups $A$ and $B$ in each Trotter step -- to further reduce Trotter error at zero additional cost, as described in \cref{ssec:randomization} and \cref{ssec:trotter_error_results}. As will be seen in \cref{sec:Applications}, we find this useful especially for second order formulas. Once randomized, we advise using Trotter error estimation to evaluate the benefit, to lock down the final allowable Trotter step size. Randomization introduces some spectral line broadening, but such an effect is weaker than the Trotter error peak shift, see~\cref{app:randomized_spectra}.

\textbf{Step 7 -- Select a remainder strategy:} Having tackled the $A$ and $B$ fragment groups, it is time to decide on the implementation strategy for the remainder group $H_C$. If we want to consider using error mitigation, we include such Hamiltonian simulation fragments on $H_C$. The decision on using error mitigation will require assessing the implementation costs of two main strategies. The first involves combining random compilation, such as qDRIFT or Randomized Taylor Expansion~\cite{campbell2019random, wan2022randomized, loaiza2023block}. This strategy would require mapping the nested commutators to be uncomputed to Pauli strings, and symmetry shifts might be used to reduce the one-norm of the resulting operator~\cite{loaiza2023block}. Alternatively, one may use unitary MPOs to directly cancel out the effect of specific error terms in the effective Hamiltonian~\cite{termanova2024tensor,nibbi2024block,guo2022quantum}. The use of these techniques will ultimately depend on the error reduction versus cost increase balance. Partially fault-tolerant rotations -- the Space-Time efficient Analog Rotations (STAR) technique -- could similarly be considered to mitigate the compilation overhead associated with very small-angle rotations~\cite{akahoshi2024partially,akahoshi2025compilation,chung2026partially,kanasugi2026enabling}.

\textbf{Step 8 -- Consider using QROM:} Finally, it is worth evaluating the prospect of using QROM in batch implementing rotations, as described in \cref{ssec:QROM}. As randomization, we find it is worthwhile in most cases, even though the benefit varies slightly depending on the Hamiltonian factorization used.

Following this guide, it is possible to determine the best combination of methods for the specific Hamiltonian and simulation task being faced. In the next section, we follow this guide in tackling one specific problem from the literature -- the task of computing the X-ray absorption spectrum of the molecular cluster Li$_4$Mn$_2$O \cite{fomichev2024simulating,fomichev2025fast}.

\section{Application: X-ray absorption spectra of batteries\label{sec:Applications}}

With the SPRINT framework and GRADE factorization approach defined, we apply these techniques to the task of computing the X-ray absorption spectrum (XAS) of the Li$_4$Mn$_2$O cluster, a model system for lithium-excess battery cathode materials~\cite{fomichev2024simulating,fomichev2025fast}. Accurate XAS calculations for such clusters can elucidate degradation mechanisms in Li-excess cathodes and thus help unlock next-generation high-energy-density batteries \cite{house2023delocalized,radin2019manganese}. Comparing to the resource estimates for quantum algorithms for XAS obtained previously \cite{fomichev2025fast} presents an opportunity to evaluate the performance of our new GRADE and SPRINT techniques in practice on an industry-relevant challenge. 

In summary, we report three main findings:
\begin{enumerate}
\item On the Li$_4$Mn$_2$O system, GRADE achieves lower per-step gate costs and scaling than CDF. It requires half the qubits needed by isometric THC. However, if not addressed properly, GRADE suffers from higher Trotter and leakage errors than CDF. In~\cref{app:GRADE_error} we study how to mitigate this, with a mixed result: we managed to significantly reduce the Trotter error, but not the leakage error. We find CDF to be the lowest-cost factorization for a Trotterized simulation approach in our target problem and system, even though this will likely differ for different simulation problems. 
\item We empirically found that near-integrability (\cref{ssec:Near_integrability}), randomization (\cref{ssec:randomization}) and compilation (\cref{ssec:QROM}) are the main methods that deliver Toffoli cost savings for the Li-excess XAS calculation. 
%In particular, the near-integrable formulas of SPRINT can deliver Toffoli gate savings of as much as $\times 2$ or $\times 5$ at comparable error, depending on the product formula order and the desired evolution time. 
Empirically, near-integrability, randomization and compilation yield saving factors of approximately $\times 1.6$, $\times 1.5$ and $\times 1.4$, respectively, see~\cref{fig:SPRINT_Towers}. 
\item The overall SPRINT pipeline yields a $\times 4.5$ Toffoli reduction over the previous state of the art \cite{fomichev2025fast}, and is only $\times 2.5$ more costly than qubitization, while using a dramatic $\times 5.5$ fewer logical qubits. The different factors contributing to this speedup can be seen in~\cref{fig:SPRINT_Towers}, and the final gate counts are described in~\cref{tab:resources}.
\end{enumerate}

In the following subsections we present these findings in more detail. We briefly introduce the X-ray absorption spectroscopy problem for Li-excess and sketch the construction of the Hamiltonian in \cref{ssec:xas_intro}, with additional details available in \cite{fomichev2024simulating,fomichev2025fast} and in \cref{app:hamiltonian_construction} respectively. With the problem and system defined, in~\cref{ssec:grade_results} we report the results related to rank factorization, comparing GRADE with CDF and isometric THC. In~\cref{ssec:trotter_error_results}, after sketching the computational approach to evaluating the Trotter error that we adopt from Ref. \cite{Xanadu_Trotter_error_paper}, we leverage it to compare the performance of different product formulas, specifically the various combinations of the methods presented in~\cref{sec:Product_formulas}, identifying the optimal combination of SPRINT techniques for the Li-excess problem. We close with a cost comparison between Trotter and qubitization in~\cref{ssec:Qubitization}. 

\subsection{X-ray absorption of Li-excess clusters}
\label{ssec:xas_intro}

\begin{figure*}
  \centering
  \includegraphics[width=0.4\linewidth]{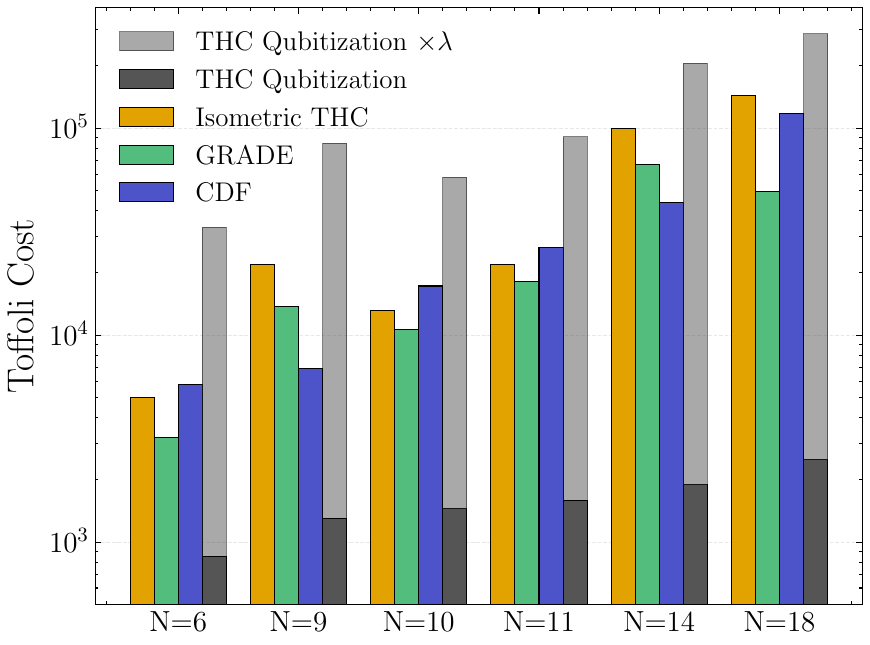}
  \includegraphics[width=0.4\linewidth]{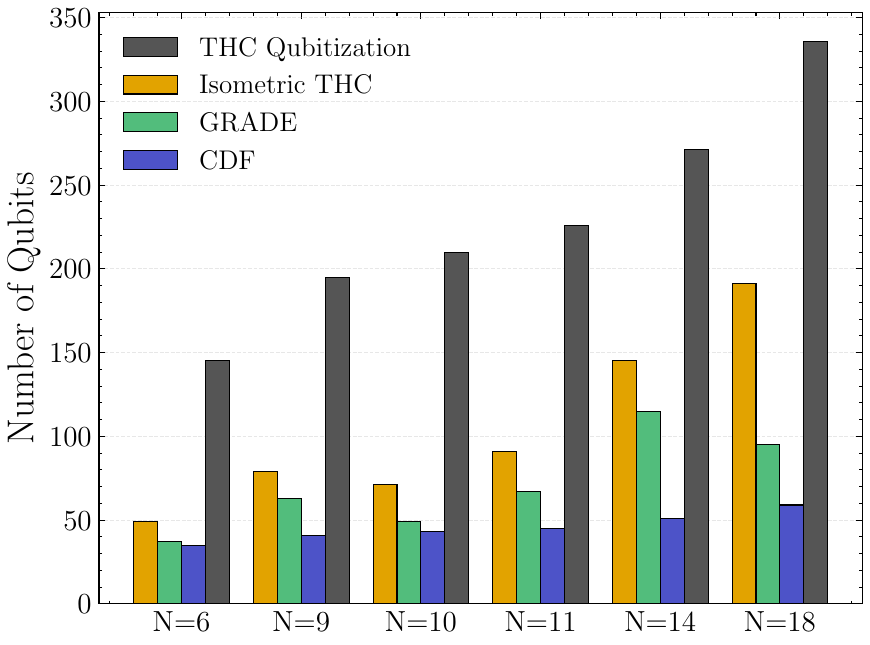}
  \caption{Toffoli gate cost and qubit cost to implement one first-order Trotter step, for different factorizations and active space sizes for Li$_4$Mn$_2$O, targeting a final precision of 1 eV $= 3.67\cdot 10^{-2}$ Ha in peak shift. At this precision, the Toffoli gains are small, though the GRADE factorization exhibits better costs than the alternatives in most cases. All factorizations include optimizing over hyperparameters to get the mean absolute energy error of the first 10 eigenvalues below 1 eV. The qubit counts include a phase gradient state of $b=15$ bits.
  }
  \label{fig:GRADE_per_step_cost_1eV}
\end{figure*}

While SPRINT and GRADE are a general framework for implementing time evolution under a Hamiltonian, in this paper we specifically evaluate them on the task of computing X-ray absorption spectra introduced in Refs. \cite{fomichev2024simulating,fomichev2025fast}. Specifically, for the Li$_4$Mn$_2$O cluster, it is shown in Ref. \cite{fomichev2025fast} that the X-ray absorption spectrum $\sigma(\omega)$ for incoming X-ray frequency $\omega$ is given by
\begin{equation}
  \sigma(\omega) \approx \frac{\eta\delta}{2\pi} \sum_{j=-\infty}^{\infty} e^{-\eta\delta|j|} \tilde{G}_\rho(\delta j) e^{ij\delta\omega},
  \label{eq:green-via-discretetime-fourier}
\end{equation}
where $\tilde G_\rho(\delta j)$ is the time-domain Green's function for the dipole Cartesian component $\rho$ at time $t_j = \delta j$ for integer $j$, given by
\begin{equation}\label{eq:G(t)}
  \tilde{G}_\rho(\delta j) = \frac{\bra{I} m_\rho e^{-i  H \delta j}  m_{\rho} \ket{I}}{\| m_\rho \ket{I} \|^2}.
\end{equation}
Here $m_\rho$ is the dipole operator, $\ket{I}$ is the cluster ground state with energy $E_I$, and $\eta$ is the broadening parameter. 
The matrix element for each time slice $\delta j$ can be evaluated using the Hadamard test on the unitary $e^{-iH\delta j}$. The discrete timestep $\delta$ is chosen so that all eigenvalues within the support of the initial state $m_{\rho}\ket{I} = \sum_\alpha c_\alpha \ket{F}$ are rescaled into the range $[-\pi, \pi)$. The discrete time step $\delta$ is not in general equal to the Trotter timestep $\tau$: rather, typically $\delta = n \tau$ for some integer $n$. 

The Hamiltonian $H$ to be constructed needs to represent the key electronic orbitals governing the absorption response of a molecular cluster surrounding the absorbing atom. In this study, we obtain an oxygen-centered Li$_4$Mn$_2$O cluster similarly to Refs. \cite{fomichev2024simulating,fomichev2025fast} by extracting it from the crystal structure of the Li-excess cathode material Li$_2$MnO$_3$. We then run a Hartree-Fock calculation using PySCF \cite{sun2015libcint,sun2018pyscf,sun2020recent} to get a basis of molecular orbitals, from which we construct a series of active spaces using the automated valence active space (AVAS) method \cite{sayfutyarova2017automated}, and then finally apply the core-valence separation approximation to the resulting $H$ \cite{cederbaum1980many,barth1981many,norman2018simulating,herbst2020quantifying}. Further details on Hamiltonian construction are provided in \cref{app:hamiltonian_construction}. We stress that while from a modeling accuracy perspective there are likely ways to improve the procedure we pursued, in this study it serves mainly as a way of generating a realistic, chemically motivated, and industrially relevant sequence of model Hamiltonians on which we can evaluate the performance of SPRINT and GRADE. A serious effort to generate quantitatively accurate models for the Li-excess X-ray problem is an important task that we intend to pursue in future work.

Having defined a concrete simulation task, we now evaluate how SPRINT and GRADE perform on it.

\subsection{Results: GRADE\label{ssec:grade_results}}

The total cost of time evolution for spectroscopy applications such as XAS is the product of the per-step cost and the number of product formula steps. Given this, we start by analyzing the per-step gate cost of different factorizations described in~\cref{ssec:GRADE}, namely CDF, isometric THC, and finally GRADE itself.

Specifically, using the Hamiltonians of varying active spaces sizes from $N = 6$ to $N = 28$ spatial orbitals built as described in \cref{ssec:xas_intro}, we implement the rank factorization procedure for CDF, isometric THC, and GRADE: in all cases, we run a meta-optimizer over different hyperparameter settings for the Pauli threshold $\bar{t}$ and the fragment numbers and dimensions $L$, $M_\ell$, and $M$. For the criterion of whether we accept a factorization -- which depends on the average deviation between the 10 lowest eigenvalues of the factorized and original Hamiltonian, as computed with DMRG -- we use the value of 1 eV, inspired by the resolution in typical XAS experiments \cite{fomichev2025fast}.  To ensure the factorization can in principle achieve such accuracy, we fixed the rotation precision to $b = 15$ bits uniformly across all methods, since that is the minimum number of bits required to reach our chosen target accuracy. We used PennyLane functionality to get the CDF decomposition of the Hamiltonian \cite{bergholm2018pennylane}; a tutorial can be found in \cite{UtkarshAzad2025how}. For isometric THC and GRADE, we developed custom code.
For all factorizations, the meta-optimizer was run in a `greedy' setting, where the goal was to achieve the smallest per-step gate cost. Once we obtain the $U^{(\ell)}$ and $Z^{(\ell)}$ matrices, we count the number of Givens rotation and $\sigma_z\otimes \sigma_z$ rotations, implemented as described in~\cref{ssec:QROM}. 

\begin{figure*}
  \centering
    \begin{overpic}[width=0.43\linewidth, percent]{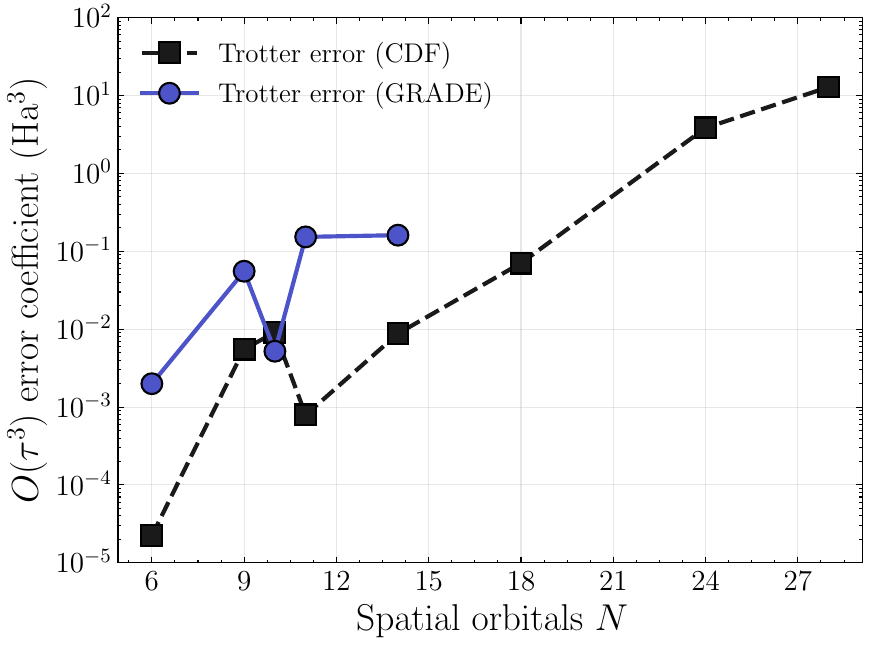}
    \put(0,75){(a)}
    \end{overpic}
    \begin{overpic}[width=0.43\linewidth, percent]{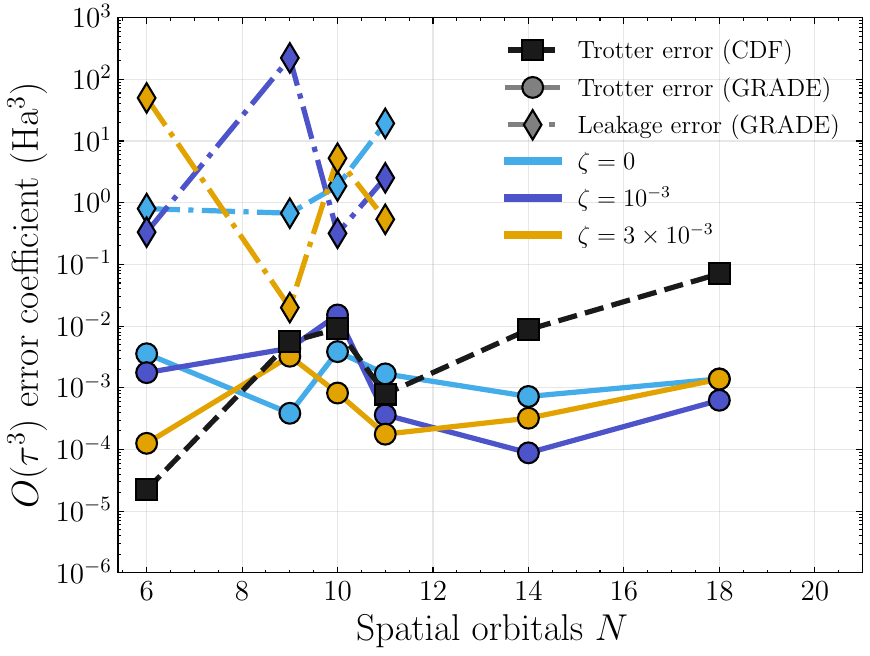}
    \put(0,75){(b)} 
  \end{overpic}
  \caption{Magnitude of the leading order coefficient of Trotter error of CDF (dashed, black), Trotter error of GRADE (solid), and leakage error of GRADE (dot-dashed), for a second order Trotter formula. 
  %Trotter error was computed with bond dimension 100 when $N\leq 18$, and the leakage error used the same bond dimension when $N\leq 9$, since $\max_\ell M_\ell = 2N$. In the largest cases, we instead used bond dimension 25.
  (a) Trotter error comparison of CDF vs GRADE, the latter `greedy-optimized' to just minimize the per-step Toffoli cost without regard to Trotter error or leakage. (b) Trotter error of CDF vs GRADE, with the latter optimized to reduce Trotter error using manually selected schedules $M_\ell$ described in~\cref{tab:grade_schedules}, as well as the leakage error of GRADE. We depict several lines corresponding to choices of regularization loss weights $\zeta$ for the number of auxiliary orbitals and Frobenius norm penalty of $\|Z^{(\ell)}\|$ and the leakage proxy in~\cref{eq:cheap_leak}. Manually designed schedules successfully decrease the Trotter error to CDF levels, but still display a large leakage, preventing GRADE outperforming CDF.}
  \label{fig:CDF_vs_GRADE}
\end{figure*}

The results of counting Toffoli costs and total qubit requirements for the different factorizations to implement a single first-order Trotter step are shown in \cref{fig:GRADE_per_step_cost_1eV}. For reference, we also include a comparison with qubitization, even though the notion of ``per-step'' is somewhat different.  The comparison we make here is twofold: first, we compare our Trotter per-step estimates to the cost of the qubitization quantum walk block-encoding (dark grey); and second, we also plot that block-encoding cost multiplied by the one-norm of the Hamiltonian, which sets the energy scale of the system (light grey). We found that qubitization generally scales better than Trotter on a per-step basis; however, the one-norm prefactor penalizes it sharply, making its gate count advantage over Trotter slim for the small system sizes considered. We will revisit the Trotter versus qubitization comparison in more detail for the full time evolution task in~\cref{ssec:Qubitization}. 

From the factorization results for CDF, isometric THC and GRADE, we see that in practice GRADE usually achieves better per-step costs than alternative factorizations across a range of active spaces for the Li-excess cluster. Impressively, GRADE consistently achieves much better qubit costs than isometric THC, providing evidence that GRADE may be strictly better than isometric THC \cite{luo2025efficient}, at least on a per-step basis. Moreover, on theoretical grounds, for large $N$ isometric THC and GRADE should outperform CDF: their per-step cost scales as $O(N^2)$, versus $O(N^3)$ for CDF. 

However, this advantage was undermined by the fact that GRADE in general tends to have significantly larger Trotter errors than CDF, as well as a new source of error -- leakage to auxiliary orbitals. In the left panel of~\cref{fig:CDF_vs_GRADE} we show on the representative example of the second-order Trotter formula that the Trotter error of GRADE can be over an order of magnitude larger than that of CDF. To see this, we compute and plot the expectation value of the leading order nested commutator error term $\bra{E_l} Y_k \ket{E_l}$, averaged across the lowest 5 eigenstates $\ket{E_l}$: this calculation method for getting an estimate of Trotter error, based on perturbation theory, is explained in the following \cref{ssec:trotter_error_results}. 

In an attempt to reduce the Trotter error by modifying hyperparameters in the GRADE factorization, we pursued a hand-designed $M_\ell$ schedule with $\|Z^{(\ell)}\|$ regularization. 
We found it was possible to substantially reduce the Trotter error of GRADE with a well selected $\{M_\ell\}_{\ell}$ schedule. Specifically, we could lower the Trotter error to essentially the level of CDF, as seen in the right panel of \cref{fig:CDF_vs_GRADE}, while preserving the per-step cost advantage (see the right panel of \cref{fig:Trotter_error_GRADE} in \cref{app:GRADE_error}). 

But while Trotter error could be tamed through optimization of hyperparameters, the leakage error proved much harder to mitigate. In the same right panel of \cref{fig:CDF_vs_GRADE} we also plot the leakage error that obtains in a symmetry-protected second-order Trotter formula. We reduced the leakage error using two strategies. First we added penalty terms $\|Z^{(\ell)}\|_F$ and
\begin{equation}\label{eq:cheap_leak}
  E_{\text{leak}}^{(\ell)} = \|P^{(\ell)} Z^{(\ell)}(I-P^{(\ell)}) + (I-P^{(\ell)})Z^{(\ell)}P^{(\ell)}\|,
\end{equation}
to the GRADE fitting loss function. Here, we define projectors $P^{(\ell)}$ for each fragment $\ell$ using the $M_{\ell} \times N$ isometry $V^{(\ell)}$, implementing part of the basis rotation $U^{(\ell)}$:
\begin{equation}
  P^{(\ell)} = V^{(\ell)} V^{(\ell)\dagger}.
\end{equation}
Second, we used symmetry protection to reduce the leading order of the error to $O(\tau^2)$ (see \cref{ssec:Symmetry_protection}). Thanks to this, the leakage error is on the same footing as the Trotter error, and so the two error sources can be compared directly in the same plot. The leakage error is computed similarly to the Trotter error, by evaluating the expectation value $\bra{E_l} Y_\text{leak}^{(2)} \ket{E_l}$ of the leading nested commutator leakage error term (defined in \cref{eq:leakage_error_symmetry_protected}) with respect to the ten lowest approximate eigenstates, and then averaging. From these results we see that the leading order leakage error remains two or more orders of magnitude higher than the Trotter error, and thus is the main driver of total simulation cost. This means that it is not obvious GRADE is preferable to CDF for the systems studied here, as any per-step cost gains are more than offset by the need to take smaller steps to control the leakage error. Additional discussion of the Trotter and leakage errors of GRADE is presented in \cref{app:GRADE_error}. Given these results, in the next section, we only analyze the performance of SPRINT product formulas using the CDF factorization.

\subsection{Results: SPRINT\label{ssec:trotter_error_results}}

After per-step cost, the second aspect that determines the total simulation cost of time evolution is the required number of Trotter steps. Ultimately, this number of Trotter steps will come from the requirement to satisfy a particular application error budget. Spectroscopies such as XAS typically require us to recover the spectrum of Hamiltonian eigenvalues. For this reason, to impose an application-defined error budget $\epsilon$, we will use the error of reproducing the spectrum -- specifically, shifts in peak positions of key eigenstates $|E_l' - E_l| < \epsilon$ \cite{fomichev2025fast}. This is as opposed to, say, directly controlling the accuracy of the evolution unitary $|e^{-iHt} - U(t)|<\epsilon$, which is how Trotter error is commonly studied \cite{childs2021theory}. The error budget will then determine the required number of Trotter steps and thus the total cost of the simulation task. 

\textbf{Estimating Trotter error:} In this manuscript we will be \textit{empirically estimating} the Trotter error rather than bounding it. Estimating the shifts in peak positions requires three key steps. First, because product formulas implement exact time evolution under an approximate Hamiltonian, we use the BCH expansion to derive the effective Hamiltonian for a given product formula, retaining only the leading-order nested commutators and dropping all higher-order contributions. Second, we treat those leading order nested commutators as \textit{perturbations} to the true Hamiltonian, and use perturbation theory to describe the effect of those nested commutators on the Hamiltonian spectrum, i.e. the difference $|E_l' - E_l|$ between the true eigenvalues $E_l$ and those of our effective Hamiltonian $E_l'$ \cite{mehendale2025estimating,fomichev2025fast}. This approach to estimating Trotter error is well-established, including for spectroscopy applications \cite{mehendale2025estimating,fomichev2025fast}: given a $p$-th order product formula $U_p(\tau)$ implementing the effective Hamiltonian $H_\text{eff} = H + \tau^p Y_{p+1} + O(\tau^{p+1})$, with $Y_{p+1}$ being the leading-order error operator, by perturbation theory for small enough $\tau$, the eigenvalues and eigenstates of the effective Hamiltonian satisfy
\begin{align}
  E'_l &= E_l + \tau^{p}\braket{E_l|Y_{p+1}|E_l} + O(\tau^{p+1}), \label{eq:perturbation_eigenvalues}\\
  \ket{E'_l} &= \ket{E_l} + \tau^{p} \sum_{k\neq l} \frac{\braket{E_k|Y_{p+1}|E_l}}{E_k-E_l}
  \ket{E_k}+O(\tau^{2p}). \label{eq:perturbation_eigenstate}
\end{align}
Consequently, the Trotter error manifests as a coherent shift in the spectral peak positions and a redistribution of spectral weight. Here we focus only on the peak shifts and leave the estimation of spectral weight changes to future work.
With the expressions for Trotter error in hand, we leverage the Trotter error estimation software of Ref. \cite{Xanadu_Trotter_error_paper} to evaluate them. This software is what ultimately allows us to estimate Trotter error for much larger systems than previously considered: this is achieved through a combination of using matrix product states (MPS) for the approximate eigenstates and matrix product operators (MPOs) for the error nested commutators, as well as a number of additional techniques, including approximate norm-ordering and importance sampling of nested commutators.

\textbf{SPRINT improvements:}
Armed with this perturbative approach to estimating peak shift positions of any product formula ansatz, we now explain how the SPRINT framework reduces the Toffoli gate cost relative to the previous state of the art on the example of XAS simulation for the Li-excess cluster \cite{fomichev2025fast}. The cumulative effect of all techniques was already shown in~\cref{fig:SPRINT_Towers}: here we unpack each contribution individually. 

To estimate the impact of each technique, we perform resource estimation of a spectroscopy simulation. Specifically, we count the logical qubits and Toffoli gates needed to carry out both the deepest circuit from among all the Hadamard test calculations we need to perform, and also the combination of all circuits of length $t = \delta, 2\delta, ... j_\text{max} \delta$ needed to build the spectrum in \cref{eq:green-via-discretetime-fourier}. This follows exactly the procedure described in Ref. \cite{fomichev2025fast}: we provide some additional details in \cref{app:resource_estimation}. Essentially, this amounts to determining the total number of Trotter steps needed for the simulation on the basis of the estimated Trotter error, and multiplying it by the respective per-step cost already shown in \cref{fig:GRADE_per_step_cost_1eV}.
We emphasize that these are \textit{not} worst-case upper bounds: the Trotter error coefficients are \textit{estimated} numerically for the specific Hamiltonian and eigenstates of interest. Having these Trotter error estimates, rather than loose error bounds, allows for a fairer comparison of the cost of Trotter-based time evolution against qubitization, beyond just the per-step gate cost comparison. With this approach, we now show how much each of the product formula techniques contributes to the overall cost reduction: 

\begin{figure}[t]
  \centering
  \includegraphics[width=0.8\linewidth]{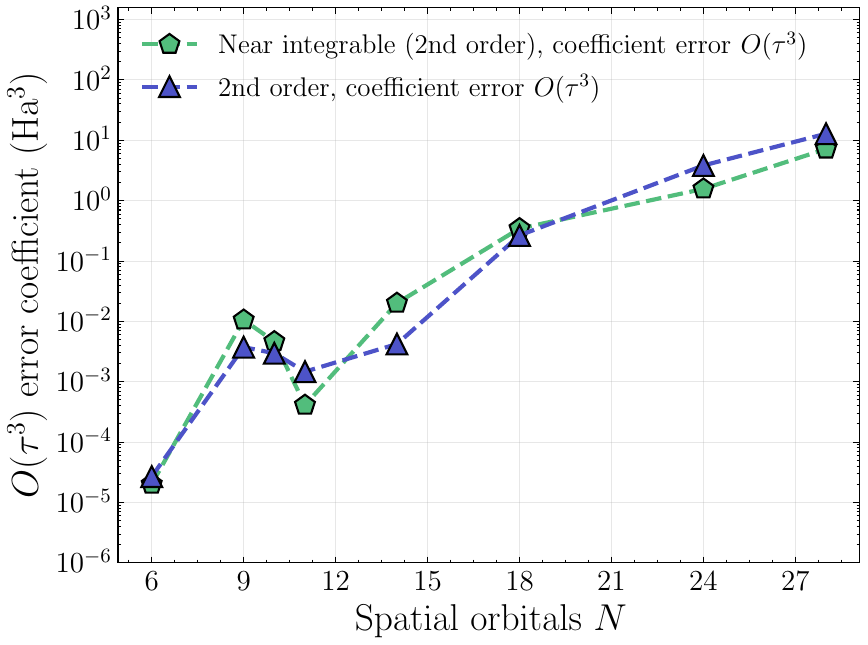}
  \includegraphics[width=0.8\linewidth]{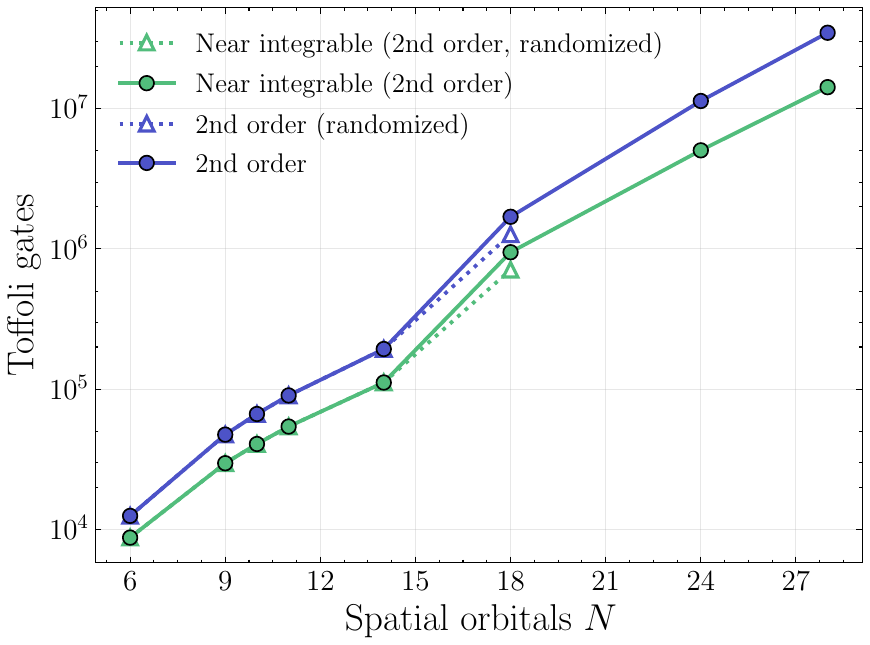}
  \caption{(a) Magnitude of the leading-order Trotter error coefficients that contribute to peak shift for the second order formula, and for the near-integrable formula $V_{2,1}(\tau)$, which we call ``near-integrable second order'', for the Li$_4$Mn$_2$O cluster systems. (b) Toffoli gate cost comparison of second order Trotter formula, and its near-integrable and randomized versions, for performing time evolution for time $t = 1$ a.u. and total error $\epsilon = 1$ eV, for Li$_4$Mn$_2$O with different active space sizes.
  Note that the second order near-integrable formulas generate an $O(\tau)$ error contribution to the effective Hamiltonian, but its leading contribution to peak shifts vanishes in perturbation theory as the error term is imaginary \cite{mehendale2025estimating}. 
  }
  \label{fig:CDF_Trotter_error_pfs}
\end{figure}

\textbf{Tighter Trotter error estimation:} Rather than bounding the Trotter error with operator-norm estimates extrapolated from small systems as in Ref. \cite{fomichev2025fast}, we estimate the \textit{actual} leading-order error for the target Hamiltonian.  Specifically, we evaluate the matrix elements of the BCH error operator~$Y_{p+1}$ on the low-energy eigenstates via first-order perturbation theory (\cref{eq:perturbation_eigenvalues,eq:perturbation_eigenstate}) using the methods of Ref. \cite{Xanadu_Trotter_error_paper}. The results for CDF were shown earlier in \cref{fig:CDF_vs_GRADE} and can be compared directly with those in Ref. \cite{fomichev2025fast}. Because the actual Trotter error ended up being much smaller than the worst-case norm bound -- and still smaller even than earlier estimates \cite{fomichev2025fast} -- we can take larger time steps than previous analysis would suggest, and therefore need fewer Trotter steps~$r$, yielding a constant-factor resource estimate reduction of $\times 1.3$ relative to prior work of Ref. \cite{fomichev2025fast}.

\begin{figure}
  \centering
    \includegraphics[width=0.8\linewidth]{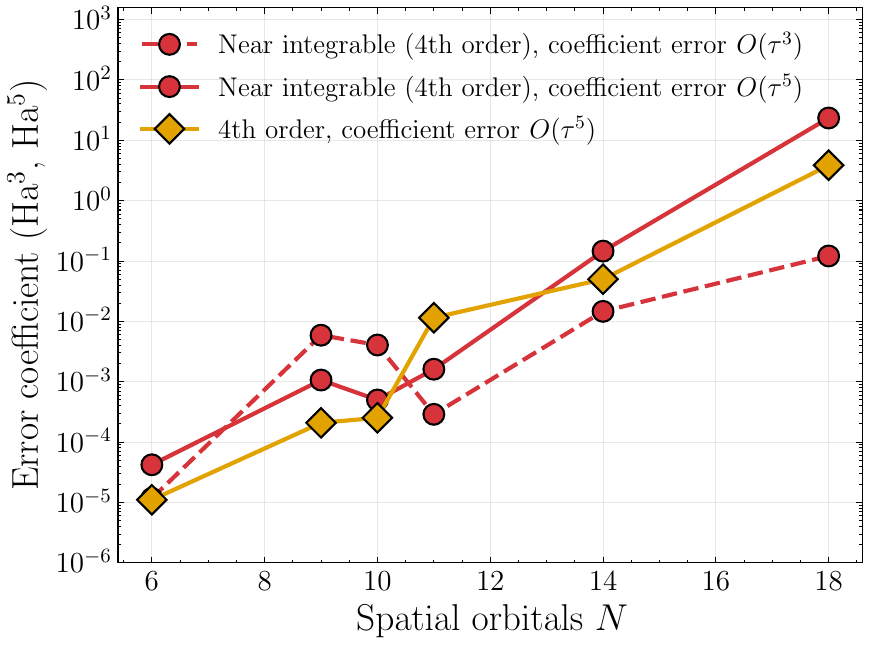}

  \includegraphics[width=0.8\linewidth]{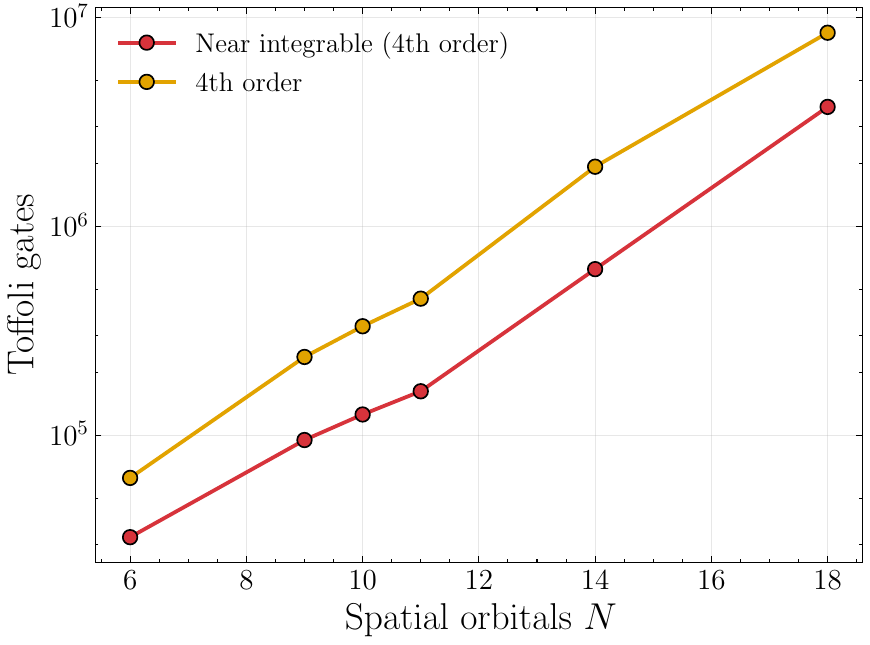}
  \caption{Same as \cref{fig:CDF_Trotter_error_pfs} but for the fourth-order Trotter product formula. Unlike in the second-order case, the near-integrable formula here contributes a higher-order leading error term of $O(\tau^3)$. For smaller system sizes this term is comparable in magnitude to the $O(\tau^5)$ prefactor and thus fully determines the Trotter step size; for larger system sizes such as $N = 18$, however, it is more than an order of magnitude smaller and therefore does not reduce the overall Trotter step size below the $O(\tau^5)$-limited value.}
  \label{fig:CDF_Trotter_error_pfs_4thorder}
\end{figure}

\textbf{Near-integrability:} Recall that the near-integrable formulas exploit the separation of the CDF Hamiltonian into a dominant group $H_A$ and a small tail $H_B$ with relative weight $\alpha \ll 1$ (\cref{ssec:Near_integrability}), selectively cancelling the dominant error terms and suppressing the remainder by $\alpha^2$. For the case of the XAS calculation for Li-excess clusters, we choose $H_A$ to contain the one-body fragment and the first two-body fragment of the CDF factorization, leaving the remaining $N-1$ of the fragments in $H_B$ (with CDF we typically found $N$ fragments led to a good factorization). This follows the pattern in \cref{fig:Zs}.
When we substituted the basic second order formula used in Ref. \cite{fomichev2025fast} with the $V_{2,1}$ near-integrable formula and computed the resulting Trotter error using the procedure described above, we found that $V_{2,1}$ can achieve the same Trotter error as the basic second-order formula, as seen in the top panel of \cref{fig:CDF_Trotter_error_pfs}. However, the $V_{2,1}$ formula can do so at an implementation cost of roughly $\times 1.6$ fewer Toffolis for the example $N = 18$ cluster system: the cost reduction is shown across system sizes in the bottom panel of \cref{fig:CDF_Trotter_error_pfs}. We find a similar result for the fourth-order Trotter formula compared with $V_{4,2}$, where we see a $\times 2.4$ reduction in the Toffoli requirement: this can be seen in \cref{fig:CDF_Trotter_error_pfs_4thorder}. In the large-system limit, these savings approach $\times 2$ and $\times 5$, respectively, since the cost becomes dominated by the cheap blocks $U_{1,B}$ and $U_{2,B}$. These per-step savings are most pronounced at short to moderate evolution times; over long evolution times the $O(\alpha^2\tau^3)$ error term eventually dominates $O(\tau^5)$, eroding the advantage of the near-integrable formula relative to Suzuki, as shown in~\cref{fig:cost_vs_time}.

\begin{table*}
\centering
\begin{tabular}{c c c c c c}
\multicolumn{2}{c}{Cost of the algorithm} & \multicolumn{2}{c}{Algorithm} & \multicolumn{2}{c}{Largest Circuit} \\
\hline\hline
 N & Logical qubits & Toffoli gates & Active Volume & Toffoli gates & Active Volume \\
\hline
$6$ & $76$ & $1.21 \times 10^{10}$ & $1.97 \times 10^{11}$ & $5.52 \times 10^{6}$ & $1.17 \times 10^{8}$ \\
$9$ & $82$ & $1.72 \times 10^{10}$ & $4.75 \times 10^{11}$ & $9.02 \times 10^{6}$ & $3.07 \times 10^{8}$ \\
$10$ & $84$ & $1.98 \times 10^{10}$ & $6.17 \times 10^{11}$ & $1.08 \times 10^{7}$ & $4.05 \times 10^{8}$ \\
$11$ & $86$ & $2.29 \times 10^{10}$ & $7.88 \times 10^{11}$ & $1.30 \times 10^{7}$ & $5.22 \times 10^{8}$ \\
$14$ & $92$ & $3.60 \times 10^{10}$ & $1.51 \times 10^{12}$ & $2.20 \times 10^{7}$ & $1.01 \times 10^{9}$ \\
$18$ & $100$ & $1.19 \times 10^{11}$ & $6.01 \times 10^{12}$ & $7.86 \times 10^{7}$ & $4.11 \times 10^{9}$ \\
$24$ & $112$ & $6.38 \times 10^{11}$ & $3.44 \times 10^{13}$ & $4.35 \times 10^{8}$ & $2.36 \times 10^{10}$ \\
$28$ & $120$ & $1.78 \times 10^{12}$ & $9.69 \times 10^{13}$ & $1.22 \times 10^{9}$ & $6.65 \times 10^{10}$ \\
\hline
\end{tabular}
\caption{Updated resource estimates for the XAS application of~\cite{fomichev2025fast}, with the updated Trotter error estimates, where we assume CDF decomposition of an active space of Li$_4$Mn$_2$O with $N$ spatial orbitals and $L = N$ fragments is used. The parameters used include $\eta = 0.05$ Ha, $\|H\|_\omega = 2$ Ha, number of shots $S = 2500$, discrete time signal time step $\delta = \pi / 2\|H\|_\omega$ and $j_{\max} = 200$. We also used $\alpha =1.3384$ to optimize the sampling procedure described in appendix B in~\cite{fomichev2025fast}.
}
\label{tab:resources}
\end{table*}

\textbf{Randomization:}
As described in \cref{ssec:randomization}, in each Trotter step we can randomly permute the fragment ordering within each group, causing certain error terms to cancel in expectation over the full evolution -- at no extra gate cost. The only trade-off is a small spectral line broadening (this is described in more detail in \cref{app:randomized_spectra}). Using the same effective Hamiltonian perturbative analysis, we can compute the reduction of the Trotter error relative to an un-randomized product formula, and translate that reduction into a corresponding increase of the maximum allowable Trotter step size -- and thus into a Toffoli cost reduction for the overall algorithm. In the specific case of the Li-excess cluster system, for system sizes $N\leq 14$ we selected 11 orderings to estimate the error. These included the $\ell$-strictly-increasing and $\ell$-strictly-decreasing orderings, along with 9 random orderings. For $N = 18$, we selected $\ell$-strictly-increasing and $\ell$-strictly-decreasing orderings, plus 2 or 3 random orderings for the near-integrable and standard second-order formulas respectively, as they were more expensive to evaluate. This number of orderings is likely not sufficient to get a fully converged estimate, but can provide an upper bound of the error and cost. The final costs of the randomized versions of the standard and near-integrable second order product formulas are shown in the bottom panel of \cref{fig:CDF_Trotter_error_pfs}: while they are relatively small for smaller system sizes, for the $N = 18$ representative system they amount to around $\times 1.4$ fewer Toffoli gates than in previous work.

\textbf{QROM-based compilation:}
Each Trotter step contains a block of mutually commuting $\sigma_z \otimes \sigma_z$ rotations.  We replace them with a single Quantum Read-Only Memory (QROM) look-up \cite{low2024trading} that precomputes the cumulative phase for every basis state (\cref{ssec:QROM,fig:QROM_circuit}), making the block $\times 2$--$\times 4$ cheaper in Toffoli gates (\cref{tab:QROM_savings}).  Combined with the improved Givens-rotation circuits of Ref. \cite{caesura2025faster} (\cref{fig:givens_fused_adder}), the per-step cost $C_{\text{step}}$ drops by $\times 1.3$--$\times 1.9$, depending on precision~$b$, fragment rank~$M_\ell$, and orbital count~$N$.

\textbf{Processing:}
In general, processing, i.e. applying a unitary~$e^P$ at the start and end of the time evolution, can cancel certain leading errors terms in the effective Hamiltonian, such as the $O(\alpha\,\tau^3)$ error in the $V_{4,2}(\tau)$ near-integrable formula, at only a fixed additive cost independent of~$r$ (\cref{ssec:Processing,eq:processor}). In practice, for the Li-excess cluster problem, since the second-order formula was typically sufficient to achieve the error requirements, we did not find a benefit to using processing. At higher orders where there could potentially be such a benefit, such as for $V_{4,2}$, it turned out that the targeted commutator already happened to have a near-zero expectation value for all but the smallest active space ($N = 6$).  However, we expect that in contexts with longer time evolution times or stricter spectral peak shift error requirements processing might be a powerful addition to the product formula toolkit.

Overall, the combined savings from compilation, randomization, near-integrability, and tighter Trotter error estimation for the $N = 18$ active space of Li$_4$Mn$_2$O are shown in \cref{fig:SPRINT_Towers} and cumulatively amount to a $\times 4.5$ improvement in total Toffoli count relative to Ref. \cite{fomichev2025fast}. For convenience, the final resource estimates for all system sizes are reported in~\cref{tab:resources}.

\subsection{Comparison with qubitization\label{ssec:Qubitization}}

\begin{figure*}
    \centering
    \includegraphics[width=0.45\linewidth]{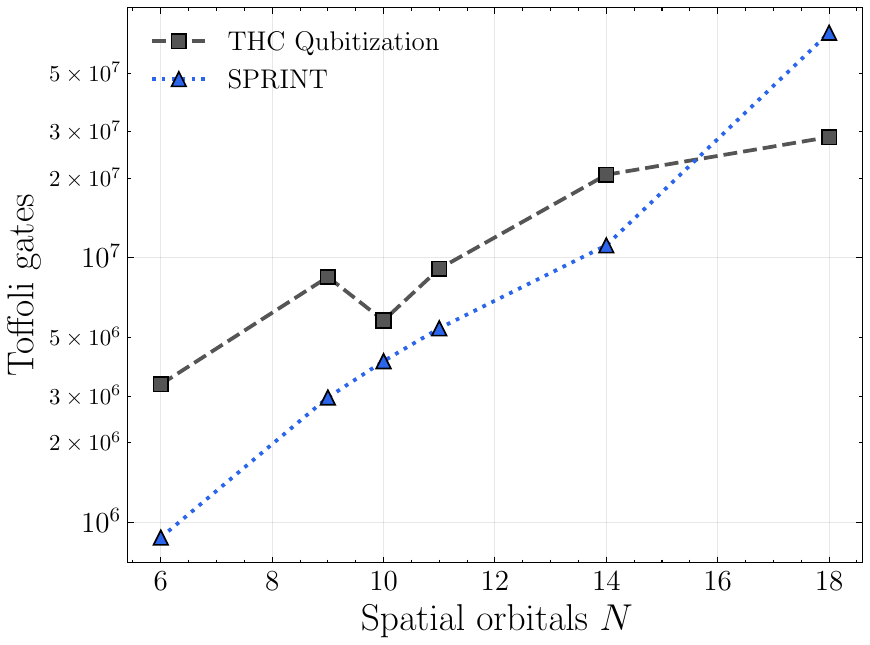}
    \includegraphics[width=0.45\linewidth]{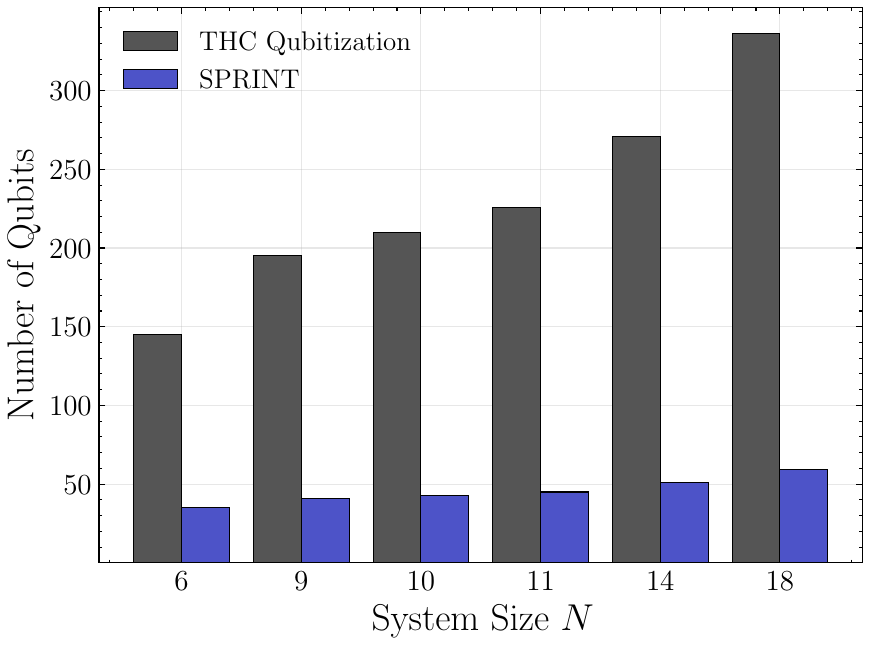}
    \caption{Toffoli gate and qubit comparison for running the spectroscopy algorithm with Trotter and qubitization. The 2nd order randomized near-integrable formula  combined with a CDF decomposition of the Hamiltonian. This leads to $\times 2.5$ more Toffoli gates in Trotter than in qubitization, but $\times 5.5$ fewer qubits.}
    \label{fig:trotter_vs_qubitization}
\end{figure*}

Having examined the gains that SPRINT brings to simulating XAS for the Li-excess cluster, we now compare it against symmetry-shifted THC qubitization, the most competitive general-purpose alternative for implementing time evolution~\cite{lee2021even,caesura2025faster,loaiza2023block,low2025fast}. We focus on THC-based qubitization because while even faster methods exist (see the spectrum amplification method of Ref.~\cite{low2025fast}), they have mostly been investigated for ground state energy estimation and it is not yet clear how they could be applied to spectroscopy-like tasks. From SPRINT, we consider a CDF-based factorization with a randomized, second-order near-integrable product formula, as evaluated in the previous section. Our main finding is that SPRINT remains within $\times 2.5$ in Toffoli count of the qubitized approach, while using a dramatic $\times 5.5$ fewer qubits -- though qubitization scales better asymptotically due to its $O(N)$ per-step cost enabled by QROM. 

To see this, we perform resource estimation for the problem of XAS on a Li-excess cluster using qubitization (see \cref{app:qubitized_xas}), and compare those estimates with the SPRINT results from \cref{fig:CDF_Trotter_error_pfs}: the comparison is shown in \cref{fig:trotter_vs_qubitization}. The qubitization resource estimate is carried out as in the literature \cite{caesura2025faster}: the only caveat is that rather than using quantum signal processing to convert the quantum walk $e^{\pm i t\arccos(H/\lambda)}$ into $e^{i t H}$, we use a Chebyshev transform of the time signal in place of the Fourier transform (see \cref{app:chebyshev_transform}). In \cref{fig:trotter_vs_qubitization} we see that SPRINT remains competitive with qubitization. In particular, for $N = 18$, the Toffoli cost of SPRINT is approximately a factor of $\times 2.5$ larger than that of qubitization.
Conversely, the qubit costs of qubitization are significantly larger, requiring $\times 5.5$ more qubits than SPRINT for the $N = 18$ system. 

Having obtained concrete estimates, we now analyze the origin of the relative advantages of the two methods. For qubitization within the THC framework, the total complexity scales as $\tilde{O}(\lambda \cdot C_{\text{block}})$, where $\lambda$ is the Hamiltonian $1$-norm~\cite{low2024trading} and $C_\text{block}$ is the cost of the block-encoding like that reported in \cref{fig:GRADE_per_step_cost_1eV}. By leveraging Quantum Read-Only Memory (QROM) to implement the PREPARE and SELECT oracles, THC qubitization applies only a single partial basis transformation at a time, reducing the block-encoding cost to $C_{\text{block}} \sim O(N)$~\cite{lee2021even,caesura2025faster}. Since $\lambda$ scales as $O(N^{1.1\text{--}1.2})$ in the thermodynamic limit, the aggregate gate complexity is $\tilde{O}(N^{2.1\text{--}2.2})$~\cite{lee2021even}. 

By contrast, although QROM does yield constant-factor savings (roughly $2$--$4\times$) for $\sigma_z \otimes \sigma_z$ rotations in Trotter methods, analogous asymptotic improvements appear fundamentally impossible: Trotter rotations act on distinct qubits rather than coherently in superposition, precluding the optimization that QROM affords to qubitization. Fundamentally, Trotterization handles basis rotations differently: it diagonalizes individual fragments via \textit{global} basis changes $\mathcal{U}$, requiring a dense mesh of interacting Givens rotations at a cost of $O(N^2)$ per step. The total gate complexity therefore scales as $O(N^2 \lambda) \sim O(N^3)$, exceeding that of qubitization. This complexity gap between the two techniques has a structural origin: qubitization (via LCU) pays a cost proportional to the maximum single-term cost, whereas a product formula must execute every term in sequence. This is consistent with the $\Omega(N^2)$ per-step lower bound for generic two-local Hamiltonians, as proved in Ref.~\cite{low2023complexity}.

Architecture considerations can also have a large impact on this cost complexity gap. For example, in the regime of abundant gate parallelism -- potentially accessible via magic state cultivation~\cite{gidney2024magic} -- the asymptotic disparity in circuit depth completely vanishes: the Givens rotations and fast-forwardable cores of each fragment can be parallelized to $O(N)$ depth per step, matching qubitization, so both algorithms exhibit a total circuit depth of $O(N^2)$. Conversely, for reaction-limited architectures, the total gate count determines the active volume and execution time~\cite{litinski2022active}, and under such constraints the $O(N^2)$ gate cost of qubitization retains an asymptotic advantage over the $O(N^3)$ cost of molecular-orbital Trotterization. This advantage can be traded between space and time, limited only by the reaction time. Qubitization therefore enjoys a fundamental asymptotic advantage over molecular-orbital Trotterization in such architectures -- not due to error complexity, but due to the block-encoding gate cost.

At the same time, the much lower constant prefactors of Trotter methods can partially counteract this scaling discrepancy even in reaction-limited architectures, as we find here and as the quantum literature has long emphasized, even if explicit empirical comparisons have been scarce. One example of this is Ref.~\cite[Fig.~1]{gunther2025phase}, which finds Trotter less competitive for second-quantized systems at large sizes; conversely, Ref.~\cite{rubin2024quantum} finds that Trotter outperforms qubitization in first quantization, even without incorporating all techniques from the geometric integration literature~\cite{blanes2000processing,blanes2024splitting}. For instance, for the plane-wave $H = T + V$ decomposition, the quantum literature has overlooked Runge--Kutta--Nystr\"{o}m methods that exploit $[V,[V,[V,T]]]=0$~\cite{mclachlan2019lie} (see~\cref{prop:VVT_vanishing} in~\cref{app:RKN}). This relative paucity of constant-factor comparisons together with favourable comparisons obtained here and previously in the literature implies there are likely other, yet-to-be-identified circumstances where Trotterized methods may be preferable to qubitization, and underscores the need for continued development of state-of-the-art Trotter methods.

In conclusion, despite asymptotics favouring qubitization, the results of this section show that well-designed Trotter product formulas can perform competitively for simulating relatively small electronic Hamiltonians in second quantization on problems of industrial relevance -- especially those hard enough for classical methods, such as excited states and dynamics. Future work should extend these comparisons to other families of Hamiltonians and simulation tasks.

\section{Conclusion}

We have introduced SPRINT -- Symmetry-Protected Randomized Near-Integrable Trotter formulas -- a family of Trotter methods tailored to electronic Hamiltonians. SPRINT unifies near-integrable product formulas with processing, symmetry protection, randomization and QROM-based compilation, exploiting the norm structure that arises naturally in rank-factorized Hamiltonians. On the problem of simulating X-ray absorption spectra of Li$_4$Mn$_2$O, SPRINT reduces the Toffoli gate cost by a factor of $\times 4.5$ relative to the previous state of the art. Moreover, for the target system size $N = 18$ our methods achieve $\times 5.5$ fewer logical qubits than qubitization, at only a $\times 2.5$ gate-count premium.

We also introduced GRADE, a generalized rank factorization that smoothly interpolates between Compressed Double Factorization and isometric THC. GRADE consistently lowers per-step gate counts, but for the systems and spectroscopy-focused workflows studied here, the growth in Trotter and leakage errors neutralizes these savings. Because GRADE strictly generalizes CDF, targeted hyperparameter optimization may yet close the gap -- a question we leave to future work. Importantly, SPRINT delivers state-of-the-art performance regardless of the underlying factorization, requiring only that fragment-by-fragment fitting be used to create the hierarchical decaying norm structure.

Beyond the numerical results, the Trotter error analysis developed here -- evaluating leading-order BCH commutators via matrix-product operators -- yields constant-factor error estimates that place Trotter and qubitization methods on an equal analytical footing. Extending this analysis to a broader class of industrially relevant Hamiltonians, and pinning down how the Trotter error scales with system size, are natural next steps toward reliable resource estimation in the fault-tolerant era.

Our numerical findings challenge the widespread assumption that qubitization is the default method for fault-tolerant molecular simulation. On early fault-tolerant hardware, where every logical qubit counts, SPRINT makes Trotter product formulas an attractive choice. As both Trotter and qubitization techniques continue to improve, the resource crossover between them will shift -- but the tools developed here ensure that product formulas remain a first-class contender, closing a methodological gap that has persisted since the early days of quantum simulation.

\section{Acknowledgements}
Pablo A M Casares thanks Fernando Casas and Sergio Blanes for early conversations on the geometric integration methods we used in this manuscript. We also thank Maxine Luo for clarifications on the isometric THC method, and code to implement the isometry decomposition into Givens rotations. This research used resources of the National Energy Research Scientific Computing Center, a DOE Office of Science User Facility supported by the Office of Science of the U.S. Department of Energy under Contract No. DE-AC02-05CH11231 using NERSC award
NERSC DDR-ERCAP0036805.

\bibliography{main}
\clearpage
\newpage
\onecolumngrid

\appendix
\section{Hamiltonian construction for experimental results\label{app:hamiltonian_construction}}

The starting point for the simulation is a molecular cluster representing the local environment of the absorbing atom. In our case, following Ref. \cite{fomichev2025fast}, we focus on the oxygen atom.
The ultralocal nature of X-ray absorption \cite{fomichev2024simulating} implies that in many cases, only the nearest-neighbor shell around the absorber needs to be included in the model.
Following the procedure established in Refs. \cite{fomichev2024simulating,fomichev2025fast}, we extract the oxygen-centered cluster Li$_4$Mn$_2$O from the crystal structure of the Li-excess cathode material Li$_2$MnO$_3$. The cluster consists of a central oxygen atom surrounded by its first coordination shell of four lithium and two manganese atoms, whose coordinates can be found in~\cref{tab:LiMnO_geometry}. This specific cluster is an example system: in principle, different stages of delithiation, i.e. of removal of lithium atoms during battery charging, yield a family of clusters with varying local coordination, each potentially corresponding to a different oxidation state of the absorbing atom.

Next, we employ the cc-pVDZ basis set for all atoms in the cluster. Starting from this basis, we execute a restricted Hartree-Fock calculation using PySCF \cite{sun2015libcint,sun2018pyscf,sun2020recent}. The resulting molecular orbitals serve as the reference single-particle basis from which the active space is subsequently constructed.  We used the automated valence active space selection method \cite{sayfutyarova2017automated} to generate a sequence of active spaces of increasing size. The choices of orbitals we retain are chemically motivated: in particular, we focus on the strongly correlated $3d$ and $4d$ orbitals of Mn; the ligand bonding orbitals $2p$ and $3p$ of O and Mn, respectively; the valence $2s$ of Li; and of course the core orbital $1s$ of O, which is key to the $K$-edge X-ray response we are interested in. Different combinations of these orbitals allow us to build Hamiltonians with as few as $N = 6$ spatial orbitals to as many as $N = 28$. Where necessary, we lower the default AVAS threshold until the desired number of orbitals is included in the active space. This threshold controls the minimum overlap between the pre-specified atomic valence orbitals (e.g. O~$1s$, Mn~$3d$) and the molecular orbitals; reducing it admits orbitals with weaker atomic character. The detailed configurations of AVAS and their associated orbital counts are shown in~\cref{tab:avas_configs}. 

\begin{table}[htbp]
  \centering
  \begin{tabular}{lccc}
    \toprule
    \textbf{Atom} & {\textbf{X}} & {\textbf{Y}} & {\textbf{Z}} \\
    \midrule
    O      &  0.00  &  0.00  &  0.00  \\
    Mn$_1$ & -0.03 &  0.05 &  1.70  \\
    Mn$_2$ & -0.01 &  1.60  &  0.10  \\
    Li$_1$ &  1.90  &  0.01 &  0.20  \\
    Li$_2$ & -1.80  & -0.07 &  0.02 \\
    Li$_3$ &  0.04 & -1.75 & -0.04 \\
    Li$_4$ &  0.06 &  0.02 & -1.60  \\
    \bottomrule
  \end{tabular}
    \caption{Atomic coordinates (in \AA) for the Li$_4$Mn$_2$O cluster.}
  \label{tab:LiMnO_geometry}
\end{table}

\begin{table}[htbp]
\centering
\begin{tabular}{clc}
\toprule
\textbf{Size ($N$)} & \textbf{Included Atomic Orbitals} & \textbf{AVAS Threshold} \\
\midrule
6  & O 1$s$, Mn$_1$ 3$d$ & 0.5000 \\
9  & O 1$s$, O 2$p$, Mn$_1$ 3$d$ & 0.5000 \\
10 & O 1$s$, Li 2$s$, Mn$_1$ 3$d$ & 0.5000 \\
11 & O 1$s$, Mn$_1$ 3$d$, Mn$_2$ 3$d$ & 0.5000 \\
14 & O 1$s$, O 2$p$, Mn$_1$ 3$d$, Mn$_2$ 3$d$ & 0.5000 \\
18 & O 1$s$, O 2$p$, Li 2$s$, Mn$_1$ 3$d$, Mn$_2$ 3$d$ & 0.5000 \\
24 & O 1$s$, O 2$p$, Li 2$s$, Mn$_1$ 3$d$, Mn$_2$ 3$d$, Mn$_1$ 4$d$, Li 3$p$ & 0.0300 \\
28 & O 1$s$, O 2$p$, Li 2$s$, Mn$_1$ 3$d$, Mn$_2$ 3$d$, Mn$_1$ 4$d$, Mn$_2$ 4$d$ & 0.0095 \\
\bottomrule
\end{tabular}
\caption{Active Space Configurations for the Li$_4$Mn$_2$O cluster used to generate the Hamiltonians.}
\label{tab:avas_configs}
\end{table}

The final modification to the Hamiltonian obtained with AVAS is the application of the core-valence separation approximation (CVS), a common method in classical XAS simulations  \cite{cederbaum1980many,barth1981many,norman2018simulating,herbst2020quantifying}. CVS allows the quantum algorithm to directly compute the spectra of core-excited states, namely those that have a core hole -- in this case, in the O $1s$ orbital -- while bypassing all the other excited states in the so-called valence-excited manifold that are formally much lower in energy. This is accomplished by exploiting the observation that the Hamiltonian matrix elements coupling core-excited and valence-excited determinants are typically small -- the very reason that the core and valence distinction makes sense \cite{cederbaum1980many}.
Setting these matrix elements to exactly zero decouples the two manifolds at negligible cost to the accuracy of the simulation \cite{norman2018simulating,herbst2020quantifying}.
Operationally, CVS is implemented in two steps:
\begin{enumerate}
    \item All two-electron integrals $(pq|rs)$ involving at least one core orbital index (here, the O~$1s$ orbital) with mixed core-valence character are set to zero \cite{norman2018simulating}. This removes the off-diagonal blocks connecting the core-excited and valence-excited sectors.
    \item Only terms in the dipole operator $m_\rho$ that involve excitations \emph{from} the core orbital are retained in the initial state $m_\rho$. This ensures that the initial state $m_\rho \ket{I}$ is placed entirely within the core-excited subspace. Since the modified Hamiltonian preserves this subspace, the subsequent time evolution remains confined to it.
\end{enumerate}

With the active space and CVS in hand, the electronic Hamiltonian and the dipole operator $m_\rho$ are constructed within the active space using PySCF \cite{sun2018pyscf,sun2020recent} and PennyLane \cite{bergholm2018pennylane}. To determine the initial state $m_\rho \ket{I}$, the ground state $\ket{I}$ is obtained by using either the complete active space (CAS) method or density matrix renormalization group (DMRG) for the larger systems, and the dipole operator is applied using the associated one-body creation-annihilation operator pairs \cite{fomichev2025fast}. This procedure yields all the ingredients needed to run the time-domain XAS algorithm: the system Hamiltonian $H$ in~\cref{eq:G(t)}, the initial state $m_\rho \ket{I}$, and the norm $\| m_\rho \ket{I} \|$ required for normalization.

\section{Improving GRADE Trotter error\label{app:GRADE_error}}

As we saw in~\cref{fig:CDF_vs_GRADE} and~\cref{ssec:grade_results}, the optimized GRADE schedule found in~\cref{fig:GRADE_per_step_cost_1eV} displayed a large Trotter error, that makes GRADE unattractive compared to CDF. However, the results in that section depended on the GRADE schedules found by an optimizer, which often tended to favor either CDF-like or most often THC-like schedules. This made it difficult to assess whether the large Trotter error was a result of the optimizer search or we were not properly exploring the space of schedules. 

Here we instead fix the schedules $M_\ell$ by hand. We use $L = N$ fragments in CDF and a single fragment of size $M = 3N$ in isometric THC. For simplicity, GRADE will use fragment sizes ranging from $M_1 = N$ to $M_{\ell = \lceil\log_2 N\rceil} = 2N$ linearly interpolated for simplicity, favoring shorter steps first if possible. What this means in practice is indicated in~\cref{tab:grade_schedules}. Since there are only $\lceil\log_2 N \rceil$ fragments, the scaling of GRADE will asymptotically remain the same as isometric THC, $O(N^2)$ up to polylogarithmic factors. 

We also explored the option of regularizing the GRADE fragments with a penalty term computed as $\zeta\|Z^{(\ell)}\|$, using the Frobenius norm. Our goal is to understand if we can reduce the Trotter and leakage errors to make these schedules more attractive than CDF. In~\cref{fig:Trotter_error_GRADE} we show that the Trotter error can be tamed with the choice of schedule. 

The conclusion is that while GRADE seems an attractive theoretical generalization to explore, more work is needed to make sure it can be competitive with Compressed Double Factorization. In this work we have hinted at some techniques future work might want to explore, including quantum singular value transform and error mitigation strategies.

\begin{figure*}
  \centering
  \includegraphics[width=0.45\linewidth]{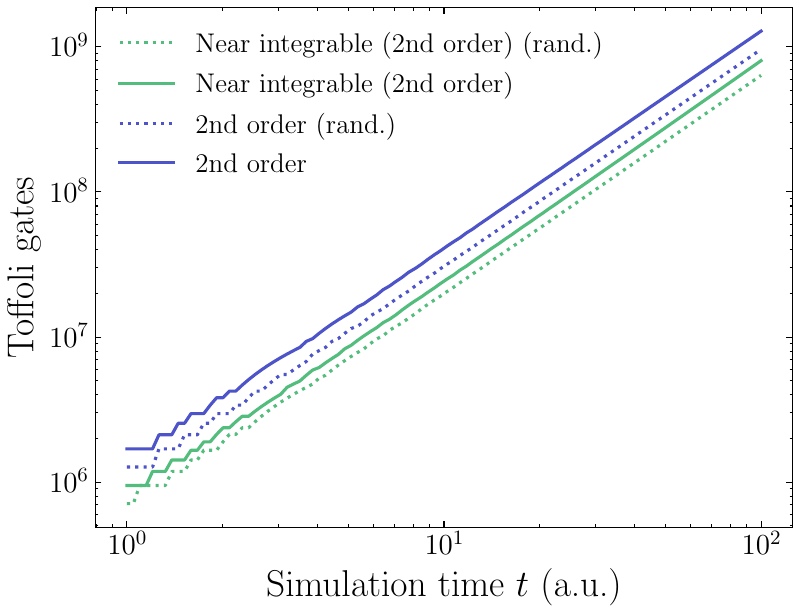}
  \includegraphics[width=0.45\linewidth]{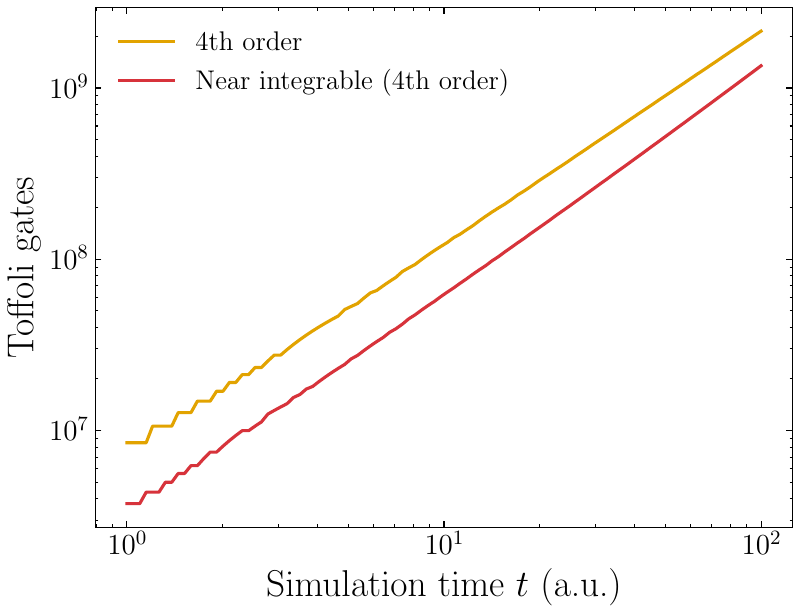}
  \caption{Toffoli gate cost comparison for near-integrable formulas of order 2 and 4, vs Strang or Suzuki respectively, for time $t$ and total error $\epsilon = 1$ eV, and Li$_4$Mn$_2$O with active space size $N=18$. Over long evolution times, the $O(\alpha^2 \tau^3)$ error term starts to dominate over $O(\tau^5)$, making the near-integrable formula less competitive than Suzuki for time evolution applications.}
  \label{fig:cost_vs_time}
\end{figure*}

\begin{figure*}
  \centering
  \begin{overpic}[width=0.47\linewidth, percent]{figures/paper/zeta_race_comparison_Strang_bd100.pdf}
    \put(0,70){(a)} 
  \end{overpic}
  \hfill
  \begin{overpic}[width=0.47\linewidth, percent]{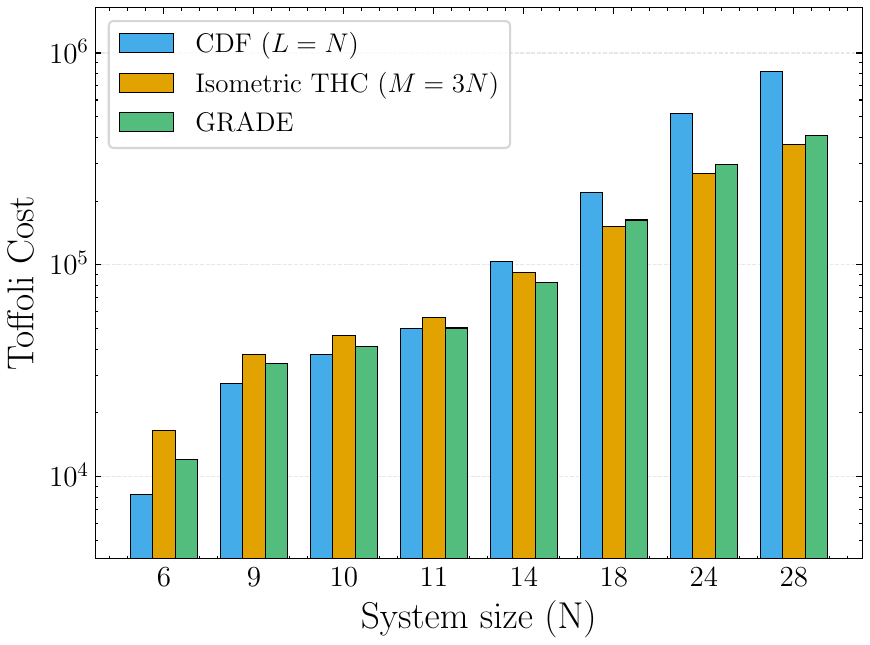}
    \put(0,70){(b)} 
  \end{overpic}
  \caption{(a) Trotter and leakage error of GRADE and a second order formula with manually selected schedule according to~\cref{tab:grade_schedules}. We used bond dimension of 100 for all points except for the leakage of $N = 10$ and $N = 11$, which use $\max M_\ell = 2N$. Since this is a large system, we used bond dimension 25 in those cases. $\zeta$ indicates the strength of the regularization loss of $\|Z^{(\ell)}\|$ and the leakage in~\cref{eq:cheap_leak}, during the optimization of $Z^{(\ell)}$ and $U^{(\ell)}$. (b) Per-step cost comparison of CDF, isometric THC and the GRADE schedule described in~\cref{app:GRADE_error} and~\cref{tab:grade_schedules}, assuming all three methods use the same rotation precision.}
  \label{fig:Trotter_error_GRADE}
\end{figure*}

\begin{table}[ht]
\centering
\begin{tabular}{cc}
\hline
$N$ & Schedule $(M_1, M_2, \ldots)$ \\
\hline
6  & $(6, 9, 12)$ \\
9  & $(9, 12, 15, 18)$ \\
10 & $(10, 13, 16, 20)$ \\
11 & $(11, 14, 18, 22)$ \\
14 & $(14, 18, 23, 28)$ \\
18 & $(18, 22, 26, 31, 36)$ \\
24 & $(24, 30, 36, 42, 48)$ \\
28 & $(28, 35, 42, 49, 56)$ \\
\hline
\end{tabular}
\caption{GRADE schedules for LiMnO systems at different active space sizes.}
\label{tab:grade_schedules}
\end{table}

\section{Convergence analysis for the Trotter error evaluation\label{app:Convergence}}

Here we do a brief summary of convergence with the relevant parameters: bond dimension, number of commutators evaluated and number of eigenstates. We start with the bond dimension. In~\cref{fig:bond_dimension_convergence} we show the result obtained for the different bond dimension sizes. 

We also depict the results for convergence with the number of commutators evaluated (prioritizing those with a larger norm) in~\cref{fig:convergence_commutators} and with the number of eigenstates in~\cref{fig:convergence_eigenstates}.

\begin{figure}
  \centering
  \begin{overpic}[width=0.47\linewidth, percent]{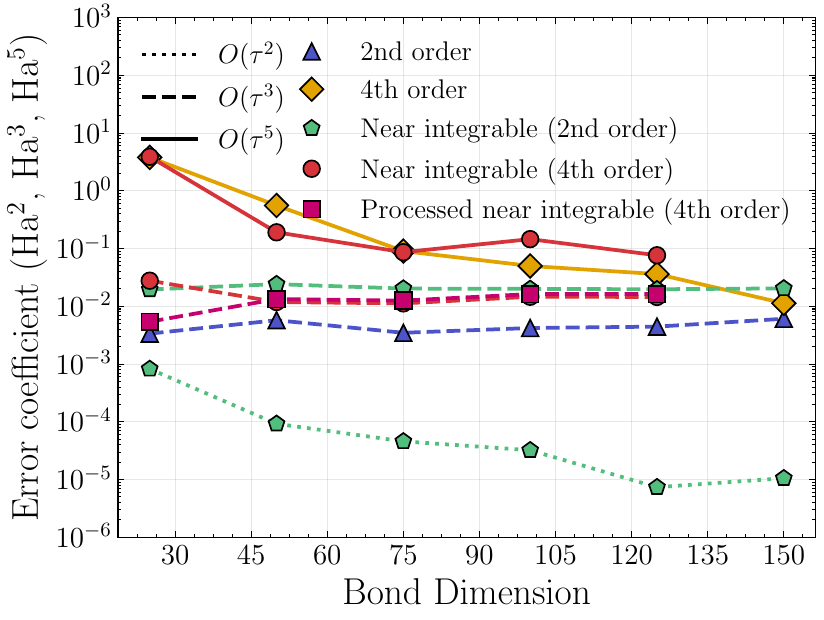}
    \put(0,75){(a)} 
  \end{overpic}
  \hfill
  \begin{overpic}[width=0.47\linewidth, percent]{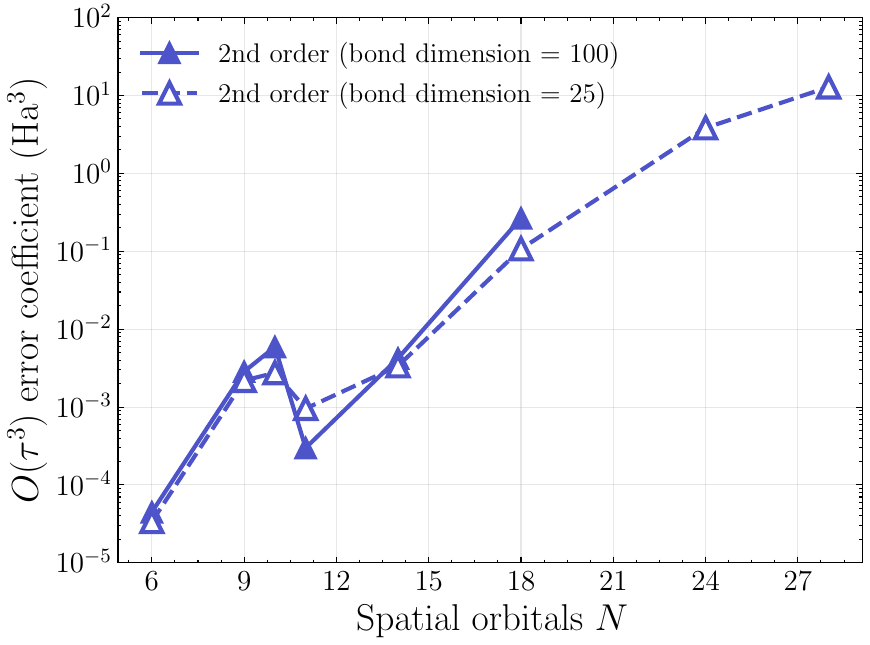}
    \put(0,75){(b)} 
  \end{overpic}
  \caption{(a) Bond dimension convergence for $O(\tau^p)$ error coefficient in Compressed Double Factorized Li$_4$Mn$_2$O Hamiltonian with $N = 14$. It seems bond dimension 100 represents a fair compromise between cost and accuracy. For second order formulas even a low bond dimension estimate can be sufficient to estimate the Trotter error. Note that in the limit of high accuracy, the $O(\tau^2)$ vanishes, see discussion in caption of~\cref{fig:CDF_Trotter_error_pfs}. (b) Error coefficient for the second order formula as a function of system size $N$, estimated with bond dimension 25 and 100. The low bond dimension is sufficiently accurate for some reasonable estimates.}
  \label{fig:bond_dimension_convergence}
\end{figure}

The observation that the importance of Trotter error commutators frequently exhibits a log-normal distribution, see~\cref{fig:convergence_commutators}, can be rationalized through the multiplicative central limit theorem, even in the presence of deterministic physical structure. In Hamiltonian simulation, the magnitude of a first-order commutator is classically bounded by $2 \|H_i\| \|H_j\|$. For higher-order Trotter-Suzuki decompositions, the error is dominated by nested commutators of depth $k$, whose importance is bounded by the product of their respective fragment norms:
\begin{equation}
  \mathcal{E}_k \leq 2^k \prod_{m=1}^{k+1} \|H_{i_m}\|.
\end{equation}
To analyze the statistical distribution of these bounds over the combinatorially large set of non-vanishing commutators, one can consider the natural logarithm of the importance metric. This transformation maps the multiplicative bound into a linear sum of random variables:
\begin{equation}
  \ln(\mathcal{E}_k) \leq k \ln(2) + \sum_{m=1}^{k+1} \ln(\|H_{i_m}\|).
\end{equation}
In practical quantum chemistry applications, the Hamiltonian is often represented using rank-reduced tensor factorizations, such as the Cholesky decomposition or double factorization. These techniques produce fragment operators whose norms span several orders of magnitude, decaying rapidly in accordance with the hierarchical energy scales of the system. Consequently, the logarithmic norms $\ln(\|H_{i_m}\|)$ can be treated as variables drawn from a distribution with a well-defined mean and variance.

A natural theoretical objection to the application of the central limit theorem in this context is that Hamiltonian coefficients are inherently correlated by physical constraints, such as spatial locality, point-group symmetries, and selection rules. Such constraints ostensibly violate the independence assumption required for the theorem. However, the matrix decomposition process intrinsically scrambles these local features, mapping the physical structure into global, highly oscillatory tensor fragments. While the principal fragments with the largest norms remain highly structured and deterministic, the vast majority of fragments residing in the tail of the decomposition exhibit weak mutual correlations. 

When sampling pairs or higher-order tuples from this bulk, the selection rules act as a pseudo-random filter. Provided that these structural correlations are sufficiently weak across the bulk ensemble, a generalized central limit theorem for weakly dependent variables applies. Under these conditions, the sum of the logarithmic norms converges to a normal distribution. Upon exponentiation, the commutator importance $\mathcal{E}_k$ necessarily approaches a log-normal distribution. This highly right-skewed statistical behavior concentrates the macroscopic Trotter error into a vanishingly small fraction of dominant terms, providing a rigorous explanation for why standard analytic bounds, which assume a uniform accumulation of error, are routinely overly pessimistic in empirical simulations.

\begin{figure*}[t]
  \centering
  \begin{overpic}[width=0.48\linewidth, percent]{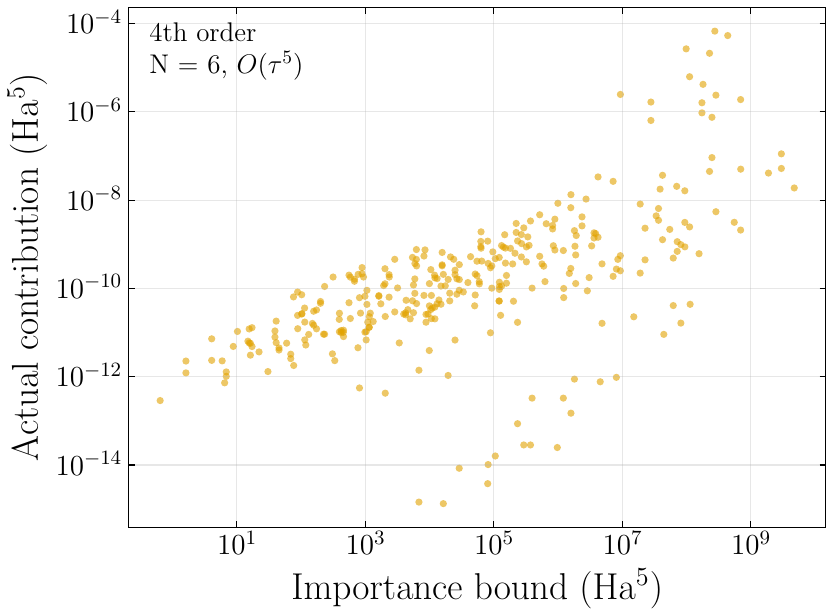}
    \put(1.,75){(a)} 
  \end{overpic}
  \hfill
  \begin{overpic}[width=0.48\linewidth, percent]{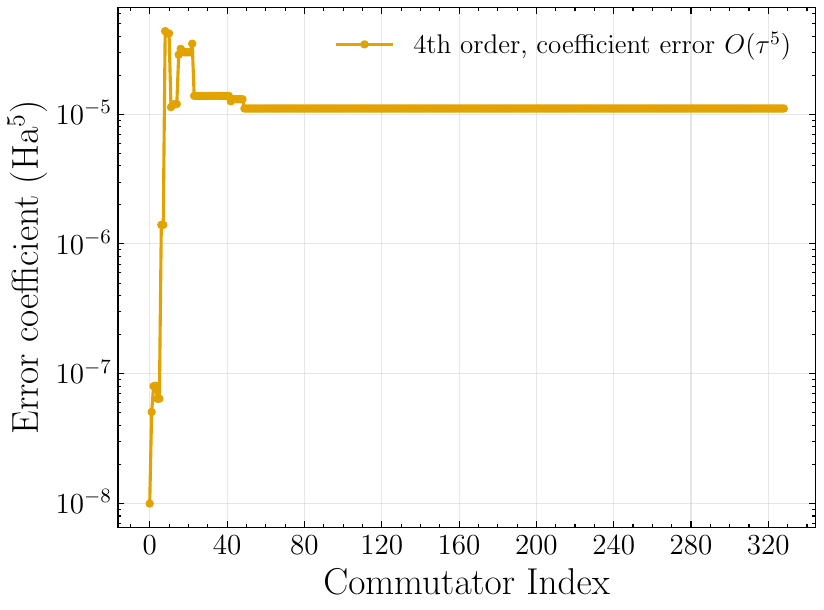}
    \put(1.,75){(b)}
  \end{overpic}
  \vspace{0.5cm}
  \begin{overpic}[width=0.48\linewidth, percent]{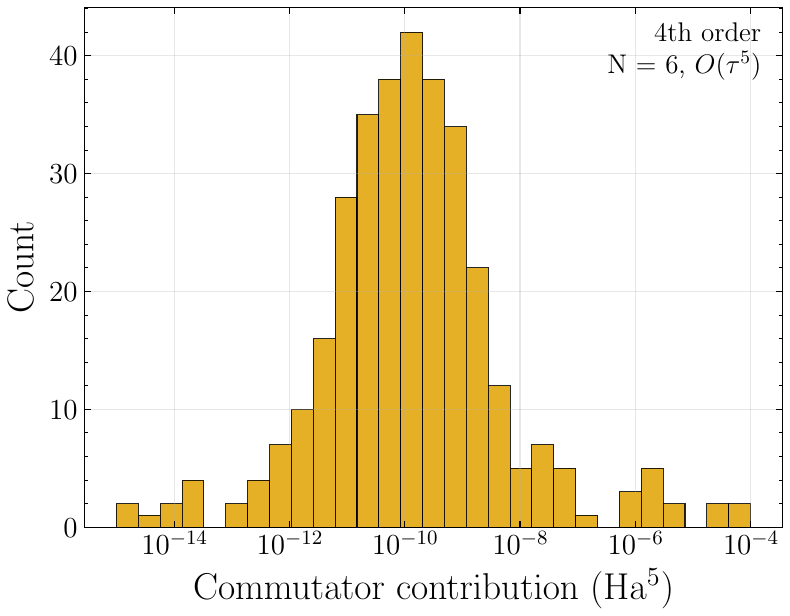}
    \put(1.,75){(c)}
  \end{overpic}
  \hfill
  \begin{overpic}[width=0.48\linewidth, percent]{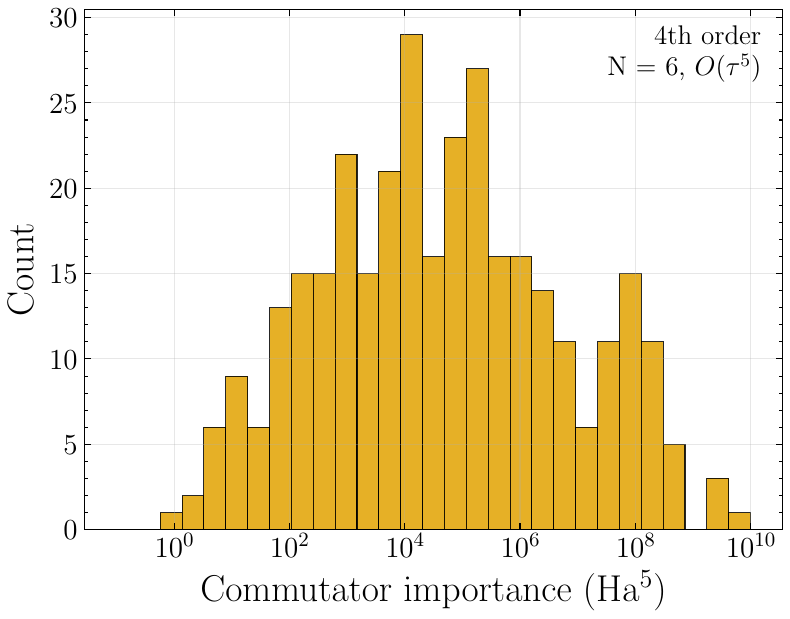}
    \put(1.,75){(d)}
  \end{overpic}
  
  \caption{Convergence of the Trotter error estimation with the number of commutators evaluated for a Suzuki fourth order formula with $N = 6$. This plot aims to understand to what degree we can use the evaluated importance -- defined as bound on the norm -- to identify the commutators to evaluate. In (a) we observe that while the correlation is weak, all the high contribution commutators are also fairly high in importance. (b) depicts how the estimated error accumulates with the importance-ordered accumulated estimate. (c) and (d) depict the histograms for the commutator contribution and importance.}
  \label{fig:convergence_commutators}
\end{figure*}

\begin{figure}
  \centering
  \includegraphics[width=0.5\linewidth]{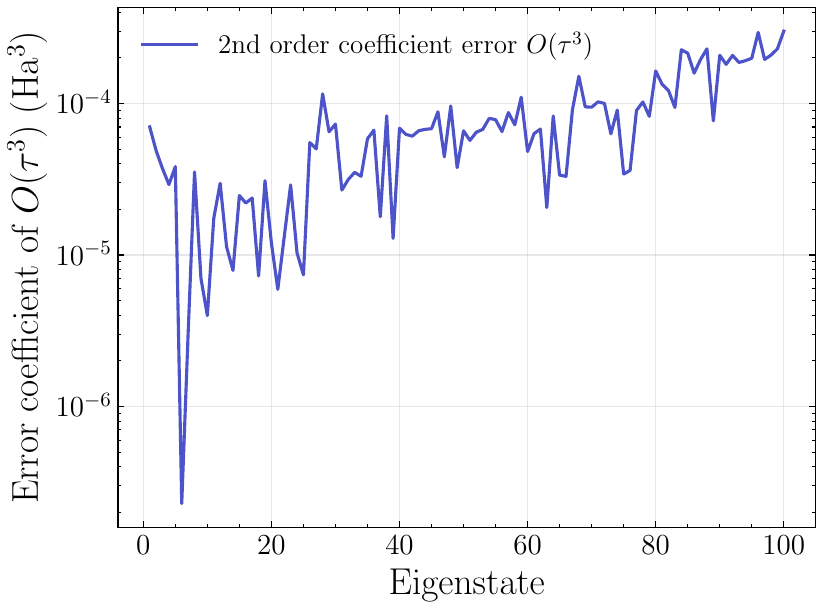}
  \caption{Raw Trotter error estimated for the first 100 eigenstates evaluated, for Li$_4$Mn$_2$O Compressed Double Factorized Hamiltonian with $N = 6$.}
  \label{fig:convergence_eigenstates}
\end{figure}

\section{Resource estimation procedure for the XAS algorithm}
\label{app:resource_estimation}

To estimate the impact of each technique, we perform resource estimation for carrying out a spectroscopy simulation. Specifically, we count the logical qubits and Toffoli gates needed to carry out both the deepest circuit from among all the Hadamard test calculations we need to perform, and also the combination of all circuits of length $t = \delta, 2\delta, ... j_\text{max} \delta$ needed to build the spectrum in \cref{eq:green-via-discretetime-fourier}. For a given active space size $N$, target error $\epsilon$ for the peak (eigenvalue) shift, and maximum evolution time $t_j = \delta j$, we evaluate the following:

\begin{enumerate}
  \item Per-step cost: We count the Toffoli gates in one product formula step, summing the Givens rotation cost (using the fused circuits of Ref. \cite{caesura2025faster}, see~\cref{fig:givens_fused_adder}) and the $\sigma_z \otimes \sigma_z$ rotation cost (using QROM where beneficial, see~\cref{tab:QROM_savings}): these are the values already reported in \cref{fig:GRADE_per_step_cost_1eV}. As mentioned before, we use CDF throughout this section, given the large leakage errors obtained with GRADE and isometric THC.
  \item Trotter error: Following the perturbative approach to estimating peak position shift that we just described, we evaluate the expectation value of the leading-order BCH error coefficient for the chosen product formula, as per \cref{eq:perturbation_eigenvalues}. The expectation value is computed with respect to approximate eigenstates prepared with DMRG: we use bond dimension $100$ for $N \leq 18$; bond dimension $25$ for $N = 24$ and $N = 28$. Additional information on these calculations, including convergence studies on the bond dimension used and other parameters, is presented in~\cref{app:Convergence}.
  \item Number of steps:  From the error coefficient and the target $\epsilon$ (which we take to be $\epsilon = 1$ eV throughout), we determine the maximum allowable time step $\tau$. That is, we choose time step $\tau$ such that the expected error $\tau^p|\braket{Y_p}|$ is under the target error $\epsilon$ on average for the evaluated eigenstates. Then, we compute $r = \lceil t_j/\tau \rceil$.
  \item Total cost:  We multiply the number of steps $r$ by the per-step cost $C_{\text{step}}$ to get the overall resource estimate $C_\text{total}$.
\end{enumerate}
In the case of the Li-excess cluster XAS calculation, the parameters for the algorithm are as follows \cite{fomichev2025fast}: we use the maximum allowed peak position error of $\epsilon = 1$ eV; broadening of $\eta = 0.05$ Ha, discrete time signal time step $\delta = \pi / 4$ and $j_{\max} = 200$, giving a maximal evolution time for $t_{j_\text{max}} = \delta j_\text{max} = 50\pi$ a.u. Each technique described below improves $C_{\text{total}}$ in one of two ways: it either reduces the per-step cost $C_{\text{step}}$, or suppresses the Trotter error so that fewer steps $r$ are needed at the same accuracy $\epsilon$.

\subsection{Qubitization resource estimate\label{app:qubitized_xas}}

\begin{table}[t]
\centering
\begin{tabular}{|c||cccc||cccc|}
\hline
$N$ & \multicolumn{4}{c||}{$M=2N$} & \multicolumn{4}{c|}{$M=3N$} \\
\hline
 & Qubits & $\lambda$ & Toffoli & $\lambda \times$ Toffoli & Qubits & $\lambda$ & Toffoli & $\lambda \times$ Toffoli \\
\hline
6  & 145 & 38.96  & 853   & 33,231  & 150 & 38.80  & 1,021 & 39,619  \\
9  & 195 & 64.88  & 1,305 & 84,670  & 196 & 61.81  & 1,407 & 86,972  \\
10 & 210 & 40.04  & 1,447 & 57,945  & 211 & 39.43  & 1,557 & 61,388  \\
11 & 226 & 57.14  & 1,593 & 91,019  & 231 & 57.17  & 1,769 & 101,142 \\
14 & 271 & 108.76 & 1,897 & 206,316 & 276 & 124.96 & 2,261 & 282,528 \\
18 & 336 & 114.84 & 2,495 & 286,523 & 337 & 117.02 & 2,713 & 317,464 \\
24 & 427 & 149.36 & 3,365 & 502,600 & 432 & 151.56 & 3,707 & 561,815 \\
28 & 487 & 153.73 & 3,691 & 567,408 & 492 & 154.82 & 4,367 & 676,080 \\
\hline
\end{tabular}
\caption{Number of qubits and Toffoli gates needed to implement time evolution for unit time $t = 1$ a.u. via THC-based qubitization for the Li$_4$Mn$_2$O cluster with active space sizes $N$, for THC ranks $M = 2N$ and $M = 3N$, and one-norm $\lambda$. These results were generated with OpenFermion~\cite{mcclean2020openfermion} and PennyLane~\cite{bergholm2018pennylane}.}
\label{tab:thc_qubitization_results}
\end{table}

For the qubitization comparison in \cref{fig:trotter_vs_qubitization} we follow the symmetry-shifted THC block-encoding of Refs.~\cite{lee2021even,caesura2025faster}, using the same Li$_4$Mn$_2$O active spaces and Hamiltonians from \cref{ssec:xas_intro} with THC rank $M = 2N$. The per-walk-step Toffoli count and logical-qubit count are obtained with the PennyLane estimator~\cite{bergholm2018pennylane} at coefficient precision $\aleph = 13$ and rotation precision $\beth = 13$ bits, matching the $\epsilon = 1$\,eV peak-shift target used throughout this appendix. The total Toffoli count reported in \cref{fig:trotter_vs_qubitization} is then $\lceil \lambda\, t_{j_\text{max}} \rceil$ walk-operator applications times this per-step cost. Here, $\lambda$ is the THC 1-norm of the factorized Hamiltonian and $t_{j_\text{max}} = 50\pi$\,a.u. is selected as in the Trotter simulation, with the Chebyshev transform of \cref{app:chebyshev_transform} used in place of quantum signal processing to recover the spectrum from the walk signal.

\section{Rigorous statements of the main text}\label{app:proofs}

\subsection{The Suzuki hierarchy}

\begin{prop}[Suzuki hierarchy]\label{prop:Suzuki}

A product formula is said to achieve order $k$ if the leading order contribution in the effective Hamiltonian is $O(\tau^k)$. We will denote it by $U_k(\tau)$. The Suzuki hierarchy provides a systematic way of constructing product formulas of arbitrary even order \cite{suzuki1990fractal,suzuki1991general}
\begin{equation}
    U_{2k}(\tau) = U_{2k-2}^2(u_k\tau) U_{2k-2}((1-4u_k)\tau) U^2_{2k-2}(u_k\tau).
\end{equation}
where $u_k = 1/(4-4^{1/(2k-1)})$. 
\end{prop}

By definition of the order of a product formula,
\begin{equation}
    U_{2k-2}(\tau) = \exp(-i \tau H + (-i\tau)^{2k-1}Y^{(2k-2)}_{2k-1} +\ldots),
\end{equation}
Combining multiple $U_{2k-2}(\tau)$ steps in~\cref{eq:U_2k} using the BCH formula,
\begin{equation}\label{eq:U_n_exp}
    U_{2k}(\tau) = \exp(-i [4 u_k + (1-4u_k)]\tau H 
    + (-i\tau)^{2k-1} [4 u_k^{2k-1}+ (1-4u_k)^{2k-1}]Y^{(2k-2)}_{2k-1}+ \ldots),
\end{equation}
The scalar
$u_k$ represents the solution of the equation that cancels the $\tau^{2k-1}$ coefficient in the exponent of~\cref{eq:U_n_exp}:
\begin{equation}
    4 (u_k \tau)^{2k-1}+ [(1-4u_k)\tau]^{2k-1} = 0
\end{equation}
while respecting the correctness of the linear evolution term.
\begin{equation}
    4 u_k + (1-4u_k) = 1.
\end{equation}

Using the BCH expansion and the coefficients, we can write
\begin{align}\label{Y^{4}_5}
    Y^{(4)}_{5} &= \frac{\left(- 7094 \cdot 2^{\frac{2}{3}} - 1531 \cdot \sqrt[3]{2} + 13189\right) }{108 \left(- 31222 \cdot \sqrt[3]{2} - 39275 + 49524 \cdot 2^{\frac{2}{3}}\right)}\left[\left[H,Y^{(2)}_{3}\right],H\right] + \frac{\left(- 4061 \cdot 2^{\frac{2}{3}} - 5060 + 9128 \cdot \sqrt[3]{2}\right)}{36 \left(- 31222 \cdot \sqrt[3]{2} - 39275 + 49524 \cdot 2^{\frac{2}{3}}\right)}Y^{(2)}_{5}\nonumber\\
    &= -0.00405944185443219\left[\left[H,Y^{(2)}_{3}\right],H\right] - 0.074375995396295 Y^{(2)}_{5}.
\end{align}

\subsection{Order conditions.\label{app:order_conditions}}

\begin{prop}[Order conditions for a five-exponential symmetric formula with two fragments]\label{prop:order_conditions}
Consider the symmetric product formula
\begin{equation}\label{eq:5exp_formula}
  S(\tau) = e^{-ia_1\tau H_0}\, e^{-ib_1\tau H_1}\,
  e^{-ia_2\tau H_0}\, e^{-ib_1\tau H_1}\, e^{-ia_1\tau H_0},
\end{equation}
where $H = H_0 + H_1$, and $a_1, a_2, b_1 \in \mathbb{R}$ are free parameters. Then:
\begin{enumerate}
  \item \textbf{First-order (consistency) conditions.}
  $S(\tau)$ approximates $e^{-i\tau H}$ to first order if and only if
  \begin{equation}\label{eq:first_order_conds}
    2a_1 + a_2 = 1, \qquad 2b_1 = 1.
  \end{equation}

  \item \textbf{Second-order conditions.}
  Because $S(\tau)$ is manifestly palindromic, $S(\tau) = S^\dagger(-\tau)$,
  all even-order error terms in the effective Hamiltonian vanish
  automatically. The formula is therefore at least second order
  whenever~\cref{eq:first_order_conds} holds.

  \item \textbf{Third-order conditions.}
  Assuming~\cref{eq:first_order_conds}, the effective Hamiltonian is
  \begin{equation}\label{eq:5exp_eff_ham}
    S(\tau) = \exp\!\Big(-i\tau H
    + i\tau^3\big[c_{001}\,[H_0,[H_0,H_1]]
    + c_{101}\,[H_1,[H_0,H_1]]\big]
    + O(\tau^5)\Big),
  \end{equation}
  where
  \begin{align}
    c_{001} &= \frac{a_2^2\, b_1 - 2\,a_1\, b_1\,(a_1 + a_2)}{6},
    \label{eq:c001}\\[4pt]
    c_{101} &= \frac{a_2\, b_1^2 - 4\,a_1\, b_1^2}{6}.
    \label{eq:c101}
  \end{align}
  Setting $c_{001} = 0$ and $c_{101} = 0$
  (together with~\cref{eq:first_order_conds})
  yields the third-order conditions.
  Using $b_1 = 1/2$ and $a_2 = 1 - 2a_1$, these reduce to:
  \begin{align}
    c_{001} = 0 &\;\Longleftrightarrow\;
    6a_1^2 - 6a_1 + 1 = 0
    \quad\Rightarrow\quad
    a_1 = \frac{3 \pm \sqrt{3}}{6},
    \label{eq:cond1_reduced}\\[4pt]
    c_{101} = 0 &\;\Longleftrightarrow\;
    a_2 - 4a_1 = 0
    \quad\Rightarrow\quad
    a_1 = \tfrac{1}{6},\;a_2 = \tfrac{2}{3}.
    \label{eq:cond2_reduced}
  \end{align}
  Since $a_1 = 1/6$ does not satisfy~\cref{eq:cond1_reduced},
  the two conditions are incompatible:
  this five-exponential symmetric ansatz cannot achieve
  fourth order for generic $H_0$, $H_1$.
\end{enumerate}
\end{prop}

\begin{proof}
\textbf{First-order conditions.}
Summing the exponents at linear order in $\tau$:
\begin{align}
  (a_1 + a_2 + a_1)\,\tau H_0 &= \tau H_0 \;\Longrightarrow\; 2a_1 + a_2 = 1,\\
  (b_1 + b_1)\,\tau H_1 &= \tau H_1 \;\Longrightarrow\; 2b_1 = 1.
\end{align}

\textbf{Second-order conditions.}
The palindromic symmetry
$S(\tau) = S^\dagger(-\tau)$
forces the exponent to be an odd function of $\tau$
(up to the leading $-i\tau H$ term), so all
even-power corrections vanish identically.

\textbf{Third-order conditions.}
We evaluate the effective Hamiltonian by applying
the symmetric BCH expansion~\cref{eq:symmBCH} twice, working from the
inside out.

\emph{Step 1 (inner triple).}
Group the three inner exponentials as
$e^{B}\,e^{A_2}\,e^{B}$, where
$B = -ib_1\tau H_1$ and $A_2 = -ia_2\tau H_0$.
Identifying $X = 2B = -i(2b_1)\tau H_1$ and $Y = A_2 = -ia_2\tau H_0$,
the symmetric BCH gives
\begin{equation}
  Z_1 = X + Y
  - \tfrac{1}{24}[X,[X,Y]]
  - \tfrac{1}{12}[Y,[X,Y]] + O(\tau^5).
\end{equation}
Computing the required commutators:
\begin{align}
  [X, Y] &= -2a_2 b_1\,\tau^2\,[H_1, H_0],\\
  [X,[X,Y]] &= 4i\,a_2 b_1^2\,\tau^3\,[H_1,[H_1,H_0]],\\
  [Y,[X,Y]] &= 2i\,a_2^2 b_1\,\tau^3\,[H_0,[H_1,H_0]].
\end{align}
Using $[A,[B,C]] = -[A,[C,B]]$, we obtain
\begin{equation}\label{eq:Z1_result}
  Z_1 = -i\tau(a_2 H_0 + 2b_1 H_1)
  + \frac{i\,a_2^2 b_1\,\tau^3}{6}\,[H_0,[H_0,H_1]]
  + \frac{i\,a_2 b_1^2\,\tau^3}{6}\,[H_1,[H_0,H_1]]
  + O(\tau^5).
\end{equation}

\emph{Step 2 (outer triple).}
The full formula is
$e^{A_1}\,e^{Z_1}\,e^{A_1}$, where $A_1 = -ia_1\tau H_0$.
Setting $X' = 2A_1 = -i(2a_1)\tau H_0$ and $Y' = Z_1$:
\begin{equation}
  Z_{\mathrm{tot}} = X' + Y'
  - \tfrac{1}{24}[X',[X',Y']]
  - \tfrac{1}{12}[Y',[X',Y']] + O(\tau^5).
\end{equation}
For the commutator terms, only the leading part
of $Y'$ contributes at third order:
$Y'_{\mathrm{lead}} = -i\tau(a_2 H_0 + 2b_1 H_1)$.
The relevant commutators are:
\begin{align}
  [X', Y'] &= -4a_1 b_1\,\tau^2\,[H_0, H_1] + O(\tau^4),\\
  [X',[X',Y']] &= 8i\,a_1^2 b_1\,\tau^3\,[H_0,[H_0,H_1]],\\
  [Y',[X',Y']] &= 4i\,a_1 a_2 b_1\,\tau^3\,[H_0,[H_0,H_1]] + 8i\,a_1 b_1^2\,\tau^3\,[H_1,[H_0,H_1]].
\end{align}
The outer-layer third-order contributions are therefore
\begin{equation}\label{eq:outer_contrib}
  -\frac{i\,a_1 b_1(a_1{+}a_2)\,\tau^3}{3}\,[H_0,[H_0,H_1]]
  - \frac{2i\,a_1 b_1^2\,\tau^3}{3}\,[H_1,[H_0,H_1]].
\end{equation}

\emph{Step 3 (combining both layers).}
Adding the inner-layer contributions from~\cref{eq:Z1_result} to the outer-layer contributions~\cref{eq:outer_contrib}:
\begin{align}
  \text{Coeff.\ of } [H_0,[H_0,H_1]]:&\quad i\tau^3\!\left(
  \frac{a_2^2 b_1}{6} - \frac{a_1 b_1(a_1{+}a_2)}{3}\right) = i\tau^3 c_{001},\\[4pt]
  \text{Coeff.\ of } [H_1,[H_0,H_1]]:&\quad i\tau^3\!\left( \frac{a_2 b_1^2}{6} - \frac{2a_1 b_1^2}{3}\right) = i\tau^3 c_{101},
\end{align}
with $c_{001}$ and $c_{101}$ as in~\cref{eq:c001,eq:c101}. This establishes~\cref{eq:5exp_eff_ham}.

\emph{Incompatibility.}
Using $b_1 = 1/2$ and $a_2 = 1 - 2a_1$:
\begin{itemize}
  \item $c_{101} = 0$ requires $a_2 = 4a_1$, combined with $2a_1 + a_2 = 1$ gives $a_1 = 1/6$, $a_2 = 2/3$.
  \item $c_{001} = 0$ reduces to $6a_1^2 - 6a_1 + 1 = 0$, with roots $a_1 = (3 \pm \sqrt{3})/6 \approx 0.789$ or $0.211$.
\end{itemize}
Since $a_1 = 1/6$ does not satisfy $6(1/6)^2 - 6(1/6) + 1 = 1/6 \neq 0$, the two conditions cannot be simultaneously satisfied.
\end{proof}

\textbf{Remark.}
Despite this incompatibility, one may still choose parameters that cancel \emph{one} of the two third-order commutators. For instance, choosing $c_{101} = 0$ ($a_1 = 1/6$, $a_2 = 2/3$, $b_1 = 1/2$) eliminates the $[H_1,[H_0,H_1]]$ term, leaving a residual $c_{001} = 1/72$ on $[H_0,[H_0,H_1]]$. This selective cancellation is precisely the strategy exploited by the near-integrable formulas in~\cref{ssec:Near_integrability}: one prioritizes canceling the nested commutators whose fragments have the largest norms.

\subsection{Near-integrability and processing}

\begin{prop}[Error of the near-integrable formula]\label{prop:near_integrable_error}
Let $H = H_A + \alpha H_B$ with $\alpha \ll 1$. Define the near-integrable product formula
\begin{equation}
  \tilde{V}_{4,2}(\tau) = U_{4,A}^{1/2}(\tau)\, U_{2,B}(\tau)\, U_{4,A}^{1/2}(\tau),
\end{equation}
where $U_{4,A}(\tau)$ is a fourth-order formula for $H_A$ and $U_{2,B}(\tau)$ is a second-order formula for $\alpha H_B$. Then the effective Hamiltonian satisfies
\begin{equation}
  \tilde{V}_{4,2}(\tau) = \exp\!\left(-i\tau H + O(\tau^5) + O(\alpha\tau^3)\right).
\end{equation}
Moreover, using $U_{4,A}^{1/2}(\tau)$ instead of $U_{4,A}(\tau/2)$ preserves the correct linear term and does not change the product formula order.
\end{prop}

We analyze the components of $\tilde{V}_{4,2}(\tau)$ by writing the
effective Hamiltonians of the individual formulas. A fourth-order
formula for $H_A$ satisfies
\begin{equation}
 U_{4,A}(\tau) = \exp(-i\tau H_A + i\tau^5 Y_{5,A} + O(\tau^7)),
\end{equation}
and a second-order formula for $\alpha H_B$ satisfies
\begin{equation}
 U_{2,B}(\tau) = \exp(-i\alpha \tau H_B + i\alpha^3\tau^3 Y_{3,B}
 + O(\alpha^5\tau^5)),
\end{equation}
where $Y_{5,A}$ and $Y_{3,B}$ denote the linear combinations of nested
commutators forming the leading-order error in the respective product
formulas.

Now, applying the Baker-Campbell-Hausdorff (BCH) formula to the
composition $\tilde{V}_{4,2}(\tau) = U_{4,A}^{1/2}(\tau)\,
U_{2,B}(\tau)\, U_{4,A}^{1/2}(\tau)$, we obtain
\begin{equation}
  \tilde{V}_{4,2}(\tau) = \exp\bigg(-i \tau H + (-i\tau)^5 Y_{5,A}
  + (-i\tau)^3\alpha^3 Y_{3,B} 
  - (-i\tau)^3 \alpha^2 \frac{[H_B, [H_A, H_B]]}{12}
  - (-i\tau)^3 \alpha \frac{[H_A, [H_A, H_B]]}{24} + \ldots\bigg).
\end{equation}
The leading error is therefore $O(\tau^5 + \alpha\tau^3)$, as
claimed.

\begin{prop}[Processed near-integrable formula]\label{prop:processing}
Let $V_{4,2}(\tau)$ be the near-integrable formula of~\cref{prop:near_integrable_error}, with kernel effective Hamiltonian $K_\tau$. Define the processor
\begin{equation}
 P = \alpha \tau^2  \frac{[H_A, H_B]}{24}.
\end{equation}
Then the processed formula $U = e^{-P}e^{K}e^{P}$ generates an effective Hamiltonian satisfying
\begin{equation}
 H_{\mathrm{eff},\tau} = \tau(H_A + \alpha H_B) + O(\tau^5) + O(\alpha^2 \tau^3).
\end{equation}
The processor eliminates the $O(\alpha\tau^3)$ error term and halves the coefficient of the $O(\alpha^2\tau^3)$ term. In the regime $\alpha \sim \tau$, the processed formula achieves pseudo-fourth-order accuracy at a cost comparable to a second-order formula.
\end{prop}

The effective Hamiltonian induced by conjugation with the processor $e^P$ is \cite{blanes2006composition}
\begin{equation}
  H_{\mathrm{eff},\tau} = \sum_{k=0}^\infty
  \frac{(-1)^k}{k!}\mathrm{ad}_{P_\tau}^k K_\tau
  = K_\tau - [P_\tau, K_\tau]
  + \frac{1}{2!}[P_\tau, [P_\tau, K_\tau]] + \ldots
\end{equation}
Adding a subindex to denote the $\tau$-order, the order conditions give
\begin{align}
  H_{\mathrm{eff},\tau,1} &= \tau(H_A + \alpha H_B),
  \quad \Rightarrow \quad P_{\tau,1} = 0, \\
  H_{\mathrm{eff},\tau,2} &= 0, \\
  H_{\mathrm{eff},\tau,3} &= -\tau^3 \left(\alpha^2
  \frac{[H_B, [H_A, H_B]]}{12} + \alpha
  \frac{[H_A, [H_A, H_B]]}{24}\right)  + \tau^3\alpha^3 Y_{B,3}
  - \tau [P_{\tau,2}, H_A + \alpha H_B].
\end{align}

We require $P_{\tau,2}$ to cancel the $\tau^3\alpha$ term. Choosing
$P = \alpha \tau^2  \frac{[H_A, H_B]}{24}$, we compute
\begin{equation}
  -\tau [P_{\tau,2}, (H_A + \alpha H_B)] = \tau [(H_A + \alpha H_B), P_{\tau,2}] 
  = \alpha \tau^3 \frac{[H_A, [H_A, H_B]]}{24}
  + \alpha^2 \tau^3 \frac{[H_B, [H_A, H_B]]}{24}.
\end{equation}
Adding this to $H_{\mathrm{eff},\tau,3}$ cancels the $O(\alpha \tau^3)$ error term exactly, and the $O(\alpha^2 \tau^3)$ coefficient is reduced from $\frac{1}{12}$ to $\frac{1}{12} - \frac{1}{24} = \frac{1}{24}$, i.e., halved. The processed formula therefore has error $O(\tau^5 + \alpha^2\tau^3)$, as claimed.

\begin{prop}[QSVT-stabilized leakage rate]\label{prop:qsvt_rate}
Let $P_{\mathrm{stab}}$ be a degree-$d$ (odd) stabilizing polynomial applied to the Trotter step~$U(\tau)$ via QSVT. The resulting enhanced step uses~$d$ queries to~$U$ and advances the simulation time by~$\tau$ (the stabilizing polynomial preserves the eigenphases of the physical-subspace block).

A degree-$d$ odd polynomial satisfying $P_{\mathrm{stab}}(1)=1$ has $\frac{d+1}{2}$ free coefficients. Imposing the flatness conditions $P_{\mathrm{stab}}^{(k)}(1) = 0$ for $k = 1,\ldots,\frac{d-1}{2}$ forces the deviation from unity to satisfy $1 - P_{\mathrm{stab}}(1-\delta) = O(\delta^{(d+1)/2})$. Since each singular-value deficit is $\delta_j = O(\tau^2)$, the per-step leakage probability is
\begin{equation}
  \tilde{p}_k \;=\; O\!\bigl(\tau^{2\cdot(d+1)/2}\bigr) \;=\; O\!\bigl(\tau^{d+1}\bigr),
\end{equation}
and the spectral-weight transfer rate is
\begin{equation}
  \Gamma_k' \;=\; \frac{\tilde{p}_k}{2\,\tau} \;=\; O\!\bigl(\tau^{d}\bigr). \label{eq:gamma_qsvt}
\end{equation}
For the minimal non-trivial choice $d = 3$ this gives $\tilde{p}_k = O(\tau^{4})$ and $\Gamma_k' = O(\tau^{3})$, an improvement of two orders in~$\tau$ over the unprotected rate $\Gamma_k = O(\tau)$, at a multiplicative query overhead of~$\times 3$ (three queries per~$\tau$ of simulation time, compared with one query per~$\tau$ for the unprotected formula).
\end{prop}

Write $D = \Pi_{\text{ph}} U \Pi_{\text{ph}}$ and $G = Q_{\text{ph}} U \Pi_{\text{ph}}$ for the physical and leakage blocks of a single Trotter step.

\paragraph{Step 1: The operator identity.}
Since $U$ is unitary and $\Pi_{\text{ph}} + Q_{\text{ph}} = 1$,
\begin{equation}\label{eq:block_identity}
  D^\dagger D + G^\dagger G = \Pi_{\text{ph}} U^\dagger \Pi_{\text{ph}}\cdot \Pi_{\text{ph}} U \Pi_{\text{ph}} + \Pi_{\text{ph}} U^\dagger Q_{\text{ph}}\cdot Q_{\text{ph}} U \Pi_{\text{ph}} = \Pi_{\text{ph}} U^\dagger\!(\Pi_{\text{ph}}+Q_{\text{ph}})\,U\Pi_{\text{ph}} = \Pi_{\text{ph}}.
\end{equation}
In words: for any physical state, the probability of staying plus the probability of leaking equals one.

\paragraph{Step 2: Singular values of $D$.}
Since $G^\dagger G = \Pi_{\text{ph}} - D^\dagger D$, the two operators commute and share an eigenbasis $\{|w_j^R\rangle\}$ on the physical subspace:
\begin{equation}
  D^\dagger D\,|w_j^R\rangle = \sigma_j^2\,|w_j^R\rangle, \qquad G^\dagger G\,|w_j^R\rangle = (1-\sigma_j^2)\,|w_j^R\rangle.
\end{equation}
Here $\sigma_j$ are the singular values of $D$ and the $|w_j^R\rangle$ are its right singular vectors. Since each factor of $U$ contributes $O(\tau)$ leakage amplitude, $\|G\| = O(\tau)$, so every eigenvalue of $G^\dagger G$ is at most $O(\tau^2)$:
\begin{equation}\label{eq:sigma_near_one}
  1 - \sigma_j^2 = O(\tau^2) \qquad\Longrightarrow\qquad \delta_j \coloneqq 1 - \sigma_j = O(\tau^2) \quad\text{for all } j.
\end{equation}

\paragraph{Step 3: The polynomial suppresses the deviation.}
QSVT replaces each singular value $\sigma_j$ with $P_{\mathrm{stab}}(\sigma_j)$, while preserving the right singular vectors $|w_j^R\rangle$. For general odd degree~$d$, the flatness conditions $P_{\mathrm{stab}}^{(k)}(1)=0$ for $k=1,\ldots,\tfrac{d-1}{2}$ guarantee that
\begin{equation}\label{eq:new_delta_general}
  \tilde{\delta}_j \coloneqq 1 - P_{\mathrm{stab}}(\sigma_j) = O\!\bigl(\delta_j^{(d+1)/2}\bigr) = O\!\bigl(\tau^{d+1}\bigr).
\end{equation}
For the minimal case $d=3$, $P_{\mathrm{stab}}$ has two free coefficients (after imposing oddness and $P(1)=1$), leaving one flatness condition $P'(1)=0$. The unique solution is $P_{\mathrm{stab}}(x) = \tfrac{3}{2}\,x - \tfrac{1}{2}\,x^3$. Writing $\sigma_j = 1 - \delta_j$:
\begin{equation}
  P_{\mathrm{stab}}(1-\delta_j) = \tfrac{3}{2}(1-\delta_j) - \tfrac{1}{2}(1-\delta_j)^3 = 1 - \tfrac{3}{2}\,\delta_j^2 + \tfrac{1}{2}\,\delta_j^3,
\end{equation}
so the new deviation from~$1$ is
\begin{equation}\label{eq:new_delta}
  \tilde{\delta}_j = \tfrac{3}{2}\,\delta_j^2 + O(\delta_j^3) = O(\tau^4),
\end{equation}
consistent with~\cref{eq:new_delta_general} at $d=3$, since $(d+1)/2 = 2$.

\paragraph{Step 4: New leakage probability.}
Let $\tilde{D}$ denote the QSVT-transformed contraction. It has the same right singular vectors $|w_j^R\rangle$ and new singular values $P_{\mathrm{stab}}(\sigma_j)$, so $\tilde{D}^\dagger\tilde{D}\,|w_j^R\rangle = P_{\mathrm{stab}}(\sigma_j)^2\,|w_j^R\rangle$. Expanding $|E_k\rangle = \sum_j c_{kj}\,|w_j^R\rangle$ with $\sum_j|c_{kj}|^2 = 1$:
\begin{equation}
  \tilde{p}_k = 1 - \|\tilde{D}|E_k\rangle\|^2 = \sum_j |c_{kj}|^2 \bigl(1 - P_{\mathrm{stab}}(\sigma_j)^2\bigr).
\end{equation}
Each term in the sum satisfies
\begin{equation}
  1 - P_{\mathrm{stab}}(\sigma_j)^2 = \underbrace{(1 - P_{\mathrm{stab}}(\sigma_j))}_{\tilde{\delta}_j} \;\underbrace{(1 + P_{\mathrm{stab}}(\sigma_j))}_{\approx\;2} = 2\tilde{\delta}_j + O(\tilde{\delta}_j^2) = O\!\bigl(\tau^{d+1}\bigr).
\end{equation}
Since every term is $O(\tau^{d+1})$ and the weights sum to~$1$:
\begin{equation}
  \tilde{p}_k = O\!\bigl(\tau^{d+1}\bigr).
\end{equation}
For $d=3$: $\tilde{p}_k \leq 3\delta_j^2 + O(\delta_j^3) = O(\tau^4)$.

\paragraph{Step 5: Spectral-weight transfer rate.}
Each enhanced step uses $d$ queries to~$U$, but advances the simulation time by only~$\tau$ (the stabilizing polynomial preserves the eigenphases of the physical-subspace contraction). After $n' = t/\tau$ enhanced steps, the physical-subspace spectral weight of eigenstate~$k$ decays as
\begin{equation}
  \bigl(1 - \tilde{p}_k\bigr)^{n'} \;\approx\; e^{-2\,\Gamma_k'\,t}, \qquad \Gamma_k' \;=\; \frac{\tilde{p}_k}{2\,\tau} \;=\; \frac{O(\tau^{d+1})}{2\,\tau} \;=\; O\!\bigl(\tau^{d}\bigr). \label{eq:gamma_qsvt_appendix}
\end{equation}
For $d=3$: $\Gamma_k' = O(\tau^3)$. The total query cost is $d\,n' = d\,t/\tau$; for $d=3$ this is a $\times 3$ multiplicative overhead relative to the unprotected formula.

\paragraph{Remark: robustness to non-normality of $D$.}
QSVT transforms the singular values $\sigma_j$ of $D$, but the spectroscopy signal ultimately depends on its eigenvalue magnitudes~$|\lambda_j|$. For a normal matrix these coincide; for $D = \Pi_{\text{ph}}U\Pi_{\text{ph}}$ they differ. However, $D$ is close to normal: from~\cref{eq:block_identity}, $\|[D,D^\dagger]\| = \|G^\dagger G - (\Pi_{\text{ph}} U Q_{\text{ph}})(\Pi_{\text{ph}} U Q_{\text{ph}})^\dagger\| \leq \|G\|^2 + \|\Pi_{\text{ph}} U Q_{\text{ph}}\|^2 = O(\tau^2)$, which implies $|\sigma_j - |\lambda_j|| = O(\tau^2)$. The flatness condition $P_{\mathrm{stab}}'(1) = 0$ protects against this: by the mean value theorem,
\begin{equation}
  |P_{\mathrm{stab}}(\sigma_j) - P_{\mathrm{stab}}(|\lambda_j|)| = \underbrace{|P_{\mathrm{stab}}'(\xi)|}_{O(\delta_j)\,=\,O(\tau^2)} \;\cdot\; \underbrace{|\sigma_j - |\lambda_j||}_{O(\tau^2)} = O(\tau^4) = O(\tilde{\delta}_j).
\end{equation}
The correction from non-normality is the same order as the stabilized deviation itself, so the $O(\tau^{d})$ linewidth is not degraded.

\paragraph{Remark: QSVT circuit.}
The polynomial $P_{\mathrm{stab}}$ is real, odd, and satisfies $|P_{\mathrm{stab}}(x)| \le 1$ for $x \in [-1, 1]$. By the QSVT existence theorem \cite{gilyen2019quantum}, there exist phase angles $\phi_0, \ldots, \phi_d$ such that the circuit
\begin{equation}
  U_{\vec{\phi}} = R_0\, U\, R_1\, U^\dagger\, R_2\, U\, R_3, \qquad R_k = e^{i\phi_k(2\Pi_{\text{ph}} - I)},
\end{equation}
implements $P_{\mathrm{stab}}$ on the singular values of $D = \Pi_{\text{ph}}U\Pi_{\text{ph}}$, where $\Pi_{\text{ph}} = \ket{0}\!\bra{0}_{\mathrm{aux}}\otimes 1_{\text{ph}}$. This uses $d = 3$ queries to $U$ and no ancilla: the phase gates act directly on the auxiliary-mode qubits, and are of the same form as the symmetry-protection conjugations. Amplitude that leaks to the auxiliary subspace during any enhanced step does not contribute to the physical-subspace spectroscopy signal; this is the sole source of the signal decay quantified above.

\paragraph{Remark: eigenphase fidelity under projection.}
The preceding remark bounds how non-normality affects the eigenvalue \emph{magnitudes} $|\lambda_j|$. A complementary question is whether the eigenvalue \emph{phases} $\theta_k$ of $D = \Pi_{\text{ph}} U \Pi_{\text{ph}}$ coincide with those of the full-space effective Hamiltonian $H_{\mathrm{eff}}$. Writing $H_{\mathrm{eff}}$ in block form with respect to $\Pi_{\text{ph}}$ and $Q_{\text{ph}} = 1 - \Pi_{\text{ph}}$, the standard Schur complement (or L\"owdin partitioning) gives the effective physical-subspace Hamiltonian \cite{lowdin1951note}
\begin{equation}\label{eq:downfolding}
  H_{\mathrm{eff,phys}} = P H_{\mathrm{eff}} P + P H_{\mathrm{eff}} Q\, \frac{1}{E - Q H_{\mathrm{eff}} Q}\, Q H_{\mathrm{eff}} P + \cdots,
\end{equation}
whose eigenvalues determine the eigenphases $\theta_k/\tau$ of $D$.

\subsection{Derivation of the Hadamard lemma}\label{app:hadamard}
\begin{lem}[Hadamard lemma]\label{lemma:Hadamard}
For any two operators $X$ and $Y$, we have
\begin{equation}\label{eq:hadamard_lemma_app}
 e^{X} Y e^{-X} = \sum_{j=0}^{\infty} \frac{1}{j!}\,\mathrm{ad}_{X}^{j}(Y),
\end{equation}
where the nested commutator (adjoint action) is defined recursively as
\begin{equation}
 \mathrm{ad}_{X}^{0}(Y) = Y, \qquad
 \mathrm{ad}_{X}^{j}(Y) = [X,\,\mathrm{ad}_{X}^{j-1}(Y)].
\end{equation}
\end{lem}

\begin{proof}
Define the one-parameter family
\begin{equation}\label{eq:f_def}
 f(t) = e^{tX}\, Y\, e^{-tX}.
\end{equation}
We proceed in three steps.

\paragraph{Step 1: Differential equation.}
Differentiating~\cref{eq:f_def} with respect to $t$:
\begin{equation}
 f'(t) = X\, e^{tX}\, Y\, e^{-tX} - e^{tX}\, Y\, e^{-tX}\, X = [X,\, f(t)] = \mathrm{ad}_{X}\!\bigl(f(t)\bigr).
\end{equation}

\paragraph{Step 2: Higher derivatives.}
Iterating, the $j$-th derivative satisfies
\begin{equation}
 f^{(j)}(t) = \mathrm{ad}_{X}^{j}\!\bigl(f(t)\bigr).
\end{equation}
This follows by induction: assuming $f^{(j)}(t) = \mathrm{ad}_{X}^{j}(f(t))$, we have
\begin{equation}
  f^{(j+1)}(t) = \frac{d}{dt}\,\mathrm{ad}_{X}^{j}\!\bigl(f(t)\bigr) = \mathrm{ad}_{X}^{j}\!\bigl(f'(t)\bigr) = \mathrm{ad}_{X}^{j}\!\bigl([X, f(t)]\bigr) = \mathrm{ad}_{X}^{j+1}\!\bigl(f(t)\bigr),
\end{equation}
where the second equality uses the linearity of the commutator and the fact that $X$ is $t$-independent. Evaluating at $t=0$ (where $f(0) = Y$):
\begin{equation}\label{eq:jth_deriv}
 f^{(j)}(0) = \mathrm{ad}_{X}^{j}(Y).
\end{equation}

\paragraph{Step 3: Taylor series.}
Taylor-expanding $f(t)$ around $t=0$ and setting $t=1$:
\begin{equation}
 f(1) = e^{X}\, Y\, e^{-X} = \sum_{j=0}^{\infty} \frac{f^{(j)}(0)}{j!} = \sum_{j=0}^{\infty} \frac{1}{j!}\,\mathrm{ad}_{X}^{j}(Y).
\end{equation}
This completes the proof.
\end{proof}

\subsection{Symmetry protection}

\begin{prop}[Symmetry-protected effective Hamiltonian]\label{prop:symmetry_protection}
Let $U_2(\tau)$ be a second-order Trotter step implementing an effective Hamiltonian $H_{\mathrm{eff}} = \sum_{\Delta m} Y_{\Delta m}$, decomposed by the number of electrons excited to or from the $N$ physical orbitals. Define
\begin{equation}
  U_2(\tau, \phi) = e^{-i\phi Q_{\text{ph}}}\, U_2(\tau)\, e^{+i\phi Q_{\text{ph}}},
\end{equation}
where $\Pi_{\text{ph}} = \ket{0}\bra{0}_{\text{aux}} $ is the binary projector distinguishing the physical from auxiliary subspace, and $Q_{\text{ph}} = 1- \Pi_{\text{ph}}$. Then:
\begin{enumerate}
  \item A forward cycle 
  \begin{equation}
  U_F(\tau) = \prod_{k=0}^{1} U_2(\tau, \pi k)
  \end{equation}
  achieves effective Hamiltonian
  \begin{equation}
  H_{\mathrm{eff}}^{(1)} = H + O(\tau).
  \end{equation}
  \item The symmetric super-cycle $U_{\mathrm{symm}} = U_B(\tau)\, U_F(\tau)$, where $U_B(\tau) = U_F(-\tau)^\dagger$, achieves
  \begin{equation}
    H_{\mathrm{eff}}^{(2)} = H + \underbrace{\tau^2 Y_{\mathrm{leak}}^{(2)}}_{O(\tau^2)} + \underbrace{\tau^2 Y_{\mathrm{Trotter}}^{(2)}}_{O(\tau^2)} + O(\tau^3),
  \end{equation}
  which is on par with the intrinsic Trotter error of the second-order formula.
\end{enumerate}
\end{prop}

We prove the result for the general number operator $N_{\text{ph}}$ with a $C$-step cycle; the simplified binary operator $\Pi_{\text{ph}}$ with $C$=2 then follows as described after the proposition statement in the main text.

The first step is to conjugate Trotter steps with different angles:
\begin{equation}
  U_2(\tau, \phi) = e^{-i \phi N_{\text{aux}}} U_2(\tau)
  e^{+i \phi N_{\text{aux}}}.
\end{equation}
Decomposing $H_{\mathrm{eff}} = \sum_{\Delta m} Y_{\Delta m}$ by the number of electrons exchanged with the auxiliary orbitals, we use the fundamental identity $[N_e-N_{\text{ph}}, Y_{\Delta m}] = (\Delta m) Y_{\Delta m}$ to obtain
\begin{equation}
  e^{-i \phi N_{\text{aux}}} Y_{\Delta m} e^{+i \phi N_{\text{aux}}}
  = Y_{\Delta m} \sum_{j=0}^{\infty}
  \frac{(-i\phi \Delta m)^j}{j!}
  = Y_{\Delta m} e^{-i \Delta m \phi}.
\end{equation}

\emph{Forward cycle.}
We define a $C$-step forward cycle
\begin{equation}
  U_F(\tau) = \prod_{k=0}^{C-1} U_2\!\left(\tau, \frac{2\pi k}{C}\right).
\end{equation}
Since the roots of unity satisfy $\sum_{k=0}^{C-1} e^{-i \Delta m (2\pi k/C)} = 0$ for $0 < |\Delta m| < C$, the zeroth-order leakage contribution is cancelled, yielding
\begin{equation}
  H_{\mathrm{eff}}^{(1)} = H
  + \underbrace{\tau Y_{\mathrm{leak}}^{(1)}}_{O(\tau)}
  + \underbrace{\tau^2 Y_{\mathrm{Trotter}}^{(2)}}_{O(\tau^2)}
  + O(\tau^3).
\end{equation}

\emph{Symmetric super-cycle.}
Defining the backward cycle $U_B(\tau) = U_F^\dagger(-\tau)$, whose effective Hamiltonian has odd-powered error terms with flipped sign, the composition $U_{\mathrm{symm}} = U_B(\tau)\, U_F(\tau)$ cancels the $O(\tau)$ term:
\begin{equation}
  H_{\mathrm{eff}}^{(2)} = H
  + \tau^2 (Y_{\mathrm{leak}}^{(2)} + Y_{\mathrm{Trotter}}^{(2)})
  + O(\tau^3). 
\end{equation}

\subsection{Runge-Kutta-Nystr\"{o}m conditions for kinetic and potential splittings of the Hamiltonian\label{app:RKN}}

\begin{prop}[Vanishing nested commutators for kinetic-potential splittings]\label{prop:VVT_vanishing}
Let $H = T + V$ act on $L^2(\mathbb{R}^d)$, where $T = \sum_{j=1}^{d} \frac{p_j^2}{2m_j} +f(\mathbf{q})$ is the kinetic energy operator and $V = V(\mathbf{q})$ is a smooth, real-valued multiplicative potential. Assign to each operator a \emph{differential degree}: $\mathrm{deg}(T) = 2$ and $\mathrm{deg}(V) = 0$. Then:
\begin{enumerate}
  \item $[V, T]$ is a first-order differential operator, i.e.\ $\mathrm{deg}([V,T]) = 1$.
  \item $[V, [V, T]] = -\sum_{j=1}^{d} \frac{\hbar^2}{m_j} \left(\frac{\partial V}{\partial q_j}\right)^{\!2}$ is a multiplication operator, i.e.\ $\mathrm{deg}([V,[V,T]]) = 0$.
  \item The triple commutator vanishes identically: $[V, [V, [V, T]]] = 0$.
\end{enumerate}

This algebraic constraint reduces the number of independent error terms in the Baker--Campbell--Hausdorff expansion of any product formula built from $e^{-i\alpha T}$ and $e^{-i\beta V}$.
\end{prop}

\begin{proof}
We work in $d$ spatial dimensions with coordinates $\mathbf{q} = (q_1, \ldots, q_d)$ and conjugate momenta $p_j = -i\hbar\, \partial / \partial q_j$.

\paragraph{Part~1: $[V, T]$ is a first-order differential operator.}
For any smooth $\psi \in L^2(\mathbb{R}^d)$ and with the isotropic kinetic energy $T = -\frac{\hbar^2}{2m}\nabla^2$ (the anisotropic case is analogous),
\begin{equation}\label{eq:VT_commutator}
  [V, T]\,\psi = V\!\left(-\tfrac{\hbar^2}{2m}\nabla^2\psi\right) - \left(-\tfrac{\hbar^2}{2m}\right)\nabla^2(V\psi).
\end{equation}
Expanding the Laplacian of a product, $\nabla^2(V\psi) = (\nabla^2 V)\psi + 2(\nabla V)\cdot(\nabla\psi) + V\nabla^2\psi$, and substituting yields
\begin{equation}
  [V, T]\,\psi
  = \frac{\hbar^2}{2m}
    \bigl[(\nabla^2 V)\,\psi + 2\,(\nabla V)\cdot(\nabla\psi)\bigr].
\end{equation}
The right-hand side contains $\nabla\psi$ (first order) and a zeroth-order multiplication term, but no second derivatives of $\psi$. Hence $[V, T]$ is a first-order differential operator.

\paragraph{Part~2: $[V, [V, T]]$ is a multiplication operator.}
Define $D_1 \coloneqq [V, T] = \frac{\hbar^2}{2m}[(\nabla^2 V) + 2(\nabla V)\cdot\nabla]$. Computing $[V, D_1]$:
\begin{align}
  [V, D_1]\,\psi
  &= V\cdot D_1(\psi) - D_1(V\psi) \nonumber \\
  &= \frac{\hbar^2}{2m}\Big\{ V\bigl[(\nabla^2 V)\psi + 2(\nabla V)\cdot\nabla\psi\bigr] - \bigl[(\nabla^2 V)(V\psi) + 2(\nabla V)\cdot\nabla(V\psi)\bigr] \Big\} \nonumber \\
  &= \frac{\hbar^2}{2m}\Big\{2V(\nabla V)\cdot\nabla\psi - 2(\nabla V)\cdot\bigl[(\nabla V)\psi + V\nabla\psi\bigr] \Big\} \nonumber \\
  &= -\frac{\hbar^2}{m}\,|\nabla V|^2\;\psi.
\end{align}
Therefore
\begin{equation}\label{eq:VVT_result}
  [V, [V, T]] = -\frac{\hbar^2}{m}\,|\nabla V|^2,
\end{equation}
which is a multiplication operator -- a function of $\mathbf{q}$ alone with no derivatives acting on $\psi$. In the anisotropic case $T = \sum_j p_j^2/(2m_j)$, the same calculation gives $[V, [V, T]] = -\sum_{j=1}^{d} \frac{\hbar^2}{m_j}
\bigl(\partial V / \partial q_j\bigr)^2$.

\paragraph{Part~3: $[V, [V, [V, T]]] = 0$.}
Both $V$ and $[V,[V,T]] = -(\hbar^2/m)\,|\nabla V|^2$ are multiplication operators -- smooth functions of $\mathbf{q}$. Since multiplication operators commute,
\begin{equation}
  [V, [V, [V, T]]]
  = \bigl[V,\; -\tfrac{\hbar^2}{m}\,|\nabla V|^2\bigr]
  = 0. \qedhere
\end{equation}
\end{proof}

\paragraph{Remark.}
This result is the quantum-mechanical counterpart of the classical identity $\{V, \{V, \{V, T\}\}\} = 0$ for Poisson brackets, proved in \cite{mclachlan2019lie}. The underlying mechanism is identical: the kinetic energy is a second-order object (quadratic in momenta classically, second-order differential operator quantum-mechanically), while the potential is zeroth-order. Each commutation with $V$ reduces the differential degree by one, so after two commutations the result is a multiplication operator, and one further commutation with $V$ necessarily vanishes.

McLachlan and Murua show \cite[Section~4]{mclachlan2019lie} that the Lie algebra generated by the Laplacian and a multiplicative potential under the commutator bracket is isomorphic to the Lie algebra of classical mechanics (class~$\mathcal{P}$). Consequently, all structural identities of that algebra -- including the vanishing of all brackets between degree-zero elements -- transfer directly to the quantum setting.

\paragraph{Plane-wave basis.}
In a plane-wave basis $\{|\mathbf{G}\rangle\}$, the kinetic energy $T$ is diagonal with entries $t_{\mathbf{G}} = \hbar^2|\mathbf{G}+\mathbf{k}|^2/(2m)$, and $V$ is a dense matrix with entries $V_{\mathbf{G},\mathbf{G}'} = V(\mathbf{G}-\mathbf{G}')$. In the complete (untruncated) basis, the identity $[V,[V,[V,T]]] = 0$ holds as a matrix equation. Explicitly, defining $f(\mathbf{q}) \coloneqq [V,[V,T]] = -(\hbar^2/m)\,|\nabla V|^2$, the matrix elements read
\begin{equation}
  [V,[V,[V,T]]]_{\mathbf{G},\mathbf{G}'} = \sum_{\mathbf{K}} V(\mathbf{G}{-}\mathbf{K})\, f(\mathbf{K}{-}\mathbf{G}') - f(\mathbf{G}{-}\mathbf{K})\,V(\mathbf{K}{-}\mathbf{G}'),
\end{equation}
which vanishes because $f$ is a multiplication operator: both sums reduce to discrete convolutions representing $V \cdot f$ and $f \cdot V$, which are equal.

\paragraph{Discretization and the failure of the identity.}
In a \emph{truncated} plane-wave basis ($|\mathbf{G}| \le G_{\max}$, dimension~$N$), the identity $[V,[V,[V,T]]] = 0$ does \emph{not} hold -- under either the Galerkin projection ($PVP$, $PTP$) or the split-operator / dual-basis implementation ($\tilde{V}$~diagonal in position space, $T$~diagonal in momentum space, connected by the QFT). The reason is structural: the mechanism underlying the continuum proof breaks down at finite~$N$.

In the continuum, $[V,[V,T]] = -|\nabla V|^2$ is a multiplication operator -- diagonal in position space -- so it commutes with $V$, and the triple commutator vanishes. In the dual-basis discretization, the position-space matrix
elements of $[V,[V,T]]$ are
\begin{equation}\label{eq:VVT_pos_space}
  [\tilde{V},[\tilde{V},T]]_{jk} = \bigl(V(r_j) - V(r_k)\bigr)^{\!2}\, T_{jk},
\end{equation}
where $T_{jk} = N^{-1}\sum_{\mathbf{G}} (|\mathbf{G}|^2/2m)\, e^{i\mathbf{G}\cdot(r_j - r_k)}$ is the spectral kinetic-energy kernel. Crucially, the diagonal ($j=k$) entries of Eq.~\eqref{eq:VVT_pos_space} are identically zero (since $(V_j - V_j)^2 = 0$), so $[\tilde{V},[\tilde{V},T]]$ is an entirely off-diagonal operator at finite~$N$ -- it is not a multiplication operator. The continuum cancellation rests on the distributional identity $(x{-}y)^3\,\delta''(x{-}y) = 0$; in the discrete case, $T_{jk}$ is the Fourier transform of $|\mathbf{G}|^2$ and has support on \emph{all} grid points, so no analogous cancellation occurs, and
\begin{equation}
  [\tilde{V},[\tilde{V},[\tilde{V},T]]]_{jk}
  = \bigl(V(r_j) - V(r_k)\bigr)^{\!3}\, T_{jk}
  \;\neq\; 0.
\end{equation}

\paragraph{Practical impact.}
Despite this formal breakdown, the identity remains useful in two important respects.
\begin{enumerate}
\item \textbf{Formula design.}
  Product formula coefficients $\{a_k, b_k\}$ are determined algebraically to cancel all BCH error terms through order~$p$. The identity $[V,[V,[V,T]]] = 0$ reduces the number of independent order conditions at each order~$n$. The resulting formulas require fewer stages and are valid regardless of the discretization: the order-$p$ cancellations are enforced by the choice of coefficients, not by the identity itself.

\item \textbf{Subdominance of the residual error.}
  At order~4, the three generic commutators are
  $[V,[V,[V,T]]]$, $[T,[V,[V,T]]]$, and $[T,[T,[T,V]]]$. In a plane-wave basis, $\|T\| \sim G_{\max}^2 \gg \|V\|$, so $\|[T,[T,[T,V]]]\|$ dominates the error. Numerically, $\|[V,[V,[V,T]]]\|$ contributes less than~$0.1\%$ of the total order-4 norm sum even at modest grid sizes, and this fraction decreases with~$N$. Thus, error bounds that ignore the $[V,[V,[V,T]]]$ term (as justified in the continuum) overestimate the discrete error by a negligible amount.
\end{enumerate}

\paragraph{QFT-defined momentum on a finite grid.}
A natural question is whether the identity is restored when the momentum operator $p$ is defined \emph{exactly} as the quantum Fourier transform (QFT) conjugate of the position operator $q$ on an $N$-point grid  --  the construction used in split-operator quantum simulations. In this setting, $V = \mathrm{diag}\bigl(V(x_0),\dots,V(x_{N-1})\bigr)$ is diagonal in the position basis, and the kinetic operator $T = F^\dagger \mathrm{diag}(k_n^2/2m)\, F$ is diagonal in the momentum basis, where $F$ is the $N$-point DFT matrix and
$k_n = 2\pi n/L$ are the momentum eigenvalues.

The answer is \emph{no}: the identity is violated at every finite grid size $N$. Because $V$ is diagonal, the nested commutator has matrix elements
\begin{equation}
  \bigl(\mathrm{ad}_V^{\,3}(T)\bigr)_{jk}
  = (V_j - V_k)^3 \, T_{jk}\,,
  \label{eq:ad3_grid}
\end{equation} where $T_{jk} = N^{-1}\sum_n (k_n^2/2m)\,e^{2\pi i\, n(j-k)/N}$ is a circulant matrix with \emph{global support}: $T_{jk} \neq 0$ for every $j \neq k$. For any smooth, non-constant potential, $V_j \neq V_k$ for most pairs $(j,k)$, so the product $(V_j - V_k)^3 \, T_{jk} \neq 0$ and the triple commutator does not vanish. The practical consequence, however, remains unchanged: because $\lVert T \rVert \sim G_{\max}^2 \gg \lVert V \rVert$ in plane-wave calculations, the $[V,[V,[V,T]]]$ term constitutes a negligible fraction of the total order-$4$ Trotter error regardless of how $T$ is discretized.

\section{Effect of different randomization procedures in spectroscopy problems\label{app:randomized_spectra}}

\subsection{Randomized product formulas\label{sec:randomized}}

As we discussed above, the effect of certain forms of randomization will depend on the application. Here we focus on spectroscopy \cite{fomichev2024simulating,fomichev2025fast}. Our goal in spectroscopy is to evaluate
\begin{equation}
  - \eta \operatorname{Im} G_\rho(\omega) = \frac{\eta\tau}{2\pi} \sum_{j=-\infty}^{\infty} e^{-\eta\tau|j|} \tilde{G}(\tau j) e^{ij\tau\omega},
\end{equation}
where $\eta$ is some dampening factor, $\tau$ is the time step and 
\begin{equation}
  \tilde{G}_\rho(\tau j) = \frac{\bra{I} m_\rho e^{-i H \tau j}  m_{\rho} \ket{I}}{\| m_\rho \ket{I} \|^2}
\end{equation}
for a given Hamiltonian $H$.
In \cite{fomichev2025fast} the authors explain that deterministic $p$-order product formulas implement exact Hamiltonian simulation of an approximate Hamiltonian
\begin{equation}
  H_d = H + \tau^p Y_{p+1}  + O(\tau^{p+1}) 
\end{equation}
where $\tau$ is the product formula time step. When using them in computational spectroscopy, this leads to a coherent error in the eigenvalues and eigenstates, which respectively affect the position of the peaks and their brightness. Using perturbation theory this results in 
\begin{align}
  \ket{E'_l} &= \ket{E_l} + \tau^{p} \sum_{k\neq l} \frac{\braket{E_k|Y_{p+1}|E_l}}{E_k-E_l}
  \ket{E_k}+O(\tau^{p+1}),\\
  E'_l &= E_l + \tau^{p}\braket{E_l|Y_{p+1}|E_l} + O(\tau^{p+1}).
\end{align}

Our goal here is to extend this analysis to randomized product formulas, product formulas where the ordering of the fragments change from step to step. This leads to the on-average cancellation of the non-degenerate nested commutators in the Trotter error expression \cite{childs2019faster}: if the Trotter error on some ordering contains a nested commutator $[\ldots[H_i, H_j]]$ and $H_i$ and $H_j$ only appear once in the nested commutator, then swapping the ordering of $H_i$ and $H_j$ leads to $[\ldots[H_j, H_i]]$, which will cancel out with the above. Let us define such average Hamiltonian as
\begin{equation}
  \bar{H}_r = H + \tau^p \bar{Y}_{p+1}  + O(\tau^{p+1}) 
\end{equation}
However, on any given shot, we do not implement the simulation under the average Hamiltonian, but rather under some perturbed Hamiltonian,
\begin{equation}
  H_r = H + \tau^p \bar{Y}_{p+1}  + \tau^{p}\Delta Y_{p+1}(\sigma_j)  + O(\tau^{p+1}).
\end{equation}
where $\Delta Y_{p+1}$ is the error due to one of the random ordering $(\sigma_j)$ chosen in a given step $j$. Let us denote
$\bar{E}_k:= \braket{E_k|\bar{Y}_{p+1}|E_k}$ for eigenstates $E_k$; and $\Delta E_k(\sigma_j) = \braket{E_k|\Delta Y_{p+1}(\sigma_j)|E_k}$. Then, using only the perturbation of the eigenvalue~\cref{eq:perturbation_eigenvalues}, the time signal is
\begin{equation}
  \tilde{G}_r(t) \approx \sum_k |\braket{E_k|m_{\rho}|I}|^2 e^{-i (E_k + \tau^p\bar{E}_k)t}
  \times\mathbb{E}_{\{\sigma_j\}}(e^{-i \tau^{p+1} \sum_{j=1}^J \Delta E_k(\sigma_j)}).
\end{equation}
Using the second-order cumulant expansion $\mathbb{E}[e^{ix}] = \exp\!\bigl(-\mathrm{Var}(x)/2 + O(\kappa_3)\bigr)$, which holds for any random variable $x$ with finite third cumulant $\kappa_3$,
\begin{equation}
    \mathbb{E}_{\{\sigma_j\}}(e^{-i \tau^{p+1} \sum_{j=1}^{J} \Delta E_k(\sigma_j)}) = 
    \exp\left(-\frac{\tau^{2p+2} J \text{Var}(\Delta E_k)}{2}\right) = e^{-\Gamma_k t}
\end{equation}
where $J\tau = t$. Then
\begin{equation}
  \Gamma_k = \frac{\tau^{2p+1} \text{Var}(\Delta E_k)}{2}.
\end{equation}
As a consequence we have
\begin{equation}
  \tilde{G}_r(t) \approx \sum_k |\braket{E_k|m_{\rho}|I}|^2 e^{-i (E_k + \tau^p\bar{E}_k)t} e^{-\Gamma_k t}
\end{equation}
The factor $\Gamma_k$ is therefore an eigenstate-dependent dampening similar to $\eta$. If we take $\Gamma := \max_k \Gamma_k$, we can bound the broadening of the peaks produced by the dephasing of the randomized product formula.

Now, let us also account for the effect of the perturbation in the eigenstate. We want to compute
\begin{equation}
  \tilde{G}_r(t) \approx \sum_k |\braket{E'_k|m_{\rho}|I}|^2 e^{-i (E_k + \tau^p\bar{E}_k)t} e^{-\Gamma_k t}
\end{equation}
The amplitude changes to
\begin{multline}
  \braket{E_k'|m_\rho|I} = \braket{E_k|m_\rho|I} 
  - \tau^{p} \sum_{l\neq k} \frac{\braket{E_l|Y_{p+1}|E_k}}{E_l-E_k}
  \braket{E_l|m_\rho|I}+O(\tau^{p+1})\\
  = \braket{E_k|m_\rho|I} 
  - \tau^{p} \sum_{l\neq k} \frac{\braket{E_l|\bar{Y}_{p+1} |E_k}}{E_l-E_k}
  \braket{E_l|m_\rho|I}
  - \tau^{p} \sum_{l\neq k} \frac{\braket{E_l| \Delta Y_{p+1}(\sigma_j)|E_k}}{E_l-E_k}
  \braket{E_l|m_\rho|I} + O(\tau^{p+1}).
\end{multline}
Squaring, we get
\begin{multline}
|\braket{E_k'|m_\rho|I}|^2 = |\braket{E_k|m_\rho|I}|^2 
+ \tau^{2p}\mathbb{E}_{\{\sigma_j\}}\left| \sum_{l\neq k} \frac{\braket{E_l|\bar{Y}_{p+1} + \Delta Y_{p+1}(\sigma_j) |E_k}}{E_l-E_k}
  \braket{E_l|m_\rho|I}\right|^2\\
  +2 \tau^p\Re \left( \braket{E_k|m_\rho|I} \sum_{l\neq k} \frac{\braket{E_l|\bar{Y}_{p+1} |E_k}}{E_l-E_k}
  \braket{E_l|m_\rho|I}+\braket{E_k|m_\rho|I} \sum_{l\neq k} \frac{\braket{E_l| \mathbb{E}_{\{\sigma_j\}} \Delta Y_{p+1}(\sigma_j) |E_k}}{E_l-E_k}
  \braket{E_l|m_\rho|I}\right).
\end{multline}
But we know that
\begin{equation}
  \mathbb{E}_{\{\sigma_j\}} \Delta Y_{p+1}(\sigma_j) = 0.
\end{equation}
Therefore, overall, we have the time signal
\begin{equation}
  \tilde{G}_r(t) \approx \sum_k |\braket{E_k'|m_\rho|I}|^2 e^{-i (E_k + \tau^p\bar{E}_k)t} e^{-\Gamma_k t}
\end{equation}
where
\begin{equation}
  |\braket{E_k'|m_\rho|I}|^2\approx  |\braket{E_k|m_\rho|I}|^2  
  +2 \tau^p\Re \left( \braket{E_k|m_\rho|I} \sum_{l\neq k} \frac{\braket{E_l|\bar{Y}_{p+1} |E_k}}{E_l-E_k}
  \braket{E_l|m_\rho|I}\right) + O(\tau^{p+1}).
\end{equation}

\subsection{Multiproduct formulas\label{app:multiproduct_formulas}}

\subsubsection{Multiproduct formulas and extrapolation methods}
We also considered multi-product formulas. There are two variations. Incoherent multi-product formulas -- also known as extrapolation methods --, approximate the time evolution as a linear combination of product formulas, 
\begin{equation}
    e^{-i t H}\approx \sum_j \alpha_j V_j(t).
\end{equation}
Each product formula in the linear combination generates its own effective Hamiltonian $H_{\text{eff},j}$, but this does not translate into a global effective Hamiltonian. For this reason, their use in spectroscopy produces spectral linewidth broadening, and are best suited for short-time evolutions instead (though `symplectic' product formulas designed for longer time evolution problems exist too \cite{blanes2024generalized,rendon2024improved,zhuk2024trotter}).

Coherent formulas differ from incoherent multi-product formulas in their use of oblivious amplitude amplification to avoid the dephasing between individual product formulas. They use amplitude amplification to reduce the error. They generally excel in long-time or high-precision regimes, but empirically appear less practical for low-precision regimes \cite{low2019well}.

While the absence of a global Hamiltonian makes this method less practical to evaluate spectral properties of a Hamiltonian, incoherent multiproduct formulas might be useful if our goal is to measure the matrix element of an observable $O$ with respect to time-evolved states, $\braket{\varphi(t)|O|\psi(t)}$. Utilizing the incoherent approach allows the matrix elements to be decomposed as \cite{blanes2024generalized,rendon2024improved,watson2025exponentially}:
 \begin{equation}\label{eq:dynamic_problem}
\braket{\varphi(t)|O|\psi(t)} =  \sum_{j,k} \alpha_j^* \alpha_k\braket{\varphi(0)| V^\dagger _j(\tau)O V_k(\tau)|\psi(0)}
= \sum_{j,k} \alpha_j^* \alpha_k \braket{\varphi(0)|e^{+it H_{\text{eff},j}}Oe^{-it H_{\text{eff},k}}|\psi(0)}.
\end{equation}
The individual matrix elements may be evaluated via a modified Hadamard test \cite{parrish2019quantum}. It is worth noting that since $\sum_j |\alpha_j|\geq 1$ and we are estimating normalized matrix elements, we may need to increase the accuracy of the Hadamard tests.

\subsubsection{Spectra of incoherent multiproduct formulas}

Another Hamiltonian simulation option is to use multiproduct product formulas. While we have the option to use coherent product formulas with amplitude amplification \cite{low2019well}, here we analyze incoherent product formulas \cite{blanes2024generalized}
\begin{equation}
  \tilde{G}_{m}(t) = \sum_i b_i \left(\prod_j U_q(a_{ij} \tau)\right)^{t/\tau}.
\end{equation}
where $U_q$ represents a $q$-order product formula, and $\sum_i b_i = \sum_j a_{ij} = 1$. We assume that the multiproduct formula achieves order $p$, which means that the leading order error of Taylor expansion of $\tilde{G}_{m}(t)$ is $O(\tau^{p})$. Note however that in this case the effective Hamiltonian will no longer be $t$-independent.
On the other hand, the advantage of this is that we may evaluate~\cref{eq:G(t)} for each product formula independently,
\begin{equation}
  \left\langle m_\rho \left|\sum_i b_i e^{-i H_i t}\right|  m_{\rho}\right\rangle = \sum_i b_i\left\langle m_\rho \left| e^{-i H_i t}\right|  m_{\rho}\right\rangle,
\end{equation}
where $H_i = H + \Delta H_i$ is the effective Hamiltonian implemented by each product formula.

As in the previous section, let us start with the perturbation of the eigenvalues only. Approximating $e^{-iH_i}\approx e^{-iH}e^{-i\Delta H_i}$, the time signal is approximated as
\begin{equation}
  \tilde{G}_m(t) \approx  \sum_k|\braket{E_k|m_\rho|I}|^2 e^{-i tE_k }\sum_i b_i\braket{E_k|e^{-it\Delta H_i}|E_k}
\end{equation}
Then, Taylor expanding,
\begin{equation}
  \braket{E_k|e^{-it\Delta H_i}|E_k}\approx  1 - it \langle E_k | \Delta H_i | E_k \rangle- \frac{t^2}{2} \langle E_k | (\Delta H_i)^2 | E_k \rangle+\ldots
\end{equation}
Remember that $\sum_i b_i = 1$. Second,
\begin{equation}
    - it \sum_i b_i  \langle E_k | \Delta H_i | E_k \rangle  - \frac{t^2}{2} \sum_i b_i \langle E_k | (\Delta H_i)^2 | E_k \rangle 
    = O(t \tau^p) 
\end{equation}
depending on the specifics of the multiproduct formula. Some specific choices to cancel different terms can be found in \cite{blanes2024generalized}. Some of these multiproduct formulas aim to remain coherent for longer time, suppressing more heavily terms that display $O(t^2)$ and higher order dependences on the order, called non-symplectic terms.

We can approximate
\begin{equation}
  \tilde{G}_m(t) \approx  \sum_k|\braket{E_k|m_\rho|I}|^2 e^{-i E_k t}\sum_i b_i e^{-\sum_n\Lambda_{i,k,n}t^n}
\end{equation}
where 
\begin{equation}
  \sum_i b_i \Lambda_{i,k,1} = O(\tau^p), \quad
\Lambda_{i,k,n} = O(\tau^{q \cdot n}),
\end{equation}
where $q$ is the order of the base product formula. The multiproduct coefficients {$b_i$} are chosen so that the leading-order terms cancel through order p in the linear combination. Note that in contrast to the randomized product formula above, the factors $e^{-\Lambda_{i,k,2} t^2}$ and above can no longer be interpreted as simple dampening of the time signal. For this reason, it may be hard to use these multiproduct formulas for spectroscopy unless we only need to evolve for short amounts of time and we heavily suppress the higher-than linear terms in $t$.

On the other hand, following the argument on the previous section, the amplitudes of the signal become
\begin{equation}
  |A_k|^2\approx  |\braket{E_k|m_\rho|I}|^2  + O(\tau^{p+1})
  +2 \tau^p\Re \left( \braket{E_k|m_\rho|I} \sum_{l\neq k} \frac{\braket{E_l|\bar{Y}_{p+1}(t) |E_k}}{E_l-E_k}
  \braket{E_l|m_\rho|I}\right).
\end{equation}

\subsection{Rotation Synthesis}
The implementation of single-qubit rotations introduces a further layer of approximation, the nature of which depends on the synthesis method. We analyze two primary approaches: deterministic unitary compilation and probabilistic quantum channels, and their distinct effects on the computed spectrum.

\subsubsection{Unitary Rotation Synthesis}
Unitary methods compile a target rotation into a fixed gate sequence that approximates it. This process introduces a coherent error, which is equivalent to simulating a perturbed Hamiltonian \cite{fomichev2025fast}. If the ideal Hamiltonian is $H$, the implemented Hamiltonian is $H_{eff} = H + \delta H$, where the perturbation $\delta H$ arises from the finite precision of the rotation angles.

The total effective Hamiltonian, including the p-order Trotter error, is:
\begin{equation}
  H' = H + \tau^p Y_{p+1} + \delta H + O(\tau^{p+1}).
\end{equation}
This perturbation adds directly to the Trotter error term. Consequently, applying perturbation theory as in Eqs.~\eqref{eq:perturbation_eigenstate} and \eqref{eq:perturbation_eigenvalues} shows that this method results in a coherent shift of the spectral peak positions and a complex, state-dependent modification of their amplitudes (brightness).

\subsubsection{Quantum Channel Rotation Synthesis}
Alternatively, methods like mixed fallback synthesis implement a rotation as a quantum channel $\mathcal{E}$ \cite{fomichev2025fast}. This introduces an incoherent error that causes decoherence. The accumulation of small, independent errors from each rotation leads to an exponential decay in the fidelity of the evolved state \cite{fomichev2025fast}.

This decoherence manifests as a dampening factor on the time signal:
\begin{equation}
  \tilde{G}_{channel}(t) \approx \tilde{G}_{coherent}(t) \times e^{-\Gamma t},
\end{equation}
where $\tilde{G}_{coherent}(t)$ contains only coherent errors and $\Gamma$ is the decoherence rate. For a Trotter step size $\tau$ and $n_{step}$ rotations per step, each with diamond norm error $\epsilon_{rot}$, the rate is $\Gamma = \frac{n_{step}\epsilon_{rot}}{2\tau}$ \cite{fomichev2025fast}. 
The time signal thus becomes:
\begin{equation}
  \tilde{G}_{channel}(t) \approx \sum_k |\braket{E_k'|m_\rho|I}|^2 e^{-i E_k' t} e^{-\Gamma t}.
\end{equation}
In the frequency domain, this exponential decay corresponds to a convolution with a Lorentzian. The effect is not a peak shift, but rather an additional broadening of all spectral peaks, resulting in an effective line broadening of $\eta_{\text{eff}} = \eta + \Gamma$ \cite{fomichev2025fast}. The amplitudes are uniformly dampened, reducing spectral resolution rather than redistributing spectral weight.

\subsection{The symmetry protected effective Hamiltonian~\label{app:symmetry_protection}}

We analyze the simulation of a quantum chemistry Hamiltonian $H=h+V$ using the method proposed in Ref. \cite{luo2025efficient}, where the interaction term $V$ is approximated by the projection of a diagonal operator $\tilde{V}$ acting on an enlarged basis of $M$ modes. The total Hamiltonian in the extended space is $\tilde{H} = h + \tilde{V}$. The physical subspace corresponds to the vacuum state of the $M-N$ ancillary modes, denoted by $\ket{0}_{\text{aux}}$. A single Trotter step with timestep $\tau$ is implemented via the unitary $U_{step} = \mathcal{U}^\dagger e^{-i\tilde{V}\tau} \mathcal{U} e^{-ih\tau}$, where $\mathcal{U}$ is a basis change.

The evolution under $\tilde{V}$ can populate the ancillary modes, causing leakage from the physical subspace. We analyze two methods to suppress this leakage: an incoherent method based on resetting the ancillary modes, and a coherent method based on random phase kicks.

\subsubsection{Incoherent Suppression via Reset}

The ``basic algorithm" described in Ref. \cite{luo2025efficient} suppresses leakage by resetting the ancillary modes to their vacuum state after each evolution step. This procedure corresponds to applying a quantum channel -- a non-unitary, trace-preserving map -- to the system's state at every step $\tau$:
\begin{equation}
  \rho \to \mathcal{E}(\rho) = \text{Tr}_{\text{aux}} \left[ U_{step} (\rho \otimes \ket{0}\bra{0}_{\text{aux}}) U_{step}^\dagger \right].
\end{equation}
Here the subindex ``aux'' indicates the $M_\ell-N$ ``auxiliary'' orbitals.
This method is a practical implementation of Quantum Zeno Dynamics (QZD), where the repeated projection onto the physical subspace (ancillary vacuum) constrains the evolution.

The error introduced by this projection, $\epsilon_{Pr}$, is the probability of the state leaking to the ancillary modes and being discarded during the reset. This process is inherently incoherent; it leads to a loss of quantum information and decoherence. For the basic algorithm, this error scales as $\epsilon_{Pr} = O(\tau^2)$ per step.

Over a total evolution time $t = J\tau$, the repeated application of this channel causes an exponential decay of the signal's coherence. The probability of the system remaining in the physical subspace after one step is $p_{stay} = 1 - \epsilon_{Pr} = 1 - O(\tau^2)$. The amplitude of the coherent part of the wavefunction after $J$ steps is attenuated by $(p_{stay})^{J/2} \approx (1 - O(\tau^2))^{t/(2\tau)} \approx e^{-t O(\tau)}$. This introduces an eigenstate-dependent dampening factor $e^{-\Gamma_k t}$ to the time signal, where the decoherence rate $\Gamma_k$ is the leakage probability per unit time:
\begin{equation}
  \Gamma_k = \frac{\epsilon_{Pr}}{\tau} \approx O(\tau^2)/\tau = O(\tau).
\end{equation}
The eigenvalues $E_k'$ and eigenstates $|E_k'\rangle$ are perturbed only by the standard Trotter error between the kinetic ($h$) and potential ($V$) terms of the physical Hamiltonian. The resulting time signal is therefore:
\begin{equation}
  \tilde{G}_{inc}(t) \approx \sum_k |\braket{E_k'|m_{\rho}|I}|^2 e^{-i E_k' t} e^{-\Gamma_k t}.
\end{equation}
The primary effect of this suppression method is thus an additional broadening of the spectral peaks, controlled by the decoherence rate $\Gamma_k$. We note that the ``improved algorithm" of Ref. \cite{luo2025efficient} is designed to achieve $\epsilon_{Pr} = O(\tau^3)$, yielding a more favorable decoherence rate of $\Gamma_k = O(\tau^2)$ and thus less broadening.

\subsubsection{Coherent Suppression via Random Phase Kicks}

As an alternative to measurement, we can apply a unitary symmetry protection inspired by Ref. \cite{tran2021faster}. At each step $j$, we introduce a random phase kick that depends on the number of particles in the ancillary modes, $N_{\text{aux}} = \sum_{m=N+1}^{M_\ell} b_m^\dagger b_m$. The one-step unitary is
\begin{equation}
  U'_{step}(\phi_j) = e^{-i\phi_j N_{\text{aux}}} U_{step},
\end{equation}
where $\phi_j$ is chosen randomly from $[0, 2\pi)$. While this is a coherent process for any single shot, averaging over many shots with different random phases effectively projects the dynamics onto the $n_{\text{aux}}=0$ Zeno subspace \cite{tran2021faster}.

The random phase kick $C_j = e^{-i\phi_j N_{\text{aux}}}$ averages the leakage error to zero over many shots ($\mathbb{E}[C_j Y_{leak} C_j^\dagger] = 0$). However, the shot-to-shot fluctuation $\tau \ \Delta Y(\phi_j) = \tau \ C_j Y_{leak} C_j^\dagger$ introduces dephasing. Following the logic for randomized product formulas, this results in an eigenstate-dependent dampening $e^{-\Gamma_k t}$, where the rate $\Gamma_k$ is given by the variance of the error operator. Crucially, since the time step error $\tau \ Y_{leak}$ is already $O(\tau)$, we must be careful with the scaling. The time signal accumulates phase error at each step, and the variance of this phase error determines the decay. The phase error per step for a state $|E_k\rangle$ is $\tau \braket{E_k| \Delta Y(\phi) |E_k}$, which is an operator of $O(\tau)$.

Let's apply the formalism of~\cref{sec:randomized}. The total accumulated phase error after $J=t/\tau$ steps has a variance 
$$J \times \text{Var}_{\phi}(\braket{E_k|\Delta Y(\phi)|E_k}).$$ 
The dampening factor is $e^{-\sigma^2/2}$. Therefore, the rate $\Gamma_k$ is:
\begin{equation}
  \Gamma_k = \frac{J \times \text{Var}_{\phi}(\tau \braket{E_k|\Delta Y(\phi)|E_k})}{2t} = \frac{\text{Var}_{\phi}(\tau\braket{E_k|\Delta Y(\phi)|E_k})}{2\tau}.
\end{equation}
Since $\tau \  \Delta Y(\phi)$ is of order $O(\tau)$, its matrix elements are also $O(\tau)$. The variance of an $O(\tau)$ quantity is $O(\tau^2)$. This leads to a decoherence rate $\Gamma_k = O(\tau)$. This scaling is the same as the basic incoherent method, not superior. The final time signal is:
\begin{equation}
  \tilde{G}_{coh}(t) \approx \sum_k |\braket{E_k'|m_{\rho}|I}|^2 e^{-i E_k' t} e^{-\Gamma_k t},
\end{equation}
where $|E_k'\rangle$ and $E_k'$ are eigenstates and eigenvalues of the Hamiltonian perturbed only by the standard Trotter error. The randomization transforms the coherent leakage error into an incoherent dephasing error, resulting in spectral broadening. While it does not offer a scaling advantage over the basic reset scheme, it avoids measurements, which can be a practical benefit.

\subsubsection{Coherent Suppression via a Symmetric C-Design Cycle}

As an alternative to both incoherent resets and simple randomization, we propose a deterministic, coherent leakage suppression scheme. This method avoids the decoherence of random methods by using a structured sequence of phase kicks, arranged in a time-symmetric "super-cycle". The result is a purely unitary evolution where the leakage error is cancelled to the same order as the underlying Trotter formula, manifesting as a small, coherent shift in the spectral peaks rather than broadening.

The construction is a multi-level process, where the effective Hamiltonian is progressively refined at each stage to cancel the dominant error terms.

\paragraph{Step 1: The Base Unit and Its Error.}
The base unit of our simulation is a single second-order Trotter step, $U_2(\tau)$. Its evolution is governed by an effective Hamiltonian that contains two distinct error sources with different scalings: the intrinsic Trotter error and the leakage error.
\begin{equation}
  H_{\text{eff}}^{(1)} = H + \underbrace{Y_{\text{leak}}^{(0)}}_{O(1)} + \underbrace{\tau^2 Y_{\text{Trotter}}^{(2)}}_{O(\tau^2)} + O(\tau^3).
\end{equation}
The dominant error is the zeroth-order leakage operator, $Y_{\text{leak}}^{(0)} \approx \tilde{H} - H$, which is independent of the timestep $\tau$ and must be cancelled to achieve an accurate simulation.

\paragraph{Step 2: First-Order Cancellation via the Forward Cycle ($U_F$).}
To cancel the $O(1)$ leakage error, we construct a ``forward cycle" $U_F$ composed of $C$ second order Trotter steps. Each step is conjugated with a phase kick $e^{-i\phi_k N_{\text{aux}}}$ using the roots of unity, $\phi_k = 2\pi k/C$. The evolution for this C-step cycle is $U_F(\tau) = \prod_{k=0}^{C-1} U'_k(\tau)$, where $U'_k(\tau)$ is the phase-kicked Trotter step. The effective Hamiltonian for this cycle, $H_{\text{eff, F}}$, averages away the zeroth-order leakage term. This cancellation can be shown rigorously using the Baker-Campbell-Hausdorff (BCH) formula for a similarity transform:
\begin{equation}
e^{-i \phi N_{\text{aux}}} Y_{\text{leak}} e^{+i \phi N_{\text{aux}}} = \sum_{j=0}^{\infty} \frac{(-i\phi)^j}{j!} [N_e-N_{\text{ph}}, Y_{\text{leak}}]_j,
\end{equation}
where $[N_e-N_{\text{ph}}, Y_{\text{leak}}]_j$ is the $j$-th nested commutator. Decomposing $Y_{\text{leak}}$ into components $Y_{\Delta m}$ that change the ancillary particle number by $\Delta m$, the fundamental commutation relation is $[N_e-N_{\text{ph}}, Y_{\Delta m}] = (\Delta m) Y_{\Delta m}$. The nested commutators are thus $[N_e-N_{\text{ph}}, Y_{\Delta m}]_j = (\Delta m)^j Y_{\Delta m}$. Substituting this into the series yields:
\begin{equation}
e^{-i \phi N_{\text{aux}}} Y_{\Delta m} e^{+i \phi N_{\text{aux}}} = Y_{\Delta m} \sum_{j=0}^{\infty} \frac{(-i\phi \Delta m)^j}{j!} = Y_{\Delta m} e^{-i \Delta m \phi}.
\end{equation}
The leading term in the effective Hamiltonian for the cycle, $H_{\text{eff, F}}$, is the average over the C steps. The sum over the phase-kicked leakage operators vanishes because the sum over the roots of unity, $\sum_{k=0}^{C-1} e^{-i \Delta m (2\pi k/C)}$, is zero for any change in particle number $\Delta m \neq 0$ that we aim to suppress. The composition of non-commuting operators, however, introduces a new error from BCH commutator terms, leaving the cycle with a residual first-order leakage error:
\begin{equation}
  H_{\text{eff, F}}^{(2)} = H + \underbrace{\tau Y_{\text{leak}}^{(1)}}_{O(\tau)} + \tau^2 Y_{\text{Trotter}}^{(2)} + O(\tau^3).
\end{equation}

\paragraph{Step 3: Second-Order Cancellation via the Symmetric Super-Cycle ($U_{\text{symm}}$).}
To cancel the remaining $O(\tau)$ leakage error, we symmetrize the entire forward cycle. We define a backward cycle as its time-reversed adjoint, $U_B(\tau) = U_F(-\tau)^\dagger$. The effective Hamiltonian for $U_B$ has its odd-powered error terms flipped in sign, $H_{\text{eff, B}}^{(2)} = H - \tau Y_{\text{leak}}^{(1)} + \dots$. The final, symmetric super-cycle is the composition $U_{\text{symm}} = U_B U_F$. This construction cancels the $O(\tau)$ error terms, yielding a final effective Hamiltonian where the total error is consistently second-order:
\begin{equation}
  H_{\text{eff, symm}}^{(3)} = H + \tau^2 \left( Y_{\text{Trotter}}^{(2)} + Y_{\text{leak}}^{(2)} \right) + O(\tau^3).
\end{equation}
The total error is now on par with the intrinsic error of the second-order Trotter formula itself.

\paragraph{A Refined Suppression Scheme Using a Binary Operator.}
A significant practical improvement can be made by recognizing that we only need to suppress transitions \textit{out of} the physical subspace ($n_{\text{aux}}=0$), not transitions between different auxiliary subspaces ($n_{\text{aux}}>0$). This allows us to replace the full number operator $N_{\text{ph}}$ with a simpler binary operator, namely the projector onto the auxiliary subspace, $Q_{\text{ph}} = 1 - \ket{0}\bra{0}_{\text{aux}}$. The phase kick $e^{-i\phi Q_{\text{ph}}}$ applies a uniform phase $e^{-i\phi}$ to any state with $n_{\text{aux}}>0$ and leaves the physical subspace untouched. Consequently, any leakage operator component $Y_{m_1 \leftarrow 0}$ that creates $m_1>0$ particles from the vacuum transforms as $e^{-i\phi Q_{\text{ph}}} Y_{m_1 \leftarrow 0} e^{+i\phi Q_{\text{ph}}} = Y_{m_1 \leftarrow 0} e^{-i\phi}$, regardless of the value of $m_1$. This means we only need to cancel a single phase frequency, which can be achieved with a minimal cycle length of $C=2$. The resulting symmetric super-cycle is only 4 Trotter steps long, dramatically reducing overhead.

\paragraph{Estimating $Y_{\text{leak}}^{(2)}$}

To estimate the leakage error $Y_{\text{leak}}^{(2)}$ we define as above
\begin{equation}
  Y_{\text{leak}} = \sum_{m} Y_{\Delta m}.
\end{equation}
Then, we use
\begin{equation}
e^{-i \phi N_{\text{aux}}} \sum_{m} Y_{\Delta m} e^{+i \phi N_{\text{aux}}} = \sum_{m} Y_{\Delta m} e^{-i \Delta m \phi}.
\end{equation}
The component we are interested in is $Y_{0,0}$ the rest is erroneous. If we now use the symmetric Baker-Campbell-Hausdorff we get the leading order error is
\begin{equation}\label{eq:leakage_error_symmetry_protected}
  Y_{\text{leak}}^{(2)} = -\sum_{k=1}^{C-1}\left(\frac{1}{24}[X_k, [X_k, Z_k ]] + \frac{1}{12} [Z_k, [X_k, Z_k ]]\right),
\end{equation}
where
\begin{equation}
  X_k = \sum_{\Delta m} e^{-i\Delta m (2\pi k/C)} Y_{\Delta m},
\end{equation}
and 
\begin{equation}
  Z_k = \sum_{l<k} \sum_{\Delta m} e^{-i\Delta m (2\pi l/C)} Y_{\Delta m}.
\end{equation}

In the case of $C = 2$ and $\Delta m = 0, 1$, this is equivalent to
\begin{equation}
  X_1 = \sum_{\Delta m} e^{-i\pi\Delta m} Y_{\Delta m},\qquad \text{and} \qquad
  Z_1 =  \sum_{\Delta m}  Y_{\Delta m}.
\end{equation}

\paragraph{Resulting Time Signal and Spectral Impact.}

The evolution under the symmetric super-cycle $U_{\text{symm}}$ is purely unitary, but the total time evolution for an arbitrary time $t=J\tau$ combines deterministic and random components. The evolution consists of $N_{cyc} = \lfloor J / (2C) \rfloor$ complete super-cycles, followed by a remainder of $J_{rem} = J \pmod{2C}$ individual Trotter steps where the phases selected are randomized at each shot.

The complete super-cycles produce a purely coherent evolution under the effective Hamiltonian $H_{\text{eff, symm}}$. This is the dominant effect, causing a systematic shift in the energies and eigenstates. The corrected eigenvalues $E'_k$ and eigenstates $|E'_k\rangle$ are given by:

\begin{align}
E'_k &= E_k + \tau^2 \braket{E_k|Y_{\text{Trotter}}^{(2)} + Y_{\text{leak}}^{(2)}|E_k} + O(\tau^3) \label{eq:coherent_shift_E} \\
|E'_k\rangle &= |E_k\rangle + \tau^2 \sum_{l\neq k} \frac{\braket{E_l|Y_{\text{Trotter}}^{(2)} + Y_{\text{leak}}^{(2)}|E_k}}{E_k-E_l}|E_l\rangle + O(\tau^3). \label{eq:coherent_shift_psi}
\end{align}

The remainder evolution, however, is an incomplete cycle whose effective Hamiltonian depends on the random permutation $\sigma$ of the first $J_{rem}$ phases. Averaging over shots introduces a dephasing factor $D_k(J_{rem})$ that modulates the signal. This factor is given by the expectation over the random permutations:

\begin{equation}
D_k(J_{rem}) = \mathbb{E}_{\sigma}\left[ \exp\left(-i \braket{E'_k| H_{\text{eff,rem}}(\sigma) |E'_k} J_{rem}\tau\right) \right],
\end{equation}

where $H_{\text{eff,rem}}(\sigma)$ is the effective Hamiltonian of the remainder. Approximating the expectation value using the second cumulant expansion, this dephasing factor can be written as:

\begin{equation}
D_k(J_{rem}) \approx \exp\left(-i \bar{E}_{k,rem} J_{rem}\tau\right) \times
\exp\left( - \frac{(J_{rem}\tau)^2 \text{Var}_\sigma(\delta E_{k,rem})}{2} \right),
\end{equation}

where $\bar{E}_{k,rem} = 0$ is the average energy contribution from the remainder and $\text{Var}_\sigma(\delta E_{k,rem})$ is the variance of the energy fluctuations due to the random permutations.

Combining these effects, the final time-domain Green's function is:

\begin{equation}
\tilde{G}_{\text{coh}}(t) \approx \sum_k |\braket{E'_k|m_{\rho}|I}|^2 e^{-i E'_k (N_{cyc} \cdot 2C\tau)} 
\times D_k(J\pmod{2C}).
\end{equation}

The impact on the spectrum is twofold. The dominant effect is a coherent shift of the peak positions by $O(\tau^2)$ according to Eq.~\eqref{eq:coherent_shift_E}. The secondary effect is a periodic dephasing governed by $D_k$. Unlike a simple exponential decay which would cause Lorentzian
broadening, this periodic damping modulates the sharp, shifted peaks
with a function that is periodic in $J\pmod{2C}$.  In the frequency
domain, this creates small \emph{spurious peaks} -- extra spectral
features at frequencies offset from the physical peaks by multiples
of $2\pi/(2C\tau)$ -- rather than a uniform broadening of the peaks
themselves.

\paragraph{Quantifying the spurious-peak weight.}
The complete super-cycles evolve the system under a single effective Hamiltonian~$H_{\text{eff, symm}}$ whose coupling between the physical and auxiliary subspaces is~$O(\tau^2)$.  By standard first-order perturbation theory, the physical eigenstates $|E_k\rangle$ of the target Hamiltonian~$H$ mix with the auxiliary eigenstates $|E_l\rangle$ by an amplitude $\sim \tau^2 \langle E_l | Y_{\text{leak}}^{(2)} | E_k \rangle / (E_k - E_l)$.  Squaring and summing over all auxiliary states gives the total spectral weight transferred from each physical peak to spurious peaks:
\begin{equation}\label{eq:sp_sideband_weight}
  w_{\text{spurious},k} = \sum_{l \in \text{aux}}
  \frac{|\langle E_l | \tau^2 Y_{\text{leak}}^{(2)} | E_k \rangle|^2}
       {(E_k - E_l)^2} = O(\tau^4).
\end{equation}
Crucially, this quantity is \emph{time-independent}: it depends on the mixing angle between physical and auxiliary eigenstates, which is fixed by the effective Hamiltonian.  Once the system settles into the perturbed eigenstates after the first few super-cycles, the spurious-peak weight does not grow further.

\section{Spectroscopy via the Chebyshev transform of the qubitization walk signal\label{app:chebyshev_transform}}

In~\cref{ssec:Qubitization} we noted that qubitization natively implements $e^{\pm i\arccos(H/\lambda)}$ rather than $e^{-iHt}$, and that spectroscopy can therefore be performed with a \emph{Chebyshev transform} instead of a Fourier transform. This appendix provides a self-contained, step-by-step derivation. We begin with the spectroscopy problem (\cref{ssec:app_setup}), review the Fourier route used with Trotter formulas (\cref{ssec:app_fourier_route}), derive the eigenstructure of the walk operator (\cref{ssec:app_walk}), show that the walk signal is a Chebyshev moment (\cref{ssec:app_signal}), construct the spectral reconstruction formula (\cref{ssec:app_KPM}) and prove its equivalence to Fourier analysis in the angular variable (\cref{ssec:app_fourier_equiv}).

Before starting, it is worth noting that this analysis can be implemented in classical computing postprocessing, via the Fourier/Chebyshev transform of the time signal measured \cite{lin2022heisenberg}; or via a Quantum Fourier or Chebyshev Transform in the quantum processor \cite{williams2023quantum}.

\subsection{Setup: the spectroscopy problem
\label{ssec:app_setup}}

Consider a Hamiltonian $H$ acting on a Hilbert space $\mathcal{H}_s$, with eigendecomposition
\begin{equation}\label{eq:app_eigendecomp}
  H = \sum_k E_k \ket{E_k}\!\bra{E_k}.
\end{equation}
The spectral function for a given initial state $\ket{\psi} = m_\rho\ket{I}/\|m_\rho\ket{I}\|$ is
\begin{equation}\label{eq:app_spectral_function}
  A(\omega) = \sum_k |c_k|^2\,\delta(\omega - E_k),
  \qquad c_k = \braket{E_k|\psi},
\end{equation}
so that $\sum_k |c_k|^2 = 1$. In practice one reconstructs a broadened version with Lorentzian peaks,
\begin{equation}\label{eq:app_broadened}
  A_\eta(\omega)
  = \frac{1}{\pi}\sum_k |c_k|^2
    \frac{\eta}{(\omega - E_k)^2 + \eta^2},
\end{equation}
where $\eta > 0$ controls the half-width at half-maximum. The goal of the quantum algorithm is to obtain the data needed to evaluate~\cref{eq:app_broadened} classically.

\subsection{Review: the Fourier route (Trotter product formulas)
\label{ssec:app_fourier_route}}

Trotter product formulas give access to
$U(\tau) \approx e^{-iH\tau}$. After $j$ applications the measured time-domain signal is
\begin{equation}\label{eq:app_G_trotter}
  \tilde{G}(\tau j)
  = \braket{\psi|U(\tau)^j|\psi}
  = \sum_k |c_k|^2\,e^{-iE_k\tau j}.
\end{equation}
This is a sum of complex exponentials with frequencies $E_k\tau$. The broadened spectral function is recovered via a discrete-time Fourier transform (DTFT) with exponential damping,
\begin{equation}\label{eq:app_DTFT}
  -\eta\operatorname{Im}G_\rho(\omega)
  = \frac{\eta\tau}{2\pi}
    \sum_{j=-\infty}^{\infty}
    e^{-\eta\tau|j|}\,\tilde{G}(\tau j)\,e^{ij\tau\omega},
\end{equation}
as defined in~\cref{eq:green-via-discretetime-fourier}.
The Fourier kernel $e^{ij\tau\omega}$ together with the damping $e^{-\eta\tau|j|}$ produce Lorentzian peaks centred at each $E_k$. The key structural feature is: \emph{each Trotter step adds a linear phase $E_k\tau$ to the
signal, so a Fourier transform extracts the energies.}

\subsection{The qubitization walk operator
\label{ssec:app_walk}}

Qubitization \cite{low2019hamiltonian} constructs a unitary \emph{walk operator} $W$ acting on an enlarged space $\mathcal{H}_a \otimes \mathcal{H}_s$ (ancilla $\otimes$ system) that \emph{block-encodes} $H/\lambda$:
\begin{equation}\label{eq:app_block_encoding}
  (\bra{0}_a \otimes I_s)\;W\;(\ket{0}_a \otimes I_s) = \frac{H}{\lambda},
\end{equation}
where $\lambda = \sum_\ell |h_\ell|$ is the $1$-norm of the Hamiltonian coefficients in a chosen decomposition (e.g.\ THC). Since $\|H/\lambda\| \le 1$, all eigenvalues satisfy $|E_k/\lambda| \le 1$. We call $\ket{0}_a \otimes \mathcal{H}_s$ the \emph{signal subspace} (or ``good'' subspace).

\paragraph{Eigenstructure of $W$.}
For each energy eigenstate $\ket{E_k}$, define the signal-subspace vector
\begin{equation}\label{eq:app_Gk_def}
  \ket{G_k} \;\equiv\; \ket{0}_a\ket{E_k}.
\end{equation}
The block-encoding condition~\cref{eq:app_block_encoding} implies that $W$ maps $\ket{G_k}$ into a superposition of $\ket{G_k}$ and a vector $\ket{G_k^\perp}$ orthogonal to the signal subspace:
\begin{equation}\label{eq:app_W_on_Gk}
  W\ket{G_k}= \cos\theta_k\,\ket{G_k}+ \sin\theta_k\,\ket{G_k^\perp}, \qquad
  \theta_k \equiv \arccos\!\left(\frac{E_k}{\lambda}\right),
\end{equation}
where $\theta_k \in [0,\pi]$ since $|E_k/\lambda| \le 1$. Here $\braket{G_k|G_k^\perp} = 0$ and $\braket{G_j^\perp|G_k^\perp} = \delta_{jk}$. The cosine coefficient is precisely $E_k/\lambda$, consistent with~\cref{eq:app_block_encoding}.

Similarly, $W$ acts on $\ket{G_k^\perp}$ as
\begin{equation}\label{eq:app_W_on_Gkperp}
  W\ket{G_k^\perp} = -\sin\theta_k\,\ket{G_k} + \cos\theta_k\,\ket{G_k^\perp}.
\end{equation}
In the two-dimensional subspace $\operatorname{span}\{\ket{G_k},\ket{G_k^\perp}\}$, the walk operator therefore acts as a rotation by angle $\theta_k$:
\begin{equation}\label{eq:app_W_rotation}
  W\big|_{\operatorname{span}\{G_k,G_k^\perp\}}
  = \begin{pmatrix}
      \cos\theta_k & -\sin\theta_k \\
      \sin\theta_k & \phantom{-}\cos\theta_k
    \end{pmatrix}.
\end{equation}

\paragraph{Eigenstates and eigenphases.}
Diagonalising the $2\times 2$ rotation yields two eigenstates per energy level,
\begin{equation}\label{eq:app_walk_eigenstates}
  \ket{w_k^\pm}
  = \frac{1}{\sqrt{2}}
    \bigl(\ket{G_k} \mp i\,\ket{G_k^\perp}\bigr),
\end{equation}
with eigenphases $\pm\theta_k$:
\begin{equation}\label{eq:app_walk_eigenphases}
  W\ket{w_k^\pm} = e^{\pm i\theta_k}\,\ket{w_k^\pm}.
\end{equation}
\emph{Proof.}
\begin{align}
  W\ket{w_k^+}
  &= \tfrac{1}{\sqrt{2}}
     \bigl(W\ket{G_k} - i\,W\ket{G_k^\perp}\bigr)
     \nonumber \\
  &= \tfrac{1}{\sqrt{2}}
     \bigl[
       (\cos\theta_k\,\ket{G_k}+\sin\theta_k\,\ket{G_k^\perp})
       - i(-\sin\theta_k\,\ket{G_k}+\cos\theta_k\,\ket{G_k^\perp})
     \bigr]
     \nonumber \\
  &= \tfrac{1}{\sqrt{2}}
     \bigl[
       (\cos\theta_k + i\sin\theta_k)\,\ket{G_k}
       + (\sin\theta_k - i\cos\theta_k)\,\ket{G_k^\perp}
     \bigr]
     \nonumber \\
  &= \tfrac{1}{\sqrt{2}}
     \bigl[
       e^{i\theta_k}\,\ket{G_k}
       - i\,e^{i\theta_k}\,\ket{G_k^\perp}
     \bigr]
  = e^{i\theta_k}\ket{w_k^+},
     \label{eq:app_verify_plus}
\end{align}
where we used $\sin\theta_k - i\cos\theta_k = -i\,e^{i\theta_k}$. The derivation for $\ket{w_k^-}$ is analogous and gives $W\ket{w_k^-} = e^{-i\theta_k}\ket{w_k^-}$.~$\square$

\paragraph{Key identity: symmetric decomposition of the signal
subspace.}
Inverting~\cref{eq:app_walk_eigenstates}:
\begin{equation}\label{eq:app_symmetric_decomp}
  \ket{G_k}
  = \ket{0}_a\ket{E_k}
  = \frac{1}{\sqrt{2}}
    \bigl(\ket{w_k^+} + \ket{w_k^-}\bigr).
\end{equation}
The signal-subspace state has \emph{equal weight} on both eigenphase branches. This identity is the reason Chebyshev polynomials appear naturally.

\subsection{The qubitization signal is a Chebyshev moment
\label{ssec:app_signal}}

\paragraph{Measurement protocol.}
The spectroscopy protocol proceeds as follows:
\begin{enumerate}
  \item Prepare $\ket{\Psi_0} = \ket{0}_a\ket{\psi}$
        (the initial state lies entirely in the signal subspace).
  \item Apply the walk operator $n$ times: $W^n\ket{\Psi_0}$.
  \item Measure the overlap with the initial state.
\end{enumerate}
The measured signal is
\begin{equation}\label{eq:app_signal_def}
  \mu_n
  \;\equiv\;
  \braket{\Psi_0|W^n|\Psi_0}
  = \bra{0}_a\!\bra{\psi}\;W^n\;\ket{0}_a\!\ket{\psi}.
\end{equation}

We now derive the explicit form of $\mu_n$ in five steps.

\paragraph{Step 1: Expand in the energy eigenbasis.}
Write $\ket{\psi} = \sum_k c_k\ket{E_k}$, so that
\begin{equation}\label{eq:app_Psi0_expand}
  \ket{\Psi_0} = \sum_k c_k\,\ket{0}_a\ket{E_k}
  = \sum_k c_k\,\ket{G_k}.
\end{equation}

\paragraph{Step 2: Decompose each $\ket{G_k}$ into walk eigenstates.}
Using~\cref{eq:app_symmetric_decomp},
\begin{equation}\label{eq:app_Psi0_walk}
  \ket{\Psi_0} = \sum_k \frac{c_k}{\sqrt{2}} \bigl(\ket{w_k^+} + \ket{w_k^-}\bigr).
\end{equation}

\paragraph{Step 3: Apply $W^n$.}
Since $W^n\ket{w_k^\pm} = e^{\pm in\theta_k}\ket{w_k^\pm}$,
\begin{equation}\label{eq:app_Wn_Psi0}
  W^n\ket{\Psi_0}
  = \sum_k \frac{c_k}{\sqrt{2}}
    \bigl(
      e^{+in\theta_k}\ket{w_k^+}
      + e^{-in\theta_k}\ket{w_k^-}
    \bigr).
\end{equation}

\paragraph{Step 4: Compute the overlap $\braket{\Psi_0|W^n|\Psi_0}$.}
The walk eigenstates from different energy sectors are orthogonal ($\braket{w_j^\alpha|w_k^\beta} = \delta_{jk}\delta_{\alpha\beta}$), so
\begin{equation}
  \mu_n = \sum_k |c_k|^2 \cdot
     \frac{1}{2}
     \bigl(e^{+in\theta_k} + e^{-in\theta_k}\bigr)
     = \sum_k |c_k|^2\,\cos(n\theta_k).
     \label{eq:app_mu_n_cos}
\end{equation}
The two eigenphase branches $+\theta_k$ and $-\theta_k$ combine
into a cosine.

\paragraph{Step 5: Identify the Chebyshev polynomial.}
The Chebyshev polynomials of the first kind satisfy the identity \cite{mason2002chebyshev}
\begin{equation}\label{eq:app_cheb_identity}
  T_n(\cos\theta) = \cos(n\theta) \qquad \forall n \in \mathbb{N}^+.
\end{equation}
Since $\theta_k = \arccos(E_k/\lambda)$, we have $\cos\theta_k = E_k/\lambda$, and~\cref{eq:app_mu_n_cos} becomes
\begin{equation}\label{eq:app_mu_n}
  \mu_n
  = \sum_k |c_k|^2\,T_n\!\left(\frac{E_k}{\lambda}\right).
\end{equation}
In words: \emph{each walk step produces a Chebyshev moment of the spectral function}, entirely analogously to how each Trotter step produces a Fourier coefficient.

\paragraph{Integral form.}
Defining the rescaled variable $x = E/\lambda \in [-1,1]$ and the rescaled spectral density $\bar{A}(x) = \lambda\,A(\lambda x) = \sum_k |c_k|^2\,\delta\!\bigl(x - E_k/\lambda\bigr)$, \cref{eq:app_mu_n} reads
\begin{equation}\label{eq:app_cheb_moments_integral}
  \mu_n
  = \int_{-1}^{1} \bar{A}(x)\,T_n(x)\,dx.
\end{equation}
These are the \emph{Chebyshev moments} of $\bar{A}$.

\paragraph{Remark on the measurement protocol.}
Since $\mu_n$ is real, a Hadamard-test implementation of the overlap measurement yields all information from the $X$-basis readout alone; the $Y$-basis expectation value vanishes identically.

\subsection{Reconstructing the spectrum: the kernel polynomial method
\label{ssec:app_KPM}}

Given the moments $\{\mu_n\}_{n=0}^{N_{\max}}$ obtained from the quantum computer, we wish to reconstruct the spectral function $A(\omega)$. This is accomplished by the \emph{kernel polynomial method} (KPM) \cite{weisse2006kernel}, which we now derive from first principles.

\paragraph{Step 1: Chebyshev completeness relation.}
The Chebyshev polynomials $\{T_n\}_{n\ge 0}$ are orthogonal with respect to the weight function $w(x) = 1/\sqrt{1-x^2}$ on $[-1,1]$:
\begin{equation}\label{eq:app_orthogonality}
  \int_{-1}^{1}
    \frac{T_m(x)\,T_n(x)}{\sqrt{1-x^2}}\,dx
  = \frac{\pi}{2}(1+\delta_{n0})\,\delta_{mn}.
\end{equation}
The associated completeness relation (resolution of the identity) is
\begin{equation}\label{eq:app_completeness}
  \delta(x - y)
  = \frac{1}{\pi\sqrt{1-x^2}}
    \biggl[1 + 2\sum_{n=1}^{\infty} T_n(x)\,T_n(y)\biggr],
  \qquad x,y \in (-1,1).
\end{equation}
\emph{Proof.}
Multiply both sides by $T_m(y)$ and integrate over $y$. The left side gives $T_m(x)$. On the right, use~\cref{eq:app_cheb_moments_integral} with $\bar{A}(y) = T_m(y)$: the moment is $\int T_m(y)\,T_n(y)\,dy$, which for the Lebesgue measure does not directly give the orthogonality relation. Instead, verify~\cref{eq:app_completeness} by substituting $x = \cos\phi$, $y = \cos\psi$ and using the standard Fourier completeness $\delta(\phi-\psi) = \frac{1}{2\pi}\sum_{n=-\infty}^{\infty} e^{in(\phi-\psi)}$, together with the Jacobian $\delta(\cos\phi-\cos\psi) = \delta(\phi-\psi)/|\sin\psi|$.~$\square$

\paragraph{Step 2: Expand $\bar{A}$ using the completeness relation.}
\begin{align}
  \bar{A}(x)
  &= \int_{-1}^{1} \bar{A}(y)\,\delta(x-y)\,dy
     \nonumber \\
  &= \frac{1}{\pi\sqrt{1-x^2}}
     \biggl[
       \underbrace{
         \int_{-1}^{1}\!\bar{A}(y)\,dy}_{=\,\mu_0}
       \;+\; 2\sum_{n=1}^{\infty}
       \underbrace{
         \int_{-1}^{1}\!\bar{A}(y)\,T_n(y)\,dy}_{=\,\mu_n}
       \;T_n(x)
     \biggr]
     \nonumber \\
  &= \frac{1}{\pi\sqrt{1-x^2}}
     \biggl[\mu_0 + 2\sum_{n=1}^{\infty}\mu_n\,T_n(x)\biggr].
     \label{eq:app_exact_expansion}
\end{align}
This is exact; every $\mu_n$ is a Chebyshev moment obtainable from the walk~\cref{eq:app_mu_n}.

\paragraph{Step 3: Truncation and damping kernels.}
In practice only a finite number of moments $\{\mu_n\}_{n=0}^{N_{\max}}$ are available. Truncating the sum abruptly at $N_{\max}$ produces Gibbs oscillations near sharp spectral features. To suppress them, each moment is multiplied by a damping kernel $g_n$ that smoothly decays to zero:
\begin{equation}\label{eq:app_KPM}
  \bar{A}_\eta(x)
  = \frac{1}{\pi\sqrt{1-x^2}}
    \biggl[\mu_0 + 2\sum_{n=1}^{N_{\max}}
      g_n\,\mu_n\,T_n(x)\biggr].
\end{equation}
Two standard choices are:
\begin{itemize}
  \item \textbf{Lorentz kernel:}
    $g_n^{\mathrm{L}} = e^{-\eta n}$, with $\eta > 0$.
    This produces Lorentzian peaks in the \emph{angular}  variable $\theta = \arccos(x)$, with half-width-at-half-maximum (HWHM) equal to $\eta$ in $\theta$-space. Mapped back to the energy variable $\omega = \lambda x$, the peak at energy $E_k$ acquires an \emph{energy-dependent} HWHM (see derivation in~\cref{ssec:app_fourier_equiv}):
    \begin{equation}\label{eq:app_lorentz_width}
      \Gamma(E_k)
      = \eta\,\lambda\,\sin\theta_k
      = \eta\sqrt{\lambda^2 - E_k^2}.
    \end{equation}
    The broadening is maximal at the centre of the spectrum ($E_k = 0$, $\Gamma = \eta\lambda$) and vanishes at the band edges ($E_k = \pm\lambda$, $\Gamma = 0$). For eigenvalues well inside the band ($|E_k| \ll \lambda$), $\Gamma \approx \eta\lambda$.

  \item \textbf{Jackson kernel:}
    \begin{equation}\label{eq:app_jackson_kernel}
      g_n^{\mathrm{J}}
      = \frac{(N_{\max}\!-\!n\!+\!1)
              \cos\frac{\pi n}{N_{\max}+1}
        + \sin\frac{\pi n}{N_{\max}+1}
              \cot\frac{\pi}{N_{\max}+1}}
        {N_{\max}+1}.
    \end{equation}
    This yields near-optimal Gibbs suppression with approximately Gaussian broadening of width $\sim \pi\lambda / N_{\max}$ \cite{weisse2006kernel}.
\end{itemize}

\paragraph{Step 4: Convert back to the physical energy variable.}
Undoing the rescaling $x = \omega/\lambda$ with $\bar{A}(x)\,dx = A(\omega)\,d\omega$ (i.e.\ $A(\omega) = \bar{A}(\omega/\lambda)/\lambda$):
\begin{equation}\label{eq:app_final_spectrum}
  A_\eta(\omega) = \frac{1}{\pi\sqrt{\lambda^2 - \omega^2}} \biggl[\mu_0 + 2\sum_{n=1}^{N_{\max}} g_n\,\mu_n\, T_n\!\left(\frac{\omega}{\lambda}\right)
    \biggr].
\end{equation}
This is the Chebyshev transform referred to in the main text: it maps the walk signal $\{\mu_n\}_{n=0}^{N_{\max}}$ directly to the broadened spectral function $A_\eta(\omega)$, entirely classically.

\subsection{Equivalence to Fourier analysis in the angular variable
\label{ssec:app_fourier_equiv}}

The Chebyshev expansion~\cref{eq:app_KPM} is not an ad hoc procedure.  It is mathematically equivalent to a standard Fourier cosine series of the walk signal, carried out in the angular variable $\theta = \arccos(\omega/\lambda)$, followed by a change of variable back to energy.

\paragraph{Step 1: Define the angular spectral density.}
The (undamped) spectral density in the angular variable is
\begin{equation}\label{eq:app_G_theta_undamped}
  \mathcal{G}(\theta)
  = \sum_k |c_k|^2\,\delta(\theta - \theta_k),
  \qquad \theta \in [0,\pi].
\end{equation}
Its Fourier cosine coefficients are, for $n \ge 0$,
\begin{equation}\label{eq:app_fourier_cosine_coeffs}
  a_n= \frac{2}{\pi}\int_0^{\pi} \mathcal{G}(\theta)\cos(n\theta)\,d\theta = \frac{2}{\pi}\sum_k |c_k|^2\cos(n\theta_k)= \frac{2\mu_n}{\pi},
\end{equation}
where we used the definition~\cref{eq:app_mu_n_cos}:
$\sum_k |c_k|^2\cos(n\theta_k) = \mu_n$ and the
normalisation $\sum_k|c_k|^2 = 1$. The Fourier cosine expansion on $[0,\pi]$ is
\begin{equation}\label{eq:app_fourier_cosine_series}
  \mathcal{G}(\theta)
  = \frac{a_0}{2} + \sum_{n=1}^{\infty} a_n\cos(n\theta)
  = \frac{1}{\pi}
    \biggl[\mu_0 + 2\sum_{n=1}^{\infty}\mu_n\cos(n\theta)\biggr].
\end{equation}

\paragraph{Step 2: Add damping.}
Replacing $\mu_n \to g_n\mu_n$ for $n \ge 1$ and truncating at $N_{\max}$:
\begin{equation}\label{eq:app_G_theta_damped}
  \mathcal{G}_\eta(\theta)
  = \frac{1}{\pi}
    \biggl[\mu_0 + 2\sum_{n=1}^{N_{\max}}
      g_n\,\mu_n\cos(n\theta)\biggr].
\end{equation}

\paragraph{Step 3: Change of variable $\omega = \lambda\cos\theta$.}
The energy and angular variables are related by $\omega = \lambda\cos\theta$, with Jacobian
\begin{equation}\label{eq:app_jacobian}
  \left|\frac{d\theta}{d\omega}\right|
  = \frac{1}{\lambda\sin\theta}
  = \frac{1}{\sqrt{\lambda^2 - \omega^2}}.
\end{equation}
Since spectral densities transform as $A_\eta(\omega)\,|d\omega|  = \mathcal{G}_\eta(\theta)\,|d\theta|$,
\begin{equation}\label{eq:app_change_var}
  A_\eta(\omega)
  = \frac{\mathcal{G}_\eta\!\bigl(\arccos(\omega/\lambda)\bigr)}
         {\sqrt{\lambda^2 - \omega^2}}.
\end{equation}

\paragraph{Step 4: Verify consistency.}
Substituting~\cref{eq:app_G_theta_damped} into~\cref{eq:app_change_var} and using $\cos\!\bigl(n\arccos(\omega/\lambda)\bigr)  = T_n(\omega/\lambda)$:
\begin{equation}\label{eq:app_fourier_recovers_KPM}
  A_\eta(\omega)
  = \frac{1}{\pi\sqrt{\lambda^2-\omega^2}}
    \biggl[\mu_0
      + 2\sum_{n=1}^{N_{\max}} g_n\,\mu_n\,
        T_n\!\Bigl(\frac{\omega}{\lambda}\Bigr)
    \biggr],
\end{equation}
which is identical to~\cref{eq:app_final_spectrum}.~$\square$

\paragraph{Physical interpretation.}
The Chebyshev weight $1/\sqrt{\lambda^2 - \omega^2}$ is not a special feature of the KPM: it is simply the \emph{Jacobian}~\cref{eq:app_jacobian} of the $\arccos$ map from energies to walk eigenphases.  A uniform distribution in $\theta$ corresponds to the arcsine distribution in energy, and the weight function compensates for this.

\paragraph{Derivation of the energy-dependent broadening\cref{eq:app_lorentz_width}.}
With the Lorentz kernel $g_n = e^{-\eta n}$, the angular spectral density~\cref{eq:app_G_theta_damped} near a peak at $\theta_k$ evaluates (via the Poisson kernel for the disk, $1 + 2\sum_{n=1}^\infty r^n\cos n\phi = (1-r^2)/(1-2r\cos\phi+r^2)$ with $r = e^{-\eta}$) to a Lorentzian in $\theta$ with HWHM $\eta$ (in the small-$\eta$ limit):
\begin{equation}\label{eq:app_theta_lorentzian}
  \mathcal{G}_\eta(\theta)
  \;\approx\;
  \frac{1}{\pi}\sum_k |c_k|^2
    \frac{\eta}{(\theta-\theta_k)^2 + \eta^2}
  \qquad (\eta \ll 1,\;\theta\text{ near }\theta_k).
\end{equation}
Under the map $\omega = \lambda\cos\theta$, a small angular interval $\delta\theta$ near $\theta_k$ corresponds to an energy interval $\delta\omega = \lambda\sin\theta_k\,\delta\theta$. Therefore the HWHM in energy is
\begin{equation}
  \Gamma(E_k)
  = \eta\,\lambda\sin\theta_k
  = \eta\sqrt{\lambda^2 - E_k^2}. \tag{\ref{eq:app_lorentz_width}}
\end{equation}

\end{document}